\documentclass[a4paper,11pt]{article}
\usepackage{jheppub}
\usepackage{amsthm}
\usepackage{mathtools}
\usepackage{amssymb}
\usepackage{mathrsfs}
\usepackage{booktabs}
\usepackage{microtype}
\hypersetup{hypertexnames=false}
\title{\boldmath Spherical-DAHA predictions for boundary monopole bubbling of non-minuscule 't Hooft lines in $\mathcal N=4$ SYM}

\author{Ruiliang Li}
\affiliation{Tsinghua University,\\ Beijing, China}
\emailAdd{lirl23@mails.tsinghua.edu.cn}

\abstract{We study non-minuscule boundary 't Hooft lines in four-dimensional $\mathcal N=4$ $U(N)$ super Yang--Mills theory with the regular Nahm-pole boundary condition. For one-row charges we construct the spherical-DAHA operator $\mathbf e_+h_r(Y)\mathbf e_+$ and determine its finite shift expansion and discrete Macdonald spectrum. For the first non-minuscule $U(2)$ coefficient, a local two-minuscule convolution model with Higgs branch $T^*\mathbb P^1$ has two Jeffrey--Kirwan residues, while subtracting the determinant intersection-cohomology summand supplies the additional $-\mathfrak t$ term.  The same transverse $A_1$ calculation, dressed by the spectator-root factor, reproduces the first lower $U(3)$ coefficient.  The two residues are the cross terms of the square of the first Macdonald--Ruijsenaars operator and the subtraction realizes $h_2=e_1^2-e_2$, so that this calculation is a geometric realization of the algebraic coefficient, while its boundary-localization interpretation remains conjectural. We then let the operators act on the evaluation functionals at the points $q^\lambda\mathfrak t^\rho$ with $\rho=(N-1,\ldots,0)$, which for minuscule $\lambda$ coincide, up to a common factor and the order of the entries, with the fugacities of the Higgsing computation of Hatsuda and Okazaki. The zero-shift coefficient at $\mathfrak t^\rho$ of the product of the two oppositely oriented operators equals the normalized Neumann Wilson two-point function for every one-row charge, and for $U(2)$ and $r=2$ this equality requires the $-\mathfrak t$ term.}


\newtheorem{theorem}{Theorem}[section]
\newtheorem{proposition}[theorem]{Proposition}
\newtheorem{lemma}[theorem]{Lemma}
\newtheorem{corollary}[theorem]{Corollary}
\newtheorem{conjecture}[theorem]{Conjecture}
\theoremstyle{definition}

\newcommand{\C}{\mathbb C}

\newcommand{\Neu}{\mathrm{Neu}}
\newcommand{\Nahm}{\mathrm{Nahm}}

\newcommand{\cD}{\mathcal D}
\newcommand{\cL}{\mathcal L}
\newcommand{\cO}{\mathcal O}
\newcommand{\bT}{\mathbb T}
\newcommand{\eps}{\varepsilon}
\newcommand{\one}{\mathbf 1}
\newcommand{\JKRes}{\operatorname{JKRes}}
\newcommand{\ch}{\operatorname{ch}}
\newcommand{\Sym}{\operatorname{Sym}}

\begin{document}
\maketitle
\raggedbottom

\section{Introduction and summary}
\label{sec:introduction-summary}

Boundary line defects provide a protected setting in which electric--magnetic duality can be tested beyond minuscule charge.  A Wilson line ending on a Neumann boundary is represented in the half-index by a character of the boundary holonomy.  Its S-dual is a magnetic line ending on the regular Nahm-pole boundary.  When the magnetic charge is non-minuscule, smooth monopoles may be absorbed by the singular monopole, and the corresponding difference operator contains lower effective magnetic shifts.  Determining their coefficients is the boundary monopole-bubbling problem.

We consider four-dimensional $\mathcal N=4$ $U(N)$ super Yang--Mills theory on a half-space.  The boundary S-duality of Gaiotto and Witten exchanges Neumann and regular Nahm-pole boundary conditions and maps a Wilson line of highest weight $\lambda$ to a boundary 't Hooft line of coweight $\lambda$ \cite{Gaiotto:2008sa,Gaiotto:2008ak}.  The classification of Wilson--'t Hooft operators and its relation to geometric Langlands are reviewed in \cite{Kapustin:2005py,Kapustin:2006pk}.  Hatsuda and Okazaki computed the Neumann correlators of Wilson lines in the fundamental, antisymmetric and symmetric representations, and obtained the regular Nahm-pole correlators of minuscule magnetic charge by Higgsing the Dirichlet half-index \cite{Hatsuda:2025bmn}.  For $N\geq2$ the symmetric powers of degree $r\geq2$ have non-minuscule magnetic duals, for which, as noted there, the Higgsing procedure does not apply directly.  Here the main family is the first non-minuscule tower,
\begin{equation}
  \lambda=(r,0,\ldots,0),\qquad r\geq0.
  \label{eq:intro-one-row-charge}
\end{equation}

The analysis has five parts.  First, the Neumann Wilson-line two-point function is written as a finite-$N$ Macdonald spectral sum.  For the second symmetric power, and for $U(2)$ also for the third, this reproduces the correlators of \cite{Hatsuda:2025bmn}.  Second, a spherical-DAHA operator is constructed and its complete $q$-difference expansion and discrete Macdonald spectrum are determined.  Third, for $U(2)$ with $B=(2,0)\to(1,1)$, the first non-minuscule lower coefficient is realized in the local two-step convolution.  Its two residues are the cross terms of the square of the first Macdonald--Ruijsenaars operator, and the subtraction of the intersection-cohomology summand realizes $h_2=e_1^2-e_2$.  The same transverse slice admits a spectator-root lift to $U(3)$ with $(2,0,0)\to(1,1,0)$, which factorizes residue by residue.  Fourth, the identification of the general DAHA operator with the microscopic regular Nahm-pole localization operator is formulated as a boundary-localization conjecture.  Fifth, the two oppositely oriented operators are made to act on the evaluation functionals at the points $q^\lambda\mathfrak t^\rho$ with $\rho=(N-1,\ldots,0)$, and the zero-shift coefficient of their product at $\mathfrak t^\rho$ is shown to equal the Wilson two-point function for every $r$.

The shift-orbit decomposition of a difference operator is indexed by effective magnetic charges.  Its diagonal expansion in Macdonald polynomials is indexed by spectral channels.  Although both sets of labels can be represented by partitions for the family \eqref{eq:intro-one-row-charge}, they are different decompositions of the operator and are not identified term by term at that level.  They do coincide term by term in the matrix element of Section~\ref{sec:eval-module}.  Starting from the evaluation at $\mathfrak t^\rho$, a shift by a partition $v$ followed by the reverse shift contributes exactly the term labelled by $v$ in the Wilson spectral sum.

\subsection{Electric and DAHA results}
\label{subsec:intro-main-results}

Let $\langle\ ,\ \rangle_{N;q,\mathfrak t}$ be the vacuum-normalized, finite-$N$ Macdonald pairing defined in Section~\ref{sec:boundary-line-dictionary}.  The Neumann two-line half-index is
\begin{equation}
  \widehat{\mathbb I}^{\Neu}_\lambda(q,\mathfrak t)
  =\langle s_\lambda,s_\lambda\rangle_{N;q,\mathfrak t}.
  \label{eq:intro-wilson-product}
\end{equation}
For a one-row representation this has the finite spectral resolution
\begin{equation}
  \boxed{
  \widehat{\mathbb I}^{\Neu}_{(r)}(q,\mathfrak t)
  =\sum_{\substack{\mu\vdash r\\ \ell(\mu)\leq N}}
  Q_\mu\!\left[\frac{1-q}{1-\mathfrak t}\right]^2
  \mathcal N^{(N)}_\mu(q,\mathfrak t) .}
  \label{eq:intro-result-wilson-spectral}
\end{equation}
Here $P_\mu$ is the monic Macdonald polynomial \cite{Macdonald:1995}, $Q_\mu=b_\mu P_\mu$ is its stable dual normalization, and $\mathcal N^{(N)}_\mu$ is the vacuum-normalized finite-$N$ norm.  The distinction between the stable dual factor $b_\mu$ and the finite-$N$ norm is kept throughout.  The case $r=2$ of \eqref{eq:intro-result-wilson-spectral}, and for $U(2)$ also the case $r=3$, are the correlators obtained in \cite{Hatsuda:2025bmn}.  The comparison is made in Section~\ref{sec:wilson-side}.

Let $Y_i$ be the Macdonald-normalized commuting Cherednik operators and let $\mathbf e_+$ be the finite-Hecke spherical idempotent \cite{Cherednik:1995,Cherednik:2005}.  We define the one-row spherical-DAHA operator
\begin{equation}
  \mathsf H^{\mathrm{DAHA}}_{(r)}
  =\mathbf e_+h_r(Y_1,\ldots,Y_N)\mathbf e_+.
  \label{eq:intro-daha-operator}
\end{equation}
If $\mathcal D_k$ denotes the $k$th Macdonald--Ruijsenaars operator, with $\mathcal D_k=0$ for $k>N$, then
\begin{equation}
  \sum_{r\geq0}u^r\mathsf H^{\mathrm{DAHA}}_{(r)}
  =\left(\sum_{k=0}^{N}(-u)^k\mathcal D_k\right)^{-1}.
  \label{eq:intro-daha-generating-function}
\end{equation}
Consequently
\begin{equation}
  \mathsf H^{\mathrm{DAHA}}_{(r)}
  =\sum_{\substack{\gamma\in\mathbb Z_{\geq0}^N\\|\gamma|=r}}
  \mathsf C^{\mathrm{DAHA}}_{r,\gamma}(\mathbf x;q,\mathfrak t)
  T_{q^\gamma}.
  \label{eq:intro-shift-expansion}
\end{equation}
The Weyl orbit of $\gamma$ is a candidate effective magnetic sector; every composition in the orbit is retained in \eqref{eq:intro-shift-expansion}.  The same operator is diagonal on the Macdonald basis,
\begin{equation}
  \mathsf H^{\mathrm{DAHA}}_{(r)}P_\mu
  =h_r(\boldsymbol\xi_\mu)P_\mu,
  \qquad
  \boldsymbol\xi_\mu
  =(q^{\mu_1}\mathfrak t^{N-1},\ldots,q^{\mu_N}).
  \label{eq:intro-discrete-spectrum}
\end{equation}
Equating coefficients in \eqref{eq:intro-daha-generating-function} gives, for $r\geq1$, the recursion
\begin{equation}
\begin{split}
  \mathsf C^{\mathrm{DAHA}}_{r,\gamma}(\mathbf x)
  ={}&\sum_{\substack{\varnothing\neq I\subset\{1,\ldots,N\}\\
  \varepsilon_I\leq\gamma}}
  (-1)^{|I|+1}\mathfrak t^{|I|(|I|-1)/2}
  \prod_{\substack{i\in I\\j\notin I}}
  \frac{\mathfrak t x_i-x_j}{x_i-x_j}
  \\
  &\hspace{2.1cm}\times
  \mathsf C^{\mathrm{DAHA}}_{r-|I|,\gamma-\varepsilon_I}
  (q^{\varepsilon_I}\mathbf x),
\end{split}
  \label{eq:intro-localization-recursion}
\end{equation}
with $\mathsf C^{\mathrm{DAHA}}_{0,0}=1$.  Equation \eqref{eq:intro-localization-recursion} is the lambda-ring identity $H(u)E(-u)=1$ evaluated on the commuting $Y$-operators.  Classical geometric Satake gives an unrefined Grothendieck-group counterpart, in which the same symmetric-function relation is an alternating relation among convolution classes.

\subsection{The first non-minuscule IC-summand subtraction}
\label{subsec:intro-first-projector}

For $U(2)$ and $B=(2,0)\to v=(1,1)$, the local resolved two-minuscule convolution admits a rank-one gauge-theory presentation with two full hypermultiplets.  Its Higgs branch is
\begin{equation}
  X\simeq T^*\mathbb P^1,
\end{equation}
and its positive-chamber Jeffrey--Kirwan (JK) integral \cite{Jeffrey:1993cun,Hori:2014tda} has precisely two poles.  Writing $y=x_1/x_2$ and $F_{\mathfrak t}(z)=(1-\mathfrak t z)/(1-z)$, their residues are
\begin{equation}
  R_{12}=F_{\mathfrak t}(y)F_{\mathfrak t}((qy)^{-1}),
  \qquad
  R_{21}=F_{\mathfrak t}(y^{-1})F_{\mathfrak t}(y/q).
  \label{eq:intro-two-convolution-residues}
\end{equation}
The convolution therefore gives $R_{12}+R_{21}$.  The one-row intersection-cohomology summand is selected from the semismall contraction $T^*\mathbb P^1\to\mathbb C^2/\mathbb Z_2$.  The exceptional summand is removed by a Tate-twisted Koszul class supported on the zero section, whose localized equivariant fixed-point functional is $\mathfrak t$.  Thus
\begin{equation}
  \boxed{
  \mathsf C^{\mathrm{IC}}_{2,(1,1)}
  =R_{12}+R_{21}-\mathfrak t .}
  \label{eq:intro-u2-projector-coefficient}
\end{equation}
In terms of the Macdonald--Ruijsenaars operators, $R_{12}$ and $R_{21}$ are the two cross terms of $\mathcal D_1^2$ and the subtraction realizes $h_2=e_1^2-e_2$, so that \eqref{eq:intro-u2-projector-coefficient} is the lower-shift coefficient of $\mathcal D_1^2-\mathcal D_2$.  The term $-\mathfrak t$ is the contribution from subtracting the determinant IC/Satake summand of the two-fold fundamental convolution, and there is no third JK residue.

For $U(3)$ and the lower shift $(1,1,0)$, the same $A_1$ slice is accompanied by the spectator-root factor
\begin{equation}
 \mathsf S_{12|3}(\mathbf x)=
 \prod_{a=1}^{2}\frac{\mathfrak t x_a-x_3}{x_a-x_3}.
\end{equation}
The two residues and the antisymmetric IC subtraction are all multiplied by
$\mathsf S_{12|3}$, giving
\begin{equation}
 \mathsf C^{\mathrm{IC}}_{2,(1,1,0)}
 =\mathsf S_{12|3}(\mathbf x)
 \bigl(R_{12}+R_{21}-\mathfrak t\bigr).
 \label{eq:intro-u3-projector-factorization}
\end{equation}
Here $(1,1,0)=\omega_2$ is the antisymmetric minuscule channel of $U(3)$.
This spectator lift is the first rank-three check of the transverse-slice
normalization.

The local convolution result is a geometric realization of the first lower coefficient of the algebraic recursion and not a microscopic test of its physical interpretation.  The general identification is stated as a conjecture,
\begin{equation}
  \widehat{\mathcal T}^{\partial,\mathrm{loc}}_{(r)}
  \stackrel{\mathrm{conj.}}{=}
  \mathsf H^{\mathrm{DAHA}}_{(r)}.
  \label{eq:intro-localization-conjecture}
\end{equation}
More generally, for a dominant weight $\lambda$ the algebraic candidate is
\begin{equation}
  \mathsf H^{\mathrm{DAHA}}_\lambda
  =\mathbf e_+s_\lambda(Y)\mathbf e_+,
  \label{eq:intro-general-operator}
\end{equation}
and Conjecture~\ref{conj:general-boundary-Hecke} identifies its shift-orbit coefficients with the regular Nahm-pole localization sectors after the same unbubbled normalization is used.

\subsection{The Wilson two-point function on the evaluation module}
\label{subsec:intro-evaluation-module}

Write $H_r=\mathbf e_+h_r(Y)\mathbf e_+$ and $H_r^\vee=\mathbf e_+h_r(Y^{-1})\mathbf e_+$ for the two orientations of the one-row operator, and recall that $\rho=(N-1,\ldots,0)$, so that $\boldsymbol\xi_0=\mathfrak t^\rho$.  In the polynomial representation the constant function is a joint eigenvector of $H_r$ and $H_r^\vee$, with eigenvalues $h_r(\mathfrak t^\rho)$ and $h_r(\mathfrak t^{-\rho})$.  It therefore cannot represent the regular Nahm-pole boundary without line insertion, and no invertible linear map of the polynomial representation to itself exchanges multiplication by $e_1(X)$ with $\mathcal D_1$ (Section~\ref{sec:eval-module}).  We use instead the evaluation functionals $\mathrm{ev}_\lambda(g)=g(\boldsymbol\xi_\lambda)$ at the points $\boldsymbol\xi_\lambda=q^\lambda\mathfrak t^\rho$ with $\lambda$ dominant, on which a difference operator acts through its coefficients,
\begin{equation}
  \mathrm{ev}_\lambda\cdot C_\gamma(\mathbf x)T_{q^\gamma}
  =C_\gamma(\boldsymbol\xi_\lambda)\,\mathrm{ev}_{\lambda+\gamma}.
  \label{eq:intro-evaluation-action}
\end{equation}
The coefficients for which $\lambda+\gamma$ is not dominant vanish at these points.  The map $P_\lambda/P_\lambda(\mathfrak t^\rho)\mapsto\mathrm{ev}_\lambda$ exchanges multiplication by a symmetric function $p(X)$ with this action of $\mathbf e_+p(Y)\mathbf e_+$.  Proposition~\ref{prop:evaluation-module} then gives, for every $r\geq0$,
\begin{equation}
  \bigl[H_rH_r^\vee\bigr]_{T^0}(\mathfrak t^\rho)
  =\sum_{\substack{v\vdash r\\ \ell(v)\leq N}}
  \mathsf C^{\mathrm{DAHA}}_{r,v}(\mathfrak t^\rho)\,
  \mathsf C^{\vee}_{r,v}(\boldsymbol\xi_v)
  =\widehat{\mathbb I}^{\Neu}_{(r)}(q,\mathfrak t).
  \label{eq:intro-evaluation-two-point}
\end{equation}
Here $[\,\cdot\,]_{T^0}$ is the coefficient of the identity shift, $\mathsf C^\vee_{r,\gamma}$ are the shift coefficients of $H_r^\vee$, and the term labelled by $v$ is the summand $\mu=v$ of \eqref{eq:intro-result-wilson-spectral}.  For $U(2)$ and $r=2$ the term $v=(1,1)$ equals $[(\mathfrak t-q)/(1-q\mathfrak t)]^2$.  It is produced by the coefficient \eqref{eq:intro-u2-projector-coefficient} at $\mathbf x=\mathfrak t^\rho$, and the subtraction $-\mathfrak t$ is needed to obtain it.

Equation \eqref{eq:intro-evaluation-two-point} is an identity for the explicit difference operators.  For minuscule charge the points $\boldsymbol\xi_\lambda$ are, up to a common factor and the order of the entries, the fugacities at which the Dirichlet half-index is Higgsed in \cite{Hatsuda:2025bmn}, and for the fundamental charge \eqref{eq:intro-evaluation-two-point} agrees with the normalized regular Nahm-pole two-point function obtained there.  For non-minuscule charge the identification of the left-hand side with a regular Nahm-pole two-point function is not derived in this paper.

For $U(N)$ a common determinant shift is separated by
\begin{equation}
  m=B_N,\qquad \bar B=B-m\mathbf1,\qquad \bar B_N=0.
\end{equation}
The corresponding determinant translations carried by an oriented line and its dual are inverse and cancel in the product of the two operators.  Hence, for $U(2)$, all dominant charges reduce to one-row charges in \eqref{eq:intro-evaluation-two-point}.

\subsection{Relation to earlier work and organization}
\label{subsec:intro-relation-earlier-work}

The boundary conditions and their S-duality are those of \cite{Gaiotto:2008sa,Gaiotto:2008ak}, and their half-indices were studied in \cite{Gaiotto:2019jvo}.  Regular Nahm poles also enter the gauge-theoretic construction of knot invariants \cite{Witten:2011zz,Gaiotto:2011nm}.  Magnetic line defects as difference operators are central to the bulk localization and integrable-system literature \cite{Ito:2012aa,Gomis:2011pf,Gang:2012yr,Maruyoshi:2021bwe,Costello:2024cgj}.  Bulk bubbling quantum mechanics and its possible non-compact contributions were analyzed in \cite{Brennan:2018yuj,Assel:2020bubbling}.  The low-rank model used here is instead the local two-minuscule convolution.  Its semismall decomposition realizes the representation-theoretic direct-sum decomposition expected from geometric Satake \cite{Mirkovic:2007}.  Affine-Grassmannian slices and handsaw spaces provide related geometric structures \cite{Nakajima:2012handsaw,Braverman:2016wma}.

The Wilson-line correlators of Section~\ref{sec:wilson-side} overlap with Ref.~\cite{Hatsuda:2025bmn}.  That work computes the Neumann two-point functions of Wilson lines in the fundamental, antisymmetric and symmetric representations.  For the symmetric representations it gives the second and third symmetric powers for $U(2)$, proposes a closed form for arbitrary symmetric rank in $U(2)$, and gives the second symmetric power for general $N$ from the expansion of $h_2$ in Macdonald polynomials.  Our formulae \eqref{eq:U2-r2-wilson}, \eqref{eq:U2-r3-wilson} and \eqref{eq:U3-r2-wilson} reproduce these results, as indicated where they appear, and \eqref{eq:intro-result-wilson-spectral} writes the one-row correlator for all $r$ and $N$ in the same way.  On the magnetic side, Ref.~\cite{Hatsuda:2025bmn} obtains the regular Nahm-pole correlators of minuscule charge by Higgsing and notes that this procedure does not extend directly to the non-minuscule duals of the symmetric representations.  The results of the present paper which are not contained in that work concern these non-minuscule charges.  They are the spherical-DAHA candidate with its complete shift coefficients, the local convolution realization of the first lower coefficient, and the identity \eqref{eq:intro-evaluation-two-point}, which relates the shift coefficients at the points $q^v\mathfrak t^\rho$ to the known Wilson correlators.  Line operators at half-BPS interfaces are studied in \cite{Hatsuda:2025interface,Hatsuda:2026orthosymplectic}.

Section~\ref{sec:boundary-line-dictionary} fixes the Neumann pairing, dual-line convention, determinant reduction, the distinction between pure bubbling factors and complete shift coefficients, and the relation of the evaluation points to the Higgsing fugacities of \cite{Hatsuda:2025bmn}.  Section~\ref{sec:wilson-side} derives the Wilson two-point function and compares it with \cite{Hatsuda:2025bmn}.  Section~\ref{sec:boundary-bubbling} presents the algebraic lambda-ring recursion, the $U(2)$ transverse-slice calculation, and its $U(3)$ spectator lift.  Section~\ref{sec:daha-kernel} establishes the spherical-DAHA construction, its discrete Macdonald spectrum, and the evaluation-module identity for the Wilson two-point function.  Section~\ref{sec:examples-limits-checks} gives low-rank checks and special limits.  Section~\ref{sec:hecke-outlook} states the localization conjecture and its geometric motivation.  The appendices collect conventions, low-rank algebraic coefficients, expansion data, and the proof of Proposition~\ref{prop:evaluation-module}.

\section{Boundary line half-indices and the S-duality dictionary}
\label{sec:boundary-line-dictionary}

This section fixes the normalization of the Neumann pairing and separates three pieces of data which enter the magnetic calculation: the pure bubbling factor, the unbubbled effective-monopole operator, and their complete coefficient in the $q$-difference expansion.  We also specify the orientation of the two line insertions and the treatment of the determinant direction of $U(N)$, and relate the evaluation points of Section~\ref{sec:eval-module} to the Higgsing fugacities of Ref.~\cite{Hatsuda:2025bmn}.

Throughout, $|q|<1$ in the analytic interpretation, while all formulae may alternatively be read as identities in a completed Laurent ring.  The Macdonald parameter $\mathfrak t$ is fixed by the Neumann vector-multiplet determinant
\begin{equation}
  \Delta_N(\mathbf x;q,\mathfrak t)
  =\prod_{i\neq j}
  \frac{(x_i/x_j;q)_\infty}{(\mathfrak t x_i/x_j;q)_\infty}.
  \label{eq:section2-macdonald-density}
\end{equation}
Let $u_R$ denote the fugacity conjugate to $H-C$, denoted by $t$ in
Ref.~\cite{Hatsuda:2025bmn}.  In the JHEP version of record, Eq.~(2.20)
contains the weight $u_R^{H-C}$ and Eq.~(2.27) defines the Macdonald parameter
$\mathfrak t=q^{1/2}u_R^{-2}$.  The general S-duality comparison stated there
inverts $u_R$; Eq.~(3.11) displays this map for the $U(2)$ fundamental-line
correlator.  Therefore the Macdonald coordinate on the Nahm side is
\begin{equation}
  \mathfrak t_{\Nahm}:=q^{1/2}u_R^2=\frac q{\mathfrak t},
  \qquad
  (q,\mathfrak t)\longmapsto
  \left(q,\frac q{\mathfrak t}\right).
  \label{eq:physical-parameter-map-section2}
\end{equation}
The transformation in \eqref{eq:physical-parameter-map-section2} compares the physical boundary parameters.  All DAHA formulae in this paper are written in the electric parameter $\mathfrak t$, in which the rank-two IC/Satake subtraction found below is $-\mathfrak t$.
Section~\ref{subsec:higgsing-points} explains why the electric parameter is also
the one in which the evaluation points of Section~\ref{sec:eval-module} are written.

\subsection{The Neumann pairing and line orientation}
\label{subsec:half-index-geometry}

We work on $HS^3\times S^1$, whose boundary is $S^2\times S^1$.  The half-BPS boundary condition preserves three-dimensional $\mathcal N=4$ supersymmetry.  For a boundary condition $\mathcal B$ and a line wrapping $S^1$, the protected trace has the schematic form
\begin{equation}
  \mathbb I_{\mathcal B,\mathcal L}
  =\operatorname{Tr}_{\mathcal H_{\mathcal B,\mathcal L}}
  (-1)^F q^{\mathcal J}\mathfrak t^{\mathcal R}.
  \label{eq:abstract-half-index}
\end{equation}
Only the induced meromorphic pairing will be needed below.

Let $\mathscr P_N^{S_N}$ be the symmetric Laurent-polynomial representation.
For the Neumann boundary the vector-multiplet determinant
\eqref{eq:section2-macdonald-density} and the gauge projection give the
vacuum-normalized bilinear Macdonald pairing \cite{Hatsuda:2025bmn}
\begin{equation}
  \langle f,g\rangle_{\Neu;q,\mathfrak t}
  \equiv\langle f,g\rangle_{N;q,\mathfrak t}
  =
  \frac{1}{\mathcal Z_N(q,\mathfrak t)}
  \frac1{N!}\operatorname{CT}_{\mathbf x}
  \left[
  \Delta_N(\mathbf x;q,\mathfrak t)
  f(\mathbf x)g(\mathbf x^{-1})
  \right],
  \qquad
  \langle1,1\rangle_{\Neu}=1,
  \label{eq:section2-oriented-neumann-pairing}
\end{equation}
where
\begin{equation}
  \mathcal Z_N(q,\mathfrak t)
  =\frac1{N!}\operatorname{CT}_{\mathbf x}
  \Delta_N(\mathbf x;q,\mathfrak t).
  \label{eq:section2-empty-neumann}
\end{equation}
The Neumann boundary without line insertion corresponds to the constant
function $1$, and $\langle1,1\rangle_{\Neu;q,\mathfrak t}=1$.
We suppress the parameters on the boundary subscript when they are clear.
This pairing is bilinear and meromorphic, without complex conjugation.  In a real chamber such as $0<q,\mathfrak t<1$ it admits the usual analytic contour interpretation.

For an oriented line operator $\mathcal L$, let $\mathcal L^\vee$ denote the oppositely oriented dual line.  On the Neumann side the two orientations of a Wilson line insert a character and its dual, and the normalized two-line observable is the pairing \eqref{eq:section2-oriented-neumann-pairing} of the character with itself.  The inversion of $\mathbf x$ in the second argument records the orientation of the second line.  We do not assume a pairing of this form for the regular Nahm-pole boundary.  In the polynomial representation of Section~\ref{sec:daha-kernel} the constant function is a joint eigenvector of the magnetic difference operators and cannot play the role of the regular Nahm-pole boundary without line insertion.  The two-line quantity considered on the magnetic side is the matrix element on the evaluation module defined in Section~\ref{sec:eval-module}.

\subsection{Neumann Wilson lines}
\label{subsec:neumann-boundary-lines}

A Wilson line is labelled by a dominant integral $U(N)$ weight
\begin{equation}
  \lambda=(\lambda_1,\ldots,\lambda_N),
  \qquad
  \lambda_1\geq\cdots\geq\lambda_N.
  \label{eq:dominant-weight-section2}
\end{equation}
Its character is the Laurent Schur function $s_\lambda(\mathbf x)$.  The oppositely oriented dual representation has highest weight
\begin{equation}
  \lambda^\vee=(-\lambda_N,\ldots,-\lambda_1),
  \qquad
  s_{\lambda^\vee}(\mathbf x)=s_\lambda(\mathbf x^{-1}).
  \label{eq:dual-representation-section2}
\end{equation}
Thus
\begin{equation}
  \widehat{\mathbb I}^{\Neu}_\lambda(q,\mathfrak t)
  =
  \langle s_\lambda,s_\lambda\rangle_{N;q,\mathfrak t}.
  \label{eq:section2-normalized-neumann}
\end{equation}
In operator notation the two orientations are multiplication by $s_\lambda(X)$ and $s_{\lambda^\vee}(X)$.

If $\mathbf1=(1,\ldots,1)$ and $m\in\mathbb Z$, then
\begin{equation}
  s_{\lambda+m\mathbf1}(\mathbf x)
  =(x_1\cdots x_N)^m s_\lambda(\mathbf x).
  \label{eq:center-character-section2}
\end{equation}
The determinant monomial and its inverse cancel between the two orientations in \eqref{eq:section2-oriented-neumann-pairing}.  Hence
\begin{equation}
  \widehat{\mathbb I}^{\Neu}_{\lambda+m\mathbf1}
  =\widehat{\mathbb I}^{\Neu}_\lambda.
  \label{eq:section2-electric-determinant-invariance}
\end{equation}

\subsection{Regular Nahm-pole lines and screened shifts}
\label{subsec:regular-nahm-lines}

The magnetic boundary condition is the regular Nahm pole.  If $y$ is the coordinate normal to the boundary, three adjoint scalars obey Nahm's equations and have the principal singularity
\begin{equation}
  X_i(y)=\frac{\rho(t_i)}{y}+O(1),
  \qquad i=1,2,3,
  \label{eq:regular-nahm-pole-section2}
\end{equation}
where $\rho:\mathfrak{su}(2)\to\mathfrak{u}(N)$ is the principal embedding \cite{Gaiotto:2008sa,Gaiotto:2008ak}.  A boundary 't Hooft line is labelled by a dominant coweight
\begin{equation}
  B=(B_1,\ldots,B_N),
  \qquad B_1\geq\cdots\geq B_N.
  \label{eq:magnetic-charge-section2}
\end{equation}
For $U(N)$ we identify weights and coweights in the standard basis.

For non-minuscule $B$, effective magnetic charges satisfy
\begin{equation}
  v\preceq B,\qquad B-v\in Q^\vee_+.
  \label{eq:screening-order-section2}
\end{equation}
The full localization operator may be organized by dominant effective charges,
\begin{equation}
  \widehat{\mathcal T}^{\partial,\mathrm{loc}}_B
  =
  \sum_{v\preceq B}
  \mathsf Z^\partial_{B,v}\,
  \mathcal U^\partial_v,
  \qquad
  \mathsf Z^\partial_{B,B}=1.
  \label{eq:section2-magnetic-operator-bubbling}
\end{equation}
Here $\mathsf Z^\partial_{B,v}$ is the \emph{pure bubbling index}, while $\mathcal U^\partial_v$ is the complete unbubbled operator of effective charge $v$.  The top normalization in \eqref{eq:section2-magnetic-operator-bubbling} concerns only $\mathsf Z^\partial_{B,B}$.

In a composition basis, the same operator has the form
\begin{equation}
  \widehat{\mathcal T}^{\partial,\mathrm{loc}}_B
  =
  \sum_{v\preceq B}\ \sum_{\gamma^+=v}
  \mathsf C^\partial_{B,\gamma}(\mathbf x;q,\mathfrak t)
  T_{q^\gamma}.
  \label{eq:section2-composition-expansion}
\end{equation}
The complete coefficient $\mathsf C^\partial_{B,\gamma}$ contains both the pure bubbling factor and the non-trivial coefficient of the unbubbled effective-monopole shift.  Schematically, at the composition level,
\begin{equation}
  \mathsf C^\partial_{B,\gamma}
  =
  \mathsf Z^\partial_{B,v;\gamma}\,
  \mathsf U^\partial_{v,\gamma},
  \qquad \gamma^+=v,
  \label{eq:section2-pure-versus-full}
\end{equation}
with any sum over internal localization data understood on the right-hand side.  In particular,
\begin{equation}
  \mathsf Z^\partial_{B,B}=1
  \quad\not\Longrightarrow\quad
  \mathsf C^\partial_{B,\gamma}=1
  \quad(\gamma^+=B).
  \label{eq:section2-top-factor-distinction}
\end{equation}
The top full coefficient contains the usual product of rational one-loop factors.  This distinction will be maintained in the low-rank examples.

\subsection{Determinant reduction}
\label{subsec:section2-determinant-reduction}

For a dominant $U(N)$ magnetic charge define
\begin{equation}
  m=B_N,\qquad
  \bar B=B-m\mathbf1,\qquad
  \bar B_N=0.
  \label{eq:section2-determinant-reduction}
\end{equation}
The representative $\bar B$ will be called determinant-reduced.  In a boundary SQM based on simple-root screening the framing is assigned to $\bar B$, and a candidate dimension prescription is
\begin{equation}
  \dim W_i=\bar B_i-\bar B_{i+1},
  \qquad i=1,\ldots,N-1,
  \qquad W_N=0,
  \label{eq:section2-reduced-framing-dimensions}
\end{equation}
with the complete torus character of each $W_i$ specified separately whenever a localization model is used.

In the Macdonald normalization the determinant translation is
\begin{equation}
  \mathsf C_N
  =Y_1\cdots Y_N
  =\mathcal D_N
  =\mathfrak t^{N(N-1)/2}
  T_{q,x_1}\cdots T_{q,x_N}.
  \label{eq:section2-determinant-translation}
\end{equation}
It commutes with the symmetric $Y$-subalgebra, but it is not central in the full DAHA because it has non-trivial $q$-commutation with the $X_i$.  The spherical-DAHA operators of the two line orientations carry the inverse translations $\mathsf C_N^{m}$ and $\mathsf C_N^{-m}$, see \eqref{eq:general-determinant-factorization}, and Conjecture~\ref{conj:general-boundary-Hecke} assigns the same translations to the localization operators,
\begin{equation}
  \widehat{\mathcal T}^{\partial,\mathrm{loc}}_B
  =\mathsf C_N^m\widehat{\mathcal T}^{\partial,\mathrm{loc}}_{\bar B},
  \qquad
  \widehat{\mathcal T}^{\partial,\mathrm{loc}}_{B^\vee}
  =\mathsf C_N^{-m}\widehat{\mathcal T}^{\partial,\mathrm{loc}}_{\bar B^\vee}.
  \label{eq:section2-oriented-determinant-lines}
\end{equation}
They therefore cancel in the ordered two-line insertion.  For $U(2)$,
\begin{equation}
  (a,b)=b(1,1)+(a-b,0),
  \qquad a\geq b,
  \label{eq:section2-U2-determinant-decomp}
\end{equation}
so every two-line quantity considered below reduces to the one-row charge $(a-b,0)$.

\subsection{Evaluation points and the Higgsing fugacities}
\label{subsec:higgsing-points}

For minuscule magnetic charge, Ref.~\cite{Hatsuda:2025bmn} obtains the regular Nahm-pole correlators from the Dirichlet half-index by a specialization of its flavor fugacities.  This specialization is related as follows to the points
\begin{equation}
  \boldsymbol\xi_\lambda
  =(q^{\lambda_1}\mathfrak t^{N-1},q^{\lambda_2}\mathfrak t^{N-2},\ldots,q^{\lambda_N})
  \label{eq:section2-evaluation-points}
\end{equation}
which appear in the Macdonald spectrum \eqref{eq:intro-discrete-spectrum} and in Section~\ref{sec:eval-module}.  The Nahm-pole correlators of Ref.~\cite{Hatsuda:2025bmn} are written in a fugacity which we denote by $u$ and which the S-duality map of that reference identifies with $u_R^{-1}$.  Put $c=q^{1/4}u$, so that $c^2=q^{1/2}u_R^{-2}=\mathfrak t$.  The specialization (3.37) of Ref.~\cite{Hatsuda:2025bmn} for the boundary without line is
\begin{equation}
  x_i=q^{(2i-1)/4}u^{2i-1}=c\,\mathfrak t^{\,i-1},
  \qquad 1\leq i\leq N.
  \label{eq:section2-higgsing-vacuum}
\end{equation}
For the charge $B=(1^k,0^{N-k})$ the last $k$ fugacities are multiplied by $q$, see Eqs.~(3.40) and (3.48) there.  If $w_0$ denotes the reversal of the order of the entries, the specialized fugacities are therefore $c\,w_0\boldsymbol\xi_B$.  The Dirichlet half-index is symmetric in the $x_i$ and depends only on their ratios, so that neither $w_0$ nor the scale $c$ plays a role.  The coefficients of the difference operators are likewise functions of the ratios $x_i/x_j$, and the zero-shift coefficients considered in Section~\ref{sec:eval-module} do not depend on $c$.  The ratio of successive fugacities in \eqref{eq:section2-higgsing-vacuum} is $q^{1/2}u^2=\mathfrak t$.  This is the reason why the electric parameter $\mathfrak t$, and not $\mathfrak t_{\Nahm}$, appears in the points \eqref{eq:section2-evaluation-points} and in the difference operators.  For these charges the correlators (3.42) and (3.50) of Ref.~\cite{Hatsuda:2025bmn}, divided by the half-index (3.35), are in these variables
\begin{equation}
  \prod_{m=1}^{k}
  \frac{(1-q\mathfrak t^{m-1})(1-\mathfrak t^{N-k+m})}
       {(1-\mathfrak t^{m})(1-q\mathfrak t^{N-k+m-1})}
  =\mathcal N^{(N)}_{(1^k,0^{N-k})}(q,\mathfrak t),
  \label{eq:section2-higgsing-fundamental}
\end{equation}
where $\mathcal N^{(N)}_\mu$ is the normalized Macdonald norm of Section~\ref{sec:wilson-side} and $k=1$ is the fundamental charge.  The right-hand side is the normalized two-point function of the Wilson line in the $k$th antisymmetric representation, in accordance with the S-duality check of Ref.~\cite{Hatsuda:2025bmn}.  For non-minuscule charge that reference notes that the Higgsing procedure does not apply directly.  Section~\ref{sec:eval-module} uses the points \eqref{eq:section2-evaluation-points} for all dominant $\lambda$, together with the complete shift expansion of the spherical operator.

\section{The Wilson side}
\label{sec:wilson-side}

This section computes the Wilson-line two-point functions with the Neumann boundary condition.  A line in the irreducible $U(N)$ representation of highest weight $\lambda$ inserts the Laurent Schur character $s_\lambda(\mathbf x)$.  The boundary gauge projection then gives the finite-$N$ Macdonald pairing fixed in Section~\ref{sec:boundary-line-dictionary}.  This is the setting of Hatsuda and Okazaki \cite{Hatsuda:2025bmn}, whose results for the symmetric representations are compared with the formulae below where they appear.

\subsection{The normalized finite-\texorpdfstring{$N$}{N} pairing}
\label{subsec:macdonald-product}

Let
\begin{equation}
  (z;q)_\infty=\prod_{n=0}^{\infty}(1-zq^n)
\end{equation}
and recall the density
\begin{equation}
  \Delta_N(\mathbf x;q,\mathfrak t)
  =\prod_{i\neq j}
  \frac{(x_i/x_j;q)_\infty}{(\mathfrak t x_i/x_j;q)_\infty}.
  \label{eq:macdonald-density}
\end{equation}
The vacuum factor and the normalized pairing are
\begin{align}
  \mathcal Z_N(q,\mathfrak t)
  &=
  \frac1{N!}\operatorname{CT}_{\mathbf x}
  \Delta_N(\mathbf x;q,\mathfrak t),
  \label{eq:macdonald-vacuum}
  \\
  \langle f,g\rangle_{N;q,\mathfrak t}
  &=
  \frac{1}{\mathcal Z_N(q,\mathfrak t)}
  \frac1{N!}\operatorname{CT}_{\mathbf x}
  \left[
  \Delta_N(\mathbf x;q,\mathfrak t)
  f(\mathbf x)g(\mathbf x^{-1})
  \right].
  \label{eq:normalized-macdonald-product}
\end{align}
Thus $\langle1,1\rangle_{N;q,\mathfrak t}=1$.  The constant term may be represented by the standard torus contour in an analytic chamber.  For generic meromorphic parameters \eqref{eq:normalized-macdonald-product} is a bilinear pairing.

The dual weight is
\begin{equation}
  \lambda^\vee=(-\lambda_N,\ldots,-\lambda_1),
  \qquad
  s_{\lambda^\vee}(\mathbf x)=s_\lambda(\mathbf x^{-1}).
  \label{eq:dual-weight}
\end{equation}
The normalized Wilson two-line half-index is therefore
\begin{equation}
  \boxed{
  \widehat{\mathbb I}^{\Neu}_\lambda(q,\mathfrak t)
  =\langle s_\lambda,s_\lambda\rangle_{N;q,\mathfrak t}.}
  \label{eq:wilson-macdonald-product}
\end{equation}
If $\mathbf1=(1,\ldots,1)$, multiplication by the determinant monomial gives
\begin{equation}
  s_{\lambda+m\mathbf1}(\mathbf x)
  =(x_1\cdots x_N)^m s_\lambda(\mathbf x).
\end{equation}
The inverse monomial occurs in the dual insertion, and hence
\begin{equation}
  \widehat{\mathbb I}^{\Neu}_{\lambda+m\mathbf1}
  =\widehat{\mathbb I}^{\Neu}_\lambda.
  \label{eq:determinant-invariance}
\end{equation}

\subsection{Macdonald diagonalization}
\label{subsec:diagonalization}

Let $P_\mu(\mathbf x;q,\mathfrak t)$ be the monic Macdonald polynomial in $N$ variables \cite[Chapter~VI]{Macdonald:1995}, indexed by a partition $\mu$ with $\ell(\mu)\leq N$.  Laurent weights are included through
\begin{equation}
  P_{\mu+m\mathbf1}(\mathbf x;q,\mathfrak t)
  =(x_1\cdots x_N)^mP_\mu(\mathbf x;q,\mathfrak t).
  \label{eq:determinant-shift-P}
\end{equation}
The finite-$N$ orthogonality relation is
\begin{equation}
  \langle P_\mu,P_\nu\rangle_{N;q,\mathfrak t}
  =\delta_{\mu\nu}\mathcal N_\mu^{(N)}(q,\mathfrak t).
  \label{eq:macdonald-orthogonality}
\end{equation}
With the vacuum normalization of \eqref{eq:normalized-macdonald-product},
\begin{equation}
  \mathcal N_\mu^{(N)}
  =\frac{D_N(\mu;q,\mathfrak t)}{D_N(0;q,\mathfrak t)},
  \label{eq:normalized-norm}
\end{equation}
where
\begin{equation}
  D_N(\mu;q,\mathfrak t)
  =
  \prod_{1\leq i<j\leq N}
  \frac{(q^{\mu_i-\mu_j}\mathfrak t^{j-i};q)_\infty
        (q^{\mu_i-\mu_j+1}\mathfrak t^{j-i};q)_\infty}
       {(q^{\mu_i-\mu_j}\mathfrak t^{j-i+1};q)_\infty
        (q^{\mu_i-\mu_j+1}\mathfrak t^{j-i-1};q)_\infty}.
  \label{eq:D-norm}
\end{equation}
Equivalently,
\begin{equation}
\begin{split}
  \mathcal N_\mu^{(N)}
  ={}&
  \prod_{1\leq i<j\leq N}\Biggl[
  \frac{(q^{\mu_i-\mu_j}\mathfrak t^{j-i};q)_\infty
        (q^{\mu_i-\mu_j+1}\mathfrak t^{j-i};q)_\infty}
       {(q^{\mu_i-\mu_j}\mathfrak t^{j-i+1};q)_\infty
        (q^{\mu_i-\mu_j+1}\mathfrak t^{j-i-1};q)_\infty}
  \\
  &\times
  \frac{(\mathfrak t^{j-i+1};q)_\infty
        (q\mathfrak t^{j-i-1};q)_\infty}
       {(\mathfrak t^{j-i};q)_\infty
        (q\mathfrak t^{j-i};q)_\infty}\Biggr].
\end{split}
  \label{eq:expanded-normalized-norm}
\end{equation}
Only differences $\mu_i-\mu_j$ occur, so
\begin{equation}
  \mathcal N_{\mu+m\mathbf1}^{(N)}
  =\mathcal N_\mu^{(N)}.
  \label{eq:norm-determinant-invariance}
\end{equation}

Define the Schur--Macdonald transition coefficients by
\begin{equation}
  s_\lambda(\mathbf x)
  =\sum_{\mu\preceq\lambda}
  \mathsf K_{\lambda\mu}^{(N)}(q,\mathfrak t)
  P_\mu(\mathbf x;q,\mathfrak t).
  \label{eq:schur-macdonald-expansion}
\end{equation}
This convention is inverse to the convention in which the $P_\mu$ are expanded in Schur functions.  Orthogonality gives the spectral resolution
\begin{equation}
  \boxed{
  \widehat{\mathbb I}^{\Neu}_\lambda(q,\mathfrak t)
  =
  \sum_{\mu\preceq\lambda}
  \left(\mathsf K_{\lambda\mu}^{(N)}(q,\mathfrak t)\right)^2
  \mathcal N_\mu^{(N)}(q,\mathfrak t).}
  \label{eq:wilson-spectral-general}
\end{equation}
At $\mathfrak t=q$, one has $P_\mu=s_\mu$ and $\mathcal N_\mu^{(N)}=1$.  Consequently
\begin{equation}
  \widehat{\mathbb I}^{\Neu}_\lambda(q,q)=1,
\end{equation}
as required by the normalized Weyl character pairing.

\subsection{One-row representations}
\label{subsec:one-row-wilson-lines}

For $\lambda=(r,0,\ldots,0)$, the character is $h_r(\mathbf x)$.  Introduce the virtual alphabet
\begin{equation}
  \mathbb A_{q,\mathfrak t}=\frac{1-q}{1-\mathfrak t},
  \qquad
  p_n[\mathbb A_{q,\mathfrak t}]
  =\frac{1-q^n}{1-\mathfrak t^n}.
  \label{eq:plethystic-alphabet}
\end{equation}
Let $Q_\mu=b_\mu P_\mu$ denote the stable dual Macdonald polynomial, where
\begin{equation}
  b_\mu(q,\mathfrak t)
  =\prod_{s\in\mu}
  \frac{1-q^{a(s)}\mathfrak t^{\ell(s)+1}}
       {1-q^{a(s)+1}\mathfrak t^{\ell(s)}}.
  \label{eq:b-mu}
\end{equation}

\begin{proposition}
\label{prop:one-row-transition}
For every $r\geq0$,
\begin{equation}
  h_r(\mathbf x)
  =
  \sum_{\substack{\mu\vdash r\\ \ell(\mu)\leq N}}
  Q_\mu[\mathbb A_{q,\mathfrak t}]
  P_\mu(\mathbf x;q,\mathfrak t),
  \label{eq:one-row-transition}
\end{equation}
with
\begin{equation}
  Q_\mu[\mathbb A_{q,\mathfrak t}]
  =
  \prod_{s\in\mu}
  \frac{\mathfrak t^{\ell'(s)}-q^{a'(s)+1}}
       {1-q^{a(s)+1}\mathfrak t^{\ell(s)}}.
  \label{eq:one-row-product-coeff}
\end{equation}
\end{proposition}

\begin{proof}
The stable Macdonald Cauchy identity is
\begin{equation}
  \sum_\mu P_\mu(\mathbf x)Q_\mu(\mathbf y)
  =
  \exp\!\left[
  \sum_{n\geq1}\frac{1-\mathfrak t^n}{1-q^n}
  \frac{p_n(\mathbf x)p_n(\mathbf y)}{n}
  \right].
  \label{eq:macdonald-cauchy-stable}
\end{equation}
Set $p_n(\mathbf y)=u^n(1-q^n)/(1-\mathfrak t^n)$.  The right-hand side becomes
\begin{equation}
  \prod_{i=1}^N(1-ux_i)^{-1}
  =\sum_{r\geq0}u^rh_r(\mathbf x),
\end{equation}
whereas the left-hand side becomes
$\sum_\mu u^{|\mu|}P_\mu(\mathbf x)Q_\mu[\mathbb A_{q,\mathfrak t}]$.
Coefficient comparison proves \eqref{eq:one-row-transition}; Macdonald's evaluation formula gives \eqref{eq:one-row-product-coeff}.
\end{proof}

It follows that
\begin{equation}
  \boxed{
  \widehat{\mathbb I}^{\Neu}_{(r)}(q,\mathfrak t)
  =
  \sum_{\substack{\mu\vdash r\\ \ell(\mu)\leq N}}
  Q_\mu[\mathbb A_{q,\mathfrak t}]^2
  \mathcal N_\mu^{(N)}(q,\mathfrak t).}
  \label{eq:one-row-wilson-closed}
\end{equation}
For later use, the first two transition formulae are
\begin{equation}
  h_2
  =P_{(2)}
  +\frac{\mathfrak t-q}{1-q\mathfrak t}P_{(1,1)},
  \label{eq:h2-expansion}
\end{equation}
and
\begin{equation}
  h_3
  =P_{(3)}
  +\frac{(1+q)(\mathfrak t-q)}{1-q^2\mathfrak t}P_{(2,1)}
  +\frac{(q-\mathfrak t)(q-\mathfrak t^2)}
  {(1-q\mathfrak t)(1-q\mathfrak t^2)}P_{(1,1,1)}.
  \label{eq:h3-expansion}
\end{equation}
Terms with more than $N$ rows are absent after specialization to $N$ variables.  The expansion \eqref{eq:h2-expansion} is Eq.~(3.63) of \cite{Hatsuda:2025bmn}, and for $r=2$ the formula \eqref{eq:one-row-wilson-closed} agrees with Eq.~(3.64) of that reference.

\subsection{Low-rank closed formulae}
\label{subsec:low-rank-formulae}

For $U(2)$ the normalized norm depends only on the difference of the two parts:
\begin{equation}
  \mathcal M_d^{(2)}(q,\mathfrak t)
  =
  \frac{(q^d\mathfrak t;q)_\infty(q^{d+1}\mathfrak t;q)_\infty}
       {(q^d\mathfrak t^2;q)_\infty(q^{d+1};q)_\infty}
  \frac{(\mathfrak t^2;q)_\infty(q;q)_\infty}
       {(\mathfrak t;q)_\infty(q\mathfrak t;q)_\infty}.
  \label{eq:U2-norm}
\end{equation}
The first two non-minuscule one-row results are
\begin{align}
  \widehat{\mathbb I}^{\Neu,U(2)}_{(2,0)}
  &=
  \mathcal M_2^{(2)}
  +\left(\frac{\mathfrak t-q}{1-q\mathfrak t}\right)^2,
  \label{eq:U2-r2-wilson}
  \\
  \widehat{\mathbb I}^{\Neu,U(2)}_{(3,0)}
  &=
  \mathcal M_3^{(2)}
  +\left(\frac{(1+q)(\mathfrak t-q)}
  {1-q^2\mathfrak t}\right)^2\mathcal M_1^{(2)}.
  \label{eq:U2-r3-wilson}
\end{align}
The two terms in each line are the Macdonald spectral channels allowed at finite $N=2$.  Equations \eqref{eq:U2-r2-wilson} and \eqref{eq:U2-r3-wilson} are the correlators of Eqs.~(3.12) and (3.13) of \cite{Hatsuda:2025bmn}, after division by the Neumann vacuum half-index and the substitution $\mathfrak t=q^{1/2}u_R^{-2}$.  A closed form for arbitrary symmetric rank in $U(2)$ is proposed in their Eq.~(3.14).  Determinant invariance gives
\begin{equation}
  \widehat{\mathbb I}^{\Neu,U(2)}_{(a,b)}
  =
  \widehat{\mathbb I}^{\Neu,U(2)}_{(a-b,0)},
  \qquad a\geq b.
  \label{eq:U2-all-by-one-row}
\end{equation}

For $U(3)$ one obtains
\begin{equation}
  \widehat{\mathbb I}^{\Neu,U(3)}_{(2,0,0)}
  =
  \mathcal N_{(2,0,0)}^{(3)}
  +\left(\frac{\mathfrak t-q}{1-q\mathfrak t}\right)^2
  \mathcal N_{(1,1,0)}^{(3)},
  \label{eq:U3-r2-wilson}
\end{equation}
and
\begin{equation}
\begin{split}
  \widehat{\mathbb I}^{\Neu,U(3)}_{(3,0,0)}
  ={}&
  \mathcal N_{(3,0,0)}^{(3)}
  +\left(\frac{(1+q)(\mathfrak t-q)}
  {1-q^2\mathfrak t}\right)^2
  \mathcal N_{(2,1,0)}^{(3)}
  \\
  &+
  \left(\frac{(q-\mathfrak t)(q-\mathfrak t^2)}
  {(1-q\mathfrak t)(1-q\mathfrak t^2)}\right)^2.
\end{split}
  \label{eq:U3-r3-wilson}
\end{equation}
The last finite-$N$ norm equals one because $(1,1,1)$ is a determinant representation.  Hatsuda and Okazaki study the $U(3)$ correlator of the second symmetric power in Section~3.2.2 of \cite{Hatsuda:2025bmn}.  Equation \eqref{eq:U3-r2-wilson} agrees with their normalized expression (3.64) at $N=3$.  Rational forms and short $q$-expansions are collected in Appendix~\ref{app:expansion-data}.

\section{Hecke convolution, boundary localization, and monopole bubbling}
\label{sec:boundary-convolution-localization}
\label{sec:boundary-bubbling}

This section has three parts.  First, the one-row coefficient recursion is
a lambda-ring identity in the polynomial representation of the spherical
DAHA.  Second, in the first non-minuscule example the two finite
Jeffrey--Kirwan residues of the local convolution are the cross terms of
$\cD_1^2$, the IC-summand subtraction realizes $h_2=e_1^2-e_2$, and their
spectator lift gives the first rank-three lower coefficient.  Third, the
identification of these algebraic coefficients with the output of a
microscopic regular Nahm-pole boundary SQM for arbitrary charge is formulated
as a localization conjecture.

\subsection{Determinant reduction and coefficient conventions}
\label{subsec:determinant-reduction-section4}

Let
\begin{equation}
 B=(B_1,\ldots,B_N),\qquad B_1\geq\cdots\geq B_N,
\end{equation}
be a dominant cocharacter of $U(N)$.  Before discussing screening, remove
the determinant translation by setting
\begin{equation}
 m=B_N,\qquad
 \bar B=B-m\one,
 \qquad \bar B_N=0.
 \label{eq:determinant-reduced-charge-section4}
\end{equation}
Any geometric or localization construction based on simple-root screening is
first applied to the determinant-reduced representative $\bar B$.
For the DAHA candidate, the determinant part is restored by
\begin{equation}
 \mathsf H^{\mathrm{DAHA}}_{B}
 =\mathsf C_N^m\mathsf H^{\mathrm{DAHA}}_{\bar B},
 \qquad
 \mathsf C_N=\mathfrak t^{N(N-1)/2}
 \prod_{a=1}^{N}T_{q,x_a}.
 \label{eq:determinant-translation-section4}
\end{equation}
The oppositely oriented line contains $\mathsf C_N^{-m}$, so these translations
cancel in an oriented two-line insertion.

For the remainder of this section we specialize the determinant-reduced
charge to the one-row family
\begin{equation}
 \bar B=(r,0,\ldots,0),\qquad r\geq0.
 \label{eq:one-row-determinant-reduced-charge-section4}
\end{equation}
The symbols $\mathsf H^{\mathrm{DAHA}}_{(r)}$ and
$\mathsf C^{\mathrm{DAHA}}_{r,\gamma}$ below denote the spherical-DAHA operator
and its coefficients.
The symbols $\widehat{\mathcal T}^{\partial,\mathrm{loc}}$ and
$\mathsf C^\partial$ are reserved for the microscopic boundary
localization operator and its coefficients.

For a composition
$\gamma=(\gamma_1,\ldots,\gamma_N)\in\mathbb Z_{\geq0}^N$, write
\begin{equation}
 T_{q^\gamma}=\prod_{a=1}^{N}T_{q,x_a}^{\gamma_a},
 \qquad |\gamma|=\sum_a\gamma_a,
\end{equation}
and let $\gamma^+$ be its non-increasing rearrangement.  A finite magnetic
difference operator is written as
\begin{equation}
 \mathsf H^{\mathrm{DAHA}}_{(r)}
 =\sum_{\substack{\gamma\in\mathbb Z_{\geq0}^N\\|\gamma|=r}}
 \mathsf C^{\mathrm{DAHA}}_{r,\gamma}(\mathbf x;q,\mathfrak t)
 T_{q^\gamma}.
 \label{eq:complete-coefficient-expansion-section4}
\end{equation}
Here $\mathsf C^{\mathrm{DAHA}}_{r,\gamma}$ is the \emph{complete algebraic
coefficient} of a
monomial shift.  Under the proposed boundary interpretation, it includes the
one-loop determinant of the effective magnetic line.  If a microscopic screened-sector calculation is normalized
to a pure bubbling index $\mathsf Z^{\partial}_{B,v}$, the relation has the
form
\begin{equation}
 \mathsf C^\partial_{B,\gamma}
 =\mathcal U^{\partial}_{v,\gamma}\,
  \mathsf Z^{\partial}_{B,v;\gamma},
 \qquad \gamma^+=v,
 \label{eq:full-versus-pure-bubbling-section4}
\end{equation}
where $\mathcal U^{\partial}_{v,\gamma}$ is the effective-line one-loop
factor.  Thus $\mathsf Z^{\partial}_{B,B}=1$ does not imply that the top
coefficient in \eqref{eq:complete-coefficient-expansion-section4} is one.
The sector operator is the complete Weyl-orbit sum
\begin{equation}
 \mathsf H^{\mathrm{DAHA}}_{(r)\to v}
 =\sum_{\gamma^+=v}
 \mathsf C^{\mathrm{DAHA}}_{r,\gamma}T_{q^\gamma}.
 \label{eq:sector-orbit-sum-section4}
\end{equation}

The explicit $U(2)$ model below uses the framing character
\begin{equation}
 \ch W=x_1+x_2
 \label{eq:explicit-U2-framing-character}
\end{equation}
and hence fixes the map from the framing weights to the four-dimensional
Cartan fugacities.  For a general charge, the full equivariant framing
character is part of the localization data in
Conjecture~\ref{conj:boundary-localization-section4}.

\subsection{The one-row lambda-ring recursion}
\label{subsec:exact-projector-recursion}

For $I\subseteq\{1,\ldots,N\}$ let
$\eps_I=\sum_{i\in I}\eps_i$, where $\eps_i$ is the $i$th standard basis
vector, and define
\begin{equation}
 \mathsf A_I(\mathbf x)
 =\mathfrak t^{|I|(|I|-1)/2}
 \prod_{\substack{i\in I\\j\notin I}}
 \frac{\mathfrak t x_i-x_j}{x_i-x_j}.
 \label{eq:A-I-projector-section4}
\end{equation}
The minuscule antisymmetric operator is
\begin{equation}
 \cD_k
 =\sum_{\substack{I\subseteq\{1,\ldots,N\}\\|I|=k}}
 \mathsf A_I(\mathbf x)T_{q^{\eps_I}},
 \qquad
 \cD_0=1,
 \qquad
 \cD_k=0\quad(k>N).
 \label{eq:Dk-minuscule-section4}
\end{equation}

The identity of symmetric functions
\begin{equation}
 h_r=\sum_{k=1}^{\min(r,N)}(-1)^{k+1}e_kh_{r-k}\qquad(r\geq1)
 \label{eq:complete-elementary-recursion-section4}
\end{equation}
is the coefficient of $u^r$ in $H(u)E(-u)=1$.  Evaluating this identity on
the commuting Cherednik operators and using
\begin{equation}
 \mathbf e_+e_k(Y_1,\ldots,Y_N)\mathbf e_+
 =\cD_k\mathbf e_+
 \label{eq:spherical-ek-Dk-section4}
\end{equation}
in the polynomial representation gives the corresponding identity of
finite difference operators on symmetric Laurent polynomials.  Here
$\mathbf e_+$ is the finite-Hecke spherical idempotent defined in
Section~\ref{sec:daha-kernel}.

\begin{proposition}[One-row DAHA/lambda-ring recursion]
\label{prop:one-row-projector-recursion-section4}
Set $\mathsf C^{\mathrm{DAHA}}_{0,0}=1$ and set $\mathsf C^{\mathrm{DAHA}}_{r,\gamma}=0$ if an entry of
$\gamma$ is negative or if $|\gamma|\neq r$.  The coefficients of the
one-row spherical-DAHA operator obey, for $r\geq1$,
\begin{equation}
 \boxed{
 \mathsf C^{\mathrm{DAHA}}_{r,\gamma}(\mathbf x)
 =\sum_{\substack{\varnothing\neq I\subseteq\{1,\ldots,N\}\\
                   |I|\leq r,\ \eps_I\leq\gamma}}
 (-1)^{|I|+1}\mathsf A_I(\mathbf x)
 \mathsf C^{\mathrm{DAHA}}_{r-|I|,\gamma-\eps_I}
 (q^{\eps_I}\mathbf x).}
 \label{eq:projector-coefficient-recursion-section4}
\end{equation}
Its generating function is
\begin{equation}
 \sum_{r\geq0}u^r\mathsf H^{\mathrm{DAHA}}_{(r)}
 =\left(\sum_{k=0}^{N}(-u)^k\cD_k\right)^{-1}.
 \label{eq:projector-generating-function-section4}
\end{equation}
The inverse in this formula is understood as a formal power series in $u$.
\end{proposition}

\begin{proof}
Apply \eqref{eq:complete-elementary-recursion-section4} to the commuting
$Y$-operators and sandwich the result by $\mathbf e_+$.  Equation
\eqref{eq:spherical-ek-Dk-section4} gives the operator recursion with
$\cD_k$ as its leftmost factor.  For the coefficient of $T_{q^\gamma}$,
multiplication of difference operators gives
\begin{equation}
 \mathsf A_I(\mathbf x)T_{q^{\eps_I}}
 \left[
 \mathsf C^{\mathrm{DAHA}}_{r-|I|,\gamma-\eps_I}(\mathbf x)
 T_{q^{\gamma-\eps_I}}
 \right]
 =\mathsf A_I(\mathbf x)
 \mathsf C^{\mathrm{DAHA}}_{r-|I|,\gamma-\eps_I}
 (q^{\eps_I}\mathbf x)T_{q^\gamma}.
\end{equation}
Summing over $I$ proves
\eqref{eq:projector-coefficient-recursion-section4}.  The generating
function is the standard relation between complete and elementary
symmetric functions.
\end{proof}

Classical geometric Satake provides an analogue.  If $\cL_{(r)}$ denotes the one-row object and $\cL_{(1^k)}$ the
antisymmetric minuscule object, then in the unrefined Grothendieck group
\begin{equation}
 [\cL_{(r)}]
 =\sum_{k=1}^{\min(r,N)}(-1)^{k+1}
 [\cL_{(1^k)}\star\cL_{(r-k)}].
 \label{eq:Satake-one-row-recursion-section4}
\end{equation}
This is the representation-ring shadow of the same complete--elementary
identity and serves here as geometric motivation.  The
$(q,\mathfrak t)$-dependent coefficients are not obtained from it.

The subset $I$ in
\eqref{eq:projector-coefficient-recursion-section4} is the
\emph{first}, or leftmost, factor in the ordered product of difference
operators.  Removing
this first factor leaves the lower coefficient evaluated at
$q^{\eps_I}\mathbf x$.  The
subsets appearing upon iteration will be called \emph{recursion paths}.
They are the terms of the lambda-ring inclusion--exclusion formula
and need not correspond to fixed points of a localization space.  In particular, the
factor
\begin{equation}
 (-1)^{|I|+1}\mathfrak t^{|I|(|I|-1)/2}
 \label{eq:projector-sign-weight-section4}
\end{equation}
is the lambda-ring inclusion--exclusion sign times the standard
minuscule coefficient.

\begin{conjecture}[Boundary-localization identification]
\label{conj:boundary-localization-section4}
For the regular Nahm-pole boundary condition, the complete coefficients
produced by a microscopic localization calculation of the irreducible
one-row boundary 't Hooft line agree, in the normalization in which the
minuscule antisymmetric lines are represented by the operators $\cD_k$ of
\eqref{eq:Dk-minuscule-section4}, with the coefficients
$\mathsf C^{\mathrm{DAHA}}_{r,\gamma}$ in
\eqref{eq:projector-coefficient-recursion-section4}, after the pure
bubbling index and the effective-line factor have been combined as in
\eqref{eq:full-versus-pure-bubbling-section4}; explicitly,
\begin{equation}
 \mathsf C^\partial_{(r,0,\ldots,0),\gamma}
 =\mathsf C^{\mathrm{DAHA}}_{r,\gamma}.
 \label{eq:boundary-localization-conjecture-section4}
\end{equation}
\end{conjecture}

The normalization has to be
specified because the coefficients of a difference operator change under
conjugation by a multiplication operator, see
\eqref{eq:conjugation-by-multiplication-section4} below.  The conjecture
concerns the operator of a single line.  A two-point function requires in
addition a statement about the boundary, which is discussed in
Section~\ref{sec:eval-module}.

\subsection{Why a naive handsaw quotient is insufficient}
\label{subsec:handsaw-diagnostic-section4}

The first screened sector gives a test for a proposed boundary SQM.  For
\begin{equation}
 U(2),\qquad B=(2,0),\qquad v=(1,1),
 \label{eq:minimal-nonminuscule-data-section4}
\end{equation}
a one-sided handsaw ansatz with
$V_1=\C$ and $W_1=\C^2$ has only an endomorphism $A$ and a nonzero framing
map $a$.  Its stable quotient is
\begin{equation}
 \{(A,a):a\neq0\}/\C^*
 \simeq\C\times\mathbb P^1.
 \label{eq:naive-handsaw-quotient-section4}
\end{equation}
It has two torus fixed points and is not holomorphic symplectic.  Moreover,
the corresponding naive virtual character contains a trivial gauge weight;
after deleting only that zero mode, its rank-one contour has the form
\begin{equation}
 \omega_{\rm one\text{-}sided}(u)
 =\frac{1-\mathfrak t q}{1-q}\frac{du}{2\pi i\,u}
 F_{\mathfrak t}(u/w_1)F_{\mathfrak t}(u/w_2),
 \qquad
 F_{\mathfrak t}(z)=\frac{1-\mathfrak t z}{1-z}.
 \label{eq:one-sided-diagnostic-integrand-section4}
\end{equation}
The two same-charge residues sum to
\begin{equation}
 \frac{(\mathfrak t-1)(\mathfrak t+1)(q\mathfrak t-1)}{q-1},
\end{equation}
which is independent of $w_1/w_2$.  Its residues at zero and infinity are
also flavor independent.  It therefore cannot produce the coefficient
below, which depends nontrivially on $x_1/x_2$ and has affine $q$-root
poles.  Thus the required coefficient cannot be obtained by changing only
the JK chamber or by appending an ordinary residue at zero or infinity; the local
model and the IC/Satake summand subtraction must both be specified.

\subsection{The first non-minuscule local convolution calculation}
\label{subsec:U2-convolution-SQM-section4}
\label{subsec:rank-two-independent-residue}

We now compute the local convolution factor by a rank-one contour, without using
Proposition~\ref{prop:one-row-projector-recursion-section4}.  The two residues
obtained are the cross terms of $\cD_1^2$.  Put
\begin{equation}
 s=q^{-1},\qquad y=\frac{x_1}{x_2},\qquad
 F_{\mathfrak t}(z)=\frac{1-\mathfrak t z}{1-z}.
 \label{eq:Ft-U2-section4}
\end{equation}
Let the $U(1)$ gauge and framing spaces be
\begin{equation}
 V=\C_u,
 \qquad
 W=\C_{x_1}\oplus\C_{x_2}.
 \label{eq:U2-gauge-framing-spaces-section4}
\end{equation}
In a chosen complex structure the scalar fields are
\begin{equation}
 Q_a\in\operatorname{Hom}(V,W_a),
 \qquad
 \widetilde Q_a\in q\operatorname{Hom}(W_a,V),
 \qquad a=1,2,
 \label{eq:U2-full-hyper-fields-section4}
\end{equation}
and the complex moment map is
\begin{equation}
 \mu_{\C}=\sum_{a=1}^{2}\widetilde Q_aQ_a
 \in q\operatorname{End}(V).
 \label{eq:U2-complex-moment-map-section4}
\end{equation}
The moment-map Fermi has $J=\mu_{\C}$, so every term in the constraint has
the same equivariant weight.  To avoid confusing scalar charges with the
dual characters in a K-theory denominator, both are displayed in
Table~\ref{tab:U2-multiplets-section4}.

\begin{table}[ht]
\centering
\small
\begin{tabular}{@{}llll@{}}
\toprule
object
& gauge charge
& scalar and cotangent characters
& type \\
\midrule
$Q_a\in\operatorname{Hom}(V,W_a)$
& $-1$ & $(x_a/u,u/x_a)$ & chiral \\
$\widetilde Q_a\in q\operatorname{Hom}(W_a,V)$
& $+1$ & $(q u/x_a,s x_a/u)$ & chiral \\
$\mu_{\C}=0$
& $0$ & $(q,s)$ & moment-map Fermi \\
$U(1)$ quotient
& adjoint & $(1,1)$ & vector \\
\bottomrule
\end{tabular}
\caption{The doubled rank-one quiver.  In the ordered pair in the third
column, the second entry is the monomial entering
$\lambda_{-\mathfrak t}(T^*)/\lambda_{-1}(T^*)$.}
\label{tab:U2-multiplets-section4}
\end{table}

Positive stability is $Q\neq0$.  The holomorphic symplectic quotient is
\begin{equation}
 X=
 \{(Q,\widetilde Q):\widetilde QQ=0,\ Q\neq0\}/\C^*
 \simeq T^*\mathbb P^1.
 \label{eq:TstarP1-Higgs-branch-section4}
\end{equation}
This is the standard $U(1)$ $\mathcal N=(4,4)$ SQM with two full
hypermultiplets which appears in the minimal bulk bubbling problem
\cite{Brennan:2018yuj,Mekareeya:2013ija}.  Here it is used as the local
model for the convolution of two minuscule modifications, and the
calculation below is the Higgs-branch K-theoretic contour
functional of this quotient.

The cotangent virtual character of the quotient-and-constraint complex is
\begin{equation}
 T_{\rm vir}^*
 =\sum_{a=1}^{2}\left(\frac{u}{x_a}
              +\frac{s x_a}{u}\right)-1-s.
 \label{eq:U2-virtual-cotangent-character-section4}
\end{equation}
The term $-s$ is the conormal of the complex moment-map equation and
therefore contributes
$F_{\mathfrak t}(s)^{-1}=(1-s)/(1-\mathfrak t s)$.  The Cartan term $-1$
is a gauge zero mode and is not inserted literally into the lambda-class
ratio.  We normalize the abelian quotient-residue functional by the factor
$(\mathfrak t-1)^{-1}$; this is the normalization for which the
contour equals the localized lambda-class functional on $T^*\mathbb P^1$.
Here $\mathfrak t$ is the lambda-class, or Hodge-grading,
variable and does not appear among the field characters of
Table~\ref{tab:U2-multiplets-section4}.
For JK purposes, the charge covector of a denominator is the exponent of
$u$ in its cotangent monomial.  It is consequently the opposite of the
gauge charge of the corresponding scalar displayed in
Table~\ref{tab:U2-multiplets-section4}.  The normalized Higgs-branch
K-theoretic one-form is
\begin{equation}
 \omega_{\rm conv}(u)=
 \frac{1-s}{(\mathfrak t-1)(1-\mathfrak t s)}
 \frac{du}{2\pi i\,u}
 \prod_{a=1}^{2}
 F_{\mathfrak t}(u/x_a)F_{\mathfrak t}(s x_a/u).
 \label{eq:U2-convolution-integrand-section4}
\end{equation}
In this multiplicative convention, the denominator hyperplanes
$1-u/x_a=0$ have charge covector $+1$ and the hyperplanes
$1-sx_a/u=0$ have charge covector $-1$.  The chamber $\eta>0$ therefore
selects
\begin{equation}
 u=x_1,\qquad u=x_2.
 \label{eq:U2-positive-poles-section4}
\end{equation}

Define
\begin{equation}
 A_1(\mathbf x)=\frac{\mathfrak t x_1-x_2}{x_1-x_2},
 \qquad
 A_2(\mathbf x)=\frac{\mathfrak t x_2-x_1}{x_2-x_1}.
 \label{eq:U2-A1-A2-section4}
\end{equation}
Using
\begin{equation}
 \operatorname*{Res}_{u=x_i}\frac{du}{u}F_{\mathfrak t}(u/x_i)
 =\mathfrak t-1,
 \qquad
 \frac{1-s}{1-\mathfrak t s}F_{\mathfrak t}(s)=1,
\end{equation}
one obtains the individual residues
\begin{align}
 \operatorname*{Res}_{u=x_1}\omega_{\rm conv}
 &=F_{\mathfrak t}(y)F_{\mathfrak t}((qy)^{-1})
 \nonumber\\
 &=A_1(\mathbf x)A_2(qx_1,x_2)=:R_{12},
 \label{eq:U2-residue-R12-section4}\\
 \operatorname*{Res}_{u=x_2}\omega_{\rm conv}
 &=F_{\mathfrak t}(y^{-1})F_{\mathfrak t}(y/q)
 \nonumber\\
 &=A_2(\mathbf x)A_1(x_1,qx_2)=:R_{21}.
 \label{eq:U2-residue-R21-section4}
\end{align}
Thus
\begin{equation}
 \JKRes_{\eta>0}\omega_{\rm conv}=R_{12}+R_{21}.
 \label{eq:U2-positive-JK-sum-section4}
\end{equation}
At the two fixed points the cotangent localization characters are
\begin{equation}
 p_1:\ \left(y,\frac1{qy}\right),
 \qquad
 p_2:\ \left(y^{-1},\frac yq\right).
 \label{eq:U2-cotangent-weights-section4}
\end{equation}
Their products are $q^{-1}$, the chosen cotangent symplectic character.

For products of minimal bulk 't Hooft lines, FI chambers encode operator
ordering and can support wall crossing \cite{Hayashi:2019rpw}.  In the
rank-one local model used here, the two chambers instead give the same
localized sum.  The ordinary residues at the negative-charge poles are
\begin{equation}
 \operatorname*{Res}_{u=sx_1}\omega_{\rm conv}=-R_{21},
 \qquad
 \operatorname*{Res}_{u=sx_2}\omega_{\rm conv}=-R_{12}.
\end{equation}
The rank-one JK orientation supplies an additional minus sign for a
charge-$-1$ hyperplane, and therefore
\begin{equation}
 \JKRes_{\eta<0}\omega_{\rm conv}=R_{12}+R_{21}.
\end{equation}
The residues at the two ends of the multiplicative Coulomb cylinder are
\begin{equation}
 \operatorname*{Res}_{u=0}\omega_{\rm conv}
 =\frac{\mathfrak t^2(1-s)}
 {(\mathfrak t-1)(1-\mathfrak t s)},
 \qquad
 \operatorname*{Res}_{u=\infty}\omega_{\rm conv}
 =-\operatorname*{Res}_{u=0}\omega_{\rm conv}.
 \label{eq:U2-zero-infinity-section4}
\end{equation}
They cancel, and neither is $-\mathfrak t$ for generic parameters.  The
missing term is therefore neither a third Higgs fixed point nor an
ordinary contribution at zero or infinity.

\subsection{The IC/Satake summand subtraction in the first non-minuscule case}
\label{subsec:U2-IC-projector-section4}

The two-pole answer
\eqref{eq:U2-positive-JK-sum-section4} is the coefficient in the
convolution of two fundamental modifications.  It is not yet the
irreducible one-row coefficient.  In the Satake category
\begin{equation}
 IC_{\omega_1}\star IC_{\omega_1}
 \simeq IC_{2\omega_1}\oplus IC_{\omega_2},
 \label{eq:U2-Satake-decomposition-section4}
\end{equation}
corresponding to
$\C^2\otimes\C^2=\Sym^2\C^2\oplus\wedge^2\C^2$
\cite{Mirkovic:2007}.  The determinant object is minuscule and has the
complete shift
\begin{equation}
 \mathsf H^{\mathrm{DAHA}}_{(1,1)}
 =\mathfrak t T_{q,x_1}T_{q,x_2};
 \label{eq:U2-determinant-shift-section4}
\end{equation}
this follows directly from
\eqref{eq:A-I-projector-section4} with $I=\{1,2\}$, whose cross-product is
empty.  At the level of operators, subtracting it is the step from $\cD_1^2$
to $\cD_1^2-\cD_2$, which corresponds to $h_2=e_1^2-e_2$.  Hence
\begin{equation}
 \boxed{
 \mathsf C^{\mathrm{IC}}_{2,(1,1)}
 =R_{12}+R_{21}-\mathfrak t
 =\mathsf C^{\mathrm{DAHA}}_{2,(1,1)}.}
 \label{eq:U2-complete-coefficient-section4}
\end{equation}

This subtraction admits the following precise resolution-side Hodge--Chern
representative.  The
transverse slice of the convolution morphism along the determinant stratum
is the $A_1$ surface singularity, and its semismall resolution is
\begin{equation}
 \pi:X=T^*\mathbb P^1\longrightarrow
 \mathcal S\simeq\C^2/\mathbb Z_2.
 \label{eq:U2-A1-resolution-section4}
\end{equation}
Let $0\in\mathcal S$ be the singular point and
$i:E=\pi^{-1}(0)=\mathbb P^1\hookrightarrow X$ the exceptional zero
section.  The transverse-slice statement is standard
\cite{Malkin:2003cs}.  Let $\bT$ be the torus generated by the flavor
ratio and cotangent scaling, with equivariant characters expressed
through $x_1/x_2$ and $q$.  The morphism $\pi$ is proper and
$\bT$-equivariant.  The semismall decomposition theorem, in its
equivariant mixed-Hodge refinement, gives
\begin{equation}
 R\pi_*\mathbb Q_X^H[2]
 \simeq IC_{\mathcal S}^H\oplus\mathbb Q_0^H(-1).
 \label{eq:U2-Hodge-module-decomposition-section4}
\end{equation}
We use the geometric Satake and mixed-Hodge-module conventions of
\cite{Mirkovic:2007,Schuermann:2009characteristic}.
The point summand, arising from
$H^2(E,\mathbb Q)=\mathbb Q(-1)$, is the determinant channel.  Since
$\bT$ is connected, its action on $H^2(E,\mathbb Q)$ is trivial; hence this
summand carries no additional flavor or cotangent-scaling character.

The nonequivariant motivic Chern transformation, including proper
pushforward and smooth normalization, was constructed in
\cite{Brasselet:2005}.  Let
\[
 MHC_z^{\bT}:K_0\!\left(D^b\mathrm{MHM}^{\bT}(-)\right)
 \longrightarrow G_0^{\bT}(-)\otimes\mathbb Z[z^{\pm1}]
\]
be the equivariant Hodge--Chern transformation, using the mixed-Hodge-module
refinement of \cite{Schuermann:2009characteristic} and the standard
equivariant extension as in \cite{Feher:2018motivic}.  We use the Grothendieck-group
convention $MHC_z^{\bT}([M[k]])=(-1)^kMHC_z^{\bT}([M])$.  It commutes with
proper equivariant pushforward and, because the shift $[2]$ is even, has the
smooth normalization
\begin{equation}
 MHC_z^{\bT}(\mathbb Q_X^H[2])=\lambda_z(T_X^*),
 \qquad
 \lambda_z(F)=\sum_{k\geq0}z^k[\wedge^kF],
\end{equation}
and satisfies
$MHC_z^{\bT}(\mathbb Q_0^H(-1))=-z[\cO_0]$.  The Hodge variable $z$
is independent of the equivariant characters $x_a$ and $q$; only after
applying the transformation do we specialize $z=-\mathfrak t$.
Consequently
\begin{equation}
 MHC_z^{\bT}(IC_{\mathcal S}^H)
 =\pi_*^{\bT}\lambda_z(T_X^*)+z\cO_0.
 \label{eq:U2-Hodge-Chern-pushforward-section4}
\end{equation}
Since $\pi_*^{\bT}i_*\cO_E=\cO_0$, setting $z=-\mathfrak t$ gives the identity in $G_0^{\bT}(\mathcal S)[\mathfrak t]$
\begin{equation}
 \boxed{
 MHC_{-\mathfrak t}^{\bT}(IC_{\mathcal S}^H)=\pi_*^{\bT}\mathcal K_{(2)},
 \qquad
 \mathcal K_{(2)}
 =\lambda_{-\mathfrak t}(T_X^*)
  -\mathfrak t i_*\cO_E.}
 \label{eq:U2-projector-K-class-section4}
\end{equation}
Here $\mathcal K_{(2)}\in G_0^{\bT}(X)[\mathfrak t]$; since $X$ is smooth,
$G_0^{\bT}(X)$ may be identified with equivariant vector-bundle K-theory.
Thus the factor $\mathfrak t$ is the Hodge grading of the Tate-twisted
determinant summand.  The class $\mathcal K_{(2)}$ is a resolution-side
representative whose proper pushforward gives the Hodge--Chern class of the
IC summand.

The space $X$ is nonproper, so the expression used below is not an
ordinary pushforward from $X$ to a point.  It is the localized fixed-point
functional
\begin{equation}
 \chi_{\bT,z}^{\rm loc}:G_0^{\bT}(X)[z]
 \longrightarrow\operatorname{Frac}R(\bT)[z],
 \qquad
 \mathcal F\longmapsto\sum_{p\in X^{\bT}}
 \frac{\mathcal F|_p}{\lambda_{-1}(T_p^*X)}
 .
 \label{eq:localized-fixed-point-functional-section4}
\end{equation}
Here $\mathcal F|_p$ means the derived restriction
$[L i_p^*\mathcal F]$ in $R(\bT)$.
After the specialization $z=-\mathfrak t$, its values lie in
$\operatorname{Frac}R(\bT)[\mathfrak t]$.

The same result is transparent at the two fixed points.  The Koszul
self-intersection formula gives
\begin{equation}
 (i_*\cO_E)|_{p_a}=\lambda_{-1}(N^*_{E/X}|_{p_a}).
\end{equation}
The normal factor cancels against its factor in
$\lambda_{-1}(T_{p_a}^*X)$.  The cotangent weights of $E$ at $p_1$ and
$p_2$ are $y$ and $y^{-1}$, respectively, and therefore
\begin{align}
 \chi_{\bT,-\mathfrak t}^{\rm loc}(\mathcal K_{(2)})
 &=R_{12}+R_{21}
 -\mathfrak t\left(\frac1{1-y}+\frac1{1-y^{-1}}\right)
 \nonumber\\
 &=R_{12}+R_{21}-\mathfrak t.
 \label{eq:U2-projector-fixed-point-index-section4}
\end{align}
The individual fixed-point summands are
\begin{equation}
 R_{12}-\frac{\mathfrak t}{1-y},
 \qquad
 R_{21}-\frac{\mathfrak t}{1-y^{-1}}.
\end{equation}
They are the two fixed-point contributions of the virtual IC-summand
representative.

Because $E$ is a Cartier divisor, the subtracted class has the elementary
Koszul resolution
\begin{equation}
 0\longrightarrow\cO_X(-E)
 \xrightarrow{s_E}\cO_X
 \longrightarrow i_*\cO_E\longrightarrow0,
 \label{eq:U2-zero-section-Koszul-section4}
\end{equation}
and hence
\begin{equation}
 -\mathfrak t[i_*\cO_E]
 =\mathfrak t[\cO_X(-E)]-\mathfrak t[\cO_X].
\end{equation}
The negative class is the Grothendieck-group subtraction of the
Tate-twisted point summand, represented through the Koszul resolution above.

The representative $\mathcal K_{(2)}$ is the specified zero-section lift of
the pushed-forward Hodge--Chern class, defined modulo $\ker\pi_*$.
Conjecture~\ref{conj:boundary-localization-section4} predicts that its
localized fixed-point index gives the coefficient of the $(1,1)$ shift in
the regular Nahm-pole line operator of charge $(2,0)$.

The two chambers, the compactification residues, the fixed-point identity and
the rank-three lift below are checked symbolically in
\texttt{verify\_local\_convolution\_ic.py}.

In the convention \eqref{eq:full-versus-pure-bubbling-section4} the effective determinant
line in this sector contributes
$\mathcal U^{\partial}_{(1,1)}=\mathfrak t$, so that the pure bubbling factor is
the coefficient of the $(1,1)$ shift divided by $\mathfrak t$.  For the
predicted coefficient \eqref{eq:U2-complete-coefficient-section4} this gives
\begin{equation}
 \mathsf Z^{\partial}_{(2,0),(1,1)}
 =\frac{\mathsf C^{\mathrm{IC}}_{2,(1,1)}}{\mathfrak t}
 =\frac{R_{12}+R_{21}}{\mathfrak t}-1.
 \label{eq:U2-pure-bubbling-normalization-section4}
\end{equation}
Equation
\eqref{eq:U2-complete-coefficient-section4}, rather than
\eqref{eq:U2-pure-bubbling-normalization-section4}, is the coefficient
multiplying $T_{q,x_1}T_{q,x_2}$ in the full difference operator.

The top sector illustrates the same distinction.  Its pure bubbling index
is one, whereas its complete coefficients are
\begin{equation}
 \mathsf C^{\mathrm{DAHA}}_{2,(2,0)}
 =A_1(\mathbf x)A_1(qx_1,x_2),
 \qquad
 \mathsf C^{\mathrm{DAHA}}_{2,(0,2)}
 =A_2(\mathbf x)A_2(x_1,qx_2).
 \label{eq:U2-top-complete-coefficients-section4}
\end{equation}

\subsection{The first rank-three transverse-slice lift}
\label{subsec:U3-spectator-lift-section4}
Consider the lower shift
$(2,0,0)=2\omega_1\to(1,1,0)=\omega_2$ of $U(3)$.  The minimal
degeneration has the same transverse $A_1$ slice as the rank-two example
\cite{Malkin:2003cs}; the third color supplies the spectator-root factor
\begin{equation}
 \mathsf S_{12|3}(\mathbf x)
 :=\prod_{a=1}^{2}\frac{\mathfrak t x_a-x_3}{x_a-x_3}.
 \label{eq:U3-spectator-factor-section4}
\end{equation}
Writing a superscript $[N]$ for the number of colors, define
\begin{align}
 R_{12}^{[3]}
 &:={\mathsf A}_{\{1\}}(\mathbf x)
    {\mathsf A}_{\{2\}}(qx_1,x_2,x_3),
 &
 R_{21}^{[3]}
 &:={\mathsf A}_{\{2\}}(\mathbf x)
    {\mathsf A}_{\{1\}}(x_1,qx_2,x_3).
\end{align}
The spectator factor is independent of the contour variable, and the two
ordered residues factorize separately:
\begin{equation}
 R_{12}^{[3]}=\mathsf S_{12|3}R_{12},
 \qquad
 R_{21}^{[3]}=\mathsf S_{12|3}R_{21}.
 \label{eq:U3-ordered-residue-factorization-section4}
\end{equation}
The coefficient of $T_{q,x_1}T_{q,x_2}$ in the antisymmetric minuscule
operator $\mathcal D_2$ is
\begin{equation}
 \mathsf A_{\{1,2\}}(\mathbf x)
 =\mathfrak t\,\mathsf S_{12|3}(\mathbf x).
 \label{eq:U3-minuscule-coefficient-section4}
\end{equation}
Consequently the IC-subtracted transverse calculation gives
\begin{equation}
 \boxed{
 \mathsf C^{\mathrm{IC},U(3)}_{2,(1,1,0)}
 =\mathsf S_{12|3}(\mathbf x)
  \bigl(R_{12}+R_{21}-\mathfrak t\bigr)
 =\mathsf C^{\mathrm{DAHA},U(3)}_{2,(1,1,0)}.}
 \label{eq:U3-IC-factorization-section4}
\end{equation}
Thus the two ordered residues and the antisymmetric IC/Satake subtraction
acquire the same ambient factor.  After division by the complete minuscule
effective-line coefficient $\mathfrak t\mathsf S_{12|3}$, the normalized
transverse factor obeys
\begin{equation}
 \frac{\mathsf C^{\mathrm{IC},U(3)}_{2,(1,1,0)}}
      {\mathfrak t\mathsf S_{12|3}}
 =\frac{\mathsf C^{\mathrm{IC}}_{2,(1,1)}}{\mathfrak t}.
 \label{eq:U3-normalized-transverse-factor-section4}
\end{equation}
This is a
rank-three embedding check of the same $A_1$ local convolution.  The other
representatives in the $(1,1,0)$ shift orbit follow by Weyl permutation.

The ratio in \eqref{eq:U3-normalized-transverse-factor-section4} compares two
coefficients of the same shift.  It is therefore unchanged if the difference
operators are conjugated by a multiplication operator $M_\Psi f=\Psi f$, since
\begin{equation}
 M_\Psi^{-1}\bigl(C_\gamma(\mathbf x)T_{q^\gamma}\bigr)M_\Psi
 =C_\gamma(\mathbf x)\,\frac{\Psi(q^\gamma\mathbf x)}{\Psi(\mathbf x)}\,
  T_{q^\gamma}.
 \label{eq:conjugation-by-multiplication-section4}
\end{equation}
At $x_1/x_2=\mathfrak t$ the ratio takes the value
$(\mathfrak t-q)/(1-q\mathfrak t)$, which is the coefficient of $P_{(1,1)}$ in
\eqref{eq:h2-expansion}.  This value enters the matrix element
\eqref{eq:eval-r2}.

\subsection{Status of the general boundary statement}
\label{subsec:status-general-localization-section4}

The outcome can be summarized as follows.
\begin{enumerate}
 \item The one-row recursion
 \eqref{eq:projector-coefficient-recursion-section4} is an identity in the
 spherical DAHA.  Its signs
 and powers of $\mathfrak t$ come from lambda-ring inclusion--exclusion
 and the minuscule antisymmetric coefficients.

 \item The local $U(2)$ convolution quotient has precisely two Higgs fixed
 points.  Its K-theoretic contour residues are
 $R_{12}$ and $R_{21}$, the two cross terms of $\cD_1^2$, and the chamber and
 compactification checks do not
 generate $-\mathfrak t$.

 \item The $-\mathfrak t$ term results from subtracting the determinant
 Satake summand.  The Hodge--Chern pushforward and its Koszul representative
 are given in \eqref{eq:U2-projector-K-class-section4} for the transverse
 $A_1$ model.  At the level of operators the subtraction corresponds to
 $h_2=e_1^2-e_2$.

 \item In $U(3)$, the two residues and the antisymmetric subtraction acquire
 the same spectator factor $\mathsf S_{12|3}$.  Equation
 \eqref{eq:U3-IC-factorization-section4} reproduces the complete first lower
 coefficient and its effective-line normalization in the next ambient rank.

 \item Identifying the IC-subtracted coefficients with a microscopic
 regular Nahm-pole boundary localization calculation for arbitrary $r$ is
 Conjecture~\ref{conj:boundary-localization-section4}.  The
 local convolution calculation is a geometric realization of the first lower
 coefficient and not a microscopic test of this conjecture.

 \item For a general $U(N)$ charge the construction is applied to the
 determinant-reduced charge
 \eqref{eq:determinant-reduced-charge-section4}, and the determinant
 translation is restored by
 \eqref{eq:determinant-translation-section4}.
\end{enumerate}

\section{The spherical DAHA operator and its Macdonald spectrum}
\label{sec:daha-kernel}

The Wilson-line calculation of Section~\ref{sec:wilson-side} is naturally expressed in the $N$-variable Macdonald pairing.  The double affine Hecke algebra supplies the corresponding magnetic difference operators: multiplication by a symmetric function of the $X$-operators is exchanged, after a choice of Fourier orientation, with a spherical symmetric function of the commuting Cherednik $Y$-operators.  We fix the conventions needed for this statement, derive the finite shift expansion for one-row charges, state the spectral theorem directly in the discrete Macdonald basis, and realize the exchange on a module of evaluation functionals.  Difference-operator realizations of monopole operators and related integrable systems appear in \cite{Nekrasov:2009uh,Nekrasov:2009rc,Gaiotto:2013bwa,Bullimore:2015lsa}.  A complementary rank-one brane-quantization construction relates spherical-DAHA modules to line-operator algebras in four-dimensional $\mathcal N=2^*$ theories \cite{Gukov:2022gei}.  Noumi and Sano constructed explicit operators in the algebra generated by the Macdonald--Ruijsenaars operators whose eigenvalue on $P_\mu$ is $Q_{(r)}(\boldsymbol\xi_\mu;q,\mathfrak t)$ \cite{Noumi:2021sano}.  The operators studied here have the eigenvalue $h_r(\boldsymbol\xi_\mu)$ required by the character of the symmetric representation, and the two families coincide at $\mathfrak t=q$.  Our DAHA conventions follow \cite{Macdonald:1995,Cherednik:1995,Cherednik:2005}.

\subsection{The polynomial representation and the finite-Hecke spherical idempotent}
\label{subsec:polynomial-representation}

Let
\begin{equation}
  \mathscr P_N
  =\mathbb C(q,\tau)
  [x_1^{\pm1},\ldots,x_N^{\pm1}],
  \qquad \tau^2=\mathfrak t,
\end{equation}
and let $S_N$ act by permutations of the variables.  We denote by $s_i$ the transposition of $x_i$ and $x_{i+1}$.  In the polynomial representation the finite-Hecke generators act as
\begin{equation}
  \mathcal T_i
  =\tau s_i+
  \frac{\tau-\tau^{-1}}{x_i/x_{i+1}-1}(s_i-1)
  =\tau s_i+
  \frac{\tau-\tau^{-1}}{1-x_i/x_{i+1}}(1-s_i)
  \label{eq:DL-operator}
\end{equation}
for $1\leq i<N$.  These operators obey
\begin{equation}
  (\mathcal T_i-\tau)(\mathcal T_i+\tau^{-1})=0,
  \qquad
  \mathcal T_i^{-1}=\mathcal T_i-(\tau-\tau^{-1}),
  \label{eq:Hecke-quadratic-relation}
\end{equation}
together with the braid relations.  The $X_i$ act by multiplication,
\begin{equation}
  X_i f(\mathbf x)=x_i f(\mathbf x),
\end{equation}
and the affine rotation is
\begin{equation}
  \omega f(x_1,\ldots,x_N)
  =f(qx_N,x_1,\ldots,x_{N-1}).
  \label{eq:affine-rotation}
\end{equation}

It is useful to distinguish the standard raw Cherednik operators from the normalization appropriate to the Macdonald operators used here.  Define
\begin{equation}
  Y_i^{(0)}
  =\mathcal T_i\mathcal T_{i+1}\cdots\mathcal T_{N-1}\,
  \omega\,
  \mathcal T_1^{-1}\cdots\mathcal T_{i-1}^{-1},
  \qquad 1\leq i\leq N,
  \label{eq:cherednik-Y-raw}
\end{equation}
and
\begin{equation}
  \boxed{Y_i=\tau^{N-1}Y_i^{(0)}.}
  \label{eq:cherednik-Y}
\end{equation}
The $Y_i$ commute.  Since
\begin{equation}
  Y_i^{(0)}(1)=\tau^{N-2i+1},
\end{equation}
the normalized operators satisfy
\begin{equation}
  Y_i(1)=\mathfrak t^{N-i}.
  \label{eq:Y-on-vacuum}
\end{equation}
Without the common factor in \eqref{eq:cherednik-Y}, the elementary symmetric functions of the $Y_i$ do not have the normalization of the Macdonald--Ruijsenaars operators below.

The finite-Hecke spherical idempotent is defined as follows.  If $\mathcal T_w$ denotes the product along any reduced expression of $w\in S_N$, set
\begin{equation}
  \mathsf P_N(\mathfrak t)
  =\sum_{w\in S_N}\mathfrak t^{\ell(w)}
  =\prod_{j=1}^{N}\frac{1-\mathfrak t^j}{1-\mathfrak t}
\end{equation}
and
\begin{equation}
  \mathbf e_+
  =\frac{1}{\mathsf P_N(\mathfrak t)}
  \sum_{w\in S_N}\tau^{\ell(w)}\mathcal T_w .
  \label{eq:spherical-idempotent}
\end{equation}
It obeys
\begin{equation}
  \mathbf e_+^2=\mathbf e_+,
  \qquad
  \mathcal T_i\mathbf e_+
  =\mathbf e_+\mathcal T_i
  =\tau\mathbf e_+.
  \label{eq:finite-Hecke-idempotent-properties}
\end{equation}
Its image is $\mathscr P_N^{S_N}$, and $\mathbf e_+f=f$ for every symmetric Laurent polynomial $f$.  By contrast,
\begin{equation}
  \operatorname{Av}_W
  =\frac{1}{N!}\sum_{w\in S_N}w
  \label{eq:ordinary-Weyl-average}
\end{equation}
is the ordinary Weyl projector used in the gauge integral.  The two operators have the same image and both restrict to the identity on $\mathscr P_N^{S_N}$, but $\operatorname{Av}_W\neq\mathbf e_+$ for generic $\mathfrak t$.

For $0\leq k\leq N$, define
\begin{equation}
  \mathcal D_k
  =\mathfrak t^{k(k-1)/2}
  \sum_{\substack{I\subset\{1,\ldots,N\}\\ |I|=k}}
  \prod_{\substack{i\in I\\ j\notin I}}
  \frac{\mathfrak t x_i-x_j}{x_i-x_j}
  \prod_{i\in I}T_{q,x_i},
  \qquad \mathcal D_0=1,
  \label{eq:Macdonald-Ruijsenaars-Dk}
\end{equation}
where
\begin{equation}
  T_{q,x_i}f(x_1,\ldots,x_i,\ldots,x_N)
  =f(x_1,\ldots,qx_i,\ldots,x_N).
\end{equation}
We use the convention $\mathcal D_k=0$ for $k>N$.  With the normalization \eqref{eq:cherednik-Y}, the spherical identity is
\begin{equation}
  \boxed{
  \mathbf e_+e_k(Y_1,\ldots,Y_N)\mathbf e_+
  =\mathcal D_k\mathbf e_+ .
  }
  \label{eq:spherical-ekY-Dk}
\end{equation}
Equivalently, $e_k(Y)$ restricts to $\mathcal D_k$ on $\mathscr P_N^{S_N}$.  The finite-Hecke idempotent in \eqref{eq:spherical-ekY-Dk} and the ordinary Weyl average in \eqref{eq:ordinary-Weyl-average} therefore play distinct roles, even though the final restricted difference operator acts on ordinary symmetric Laurent polynomials.

The normalization can be checked directly in rank two.  For $N=2$,
\begin{equation}
  Y_1=\tau\mathcal T_1\omega,
  \qquad
  Y_2=\tau\omega\mathcal T_1^{-1}.
\end{equation}
If $f(x_1,x_2)$ is symmetric, then
\begin{align}
  Y_1f
  &=
  \frac{\mathfrak t x_1-x_2}{x_1-x_2}T_{q,x_1}f
  +\frac{(1-\mathfrak t)x_2}{x_1-x_2}T_{q,x_2}f,
  \nonumber\\
  Y_2f&=T_{q,x_2}f.
  \label{eq:N2-Y-actions}
\end{align}
Consequently,
\begin{align}
  (Y_1+Y_2)f
  &=
  \frac{\mathfrak t x_1-x_2}{x_1-x_2}T_{q,x_1}f
  +\frac{\mathfrak t x_2-x_1}{x_2-x_1}T_{q,x_2}f
  =\mathcal D_1f,
  \label{eq:N2-D1-check}\\
  Y_1Y_2f
  &=\mathfrak t T_{q,x_1}T_{q,x_2}f
  =\mathcal D_2f.
  \label{eq:N2-D2-check}
\end{align}
In particular,
\begin{equation}
  (Y_1+Y_2)(1)=1+\mathfrak t,
  \qquad
  Y_1Y_2(1)=\mathfrak t.
  \label{eq:N2-vacuum-check}
\end{equation}
The accompanying symbolic verification is supplied in
\texttt{verify\_daha\_n2.py}.

More generally,
\begin{equation}
  \mathsf C_N:=Y_1\cdots Y_N
  =\mathfrak t^{N(N-1)/2}T_{q,x_1}\cdots T_{q,x_N}
  =\mathcal D_N.
  \label{eq:determinant-translation}
\end{equation}
The operator $\mathsf C_N$ is the determinant translation and commutes with the symmetric $Y$-operators.

Let
\begin{equation}
  \Lambda_N^+
  =\{\mu\in\mathbb Z^N:\mu_1\geq\mu_2\geq\cdots\geq\mu_N\}.
\end{equation}
For $\mu\in\Lambda_N^+$, the Laurent Macdonald polynomial is defined by
\begin{equation}
  P_\mu(\mathbf x;q,\mathfrak t)
  =(x_1\cdots x_N)^{\mu_N}
  P_{\mu-\mu_N\mathbf 1}(\mathbf x;q,\mathfrak t).
  \label{eq:Laurent-Macdonald-convention}
\end{equation}
Put
\begin{equation}
  \boldsymbol\xi_\mu
  =\bigl(q^{\mu_1}\mathfrak t^{N-1},
  q^{\mu_2}\mathfrak t^{N-2},\ldots,q^{\mu_N}\bigr).
  \label{eq:spectral-point}
\end{equation}
Then
\begin{equation}
  \mathcal D_kP_\mu=e_k(\boldsymbol\xi_\mu)P_\mu.
  \label{eq:Dk-eigenvalue}
\end{equation}
The $P_\mu$ form an orthogonal basis for the pairing \eqref{eq:normalized-macdonald-product}.  It follows algebraically that each $\mathcal D_k$ is formally self-adjoint:
\begin{equation}
  \langle\mathcal D_k f,g\rangle_{N;q,\mathfrak t}
  =\langle f,\mathcal D_k g\rangle_{N;q,\mathfrak t}.
  \label{eq:Dk-formal-self-adjointness}
\end{equation}
Here self-adjointness refers to the bilinear Macdonald pairing.

\subsection{The one-row spherical operator}
\label{subsec:one-row-thooft-operator}

The character of $\operatorname{Sym}^r\mathbb C^N$ is the complete symmetric function $h_r(X_1,\ldots,X_N)$.  Its spherical $Y$-counterpart is
\begin{equation}
  \mathsf H^{\mathrm{DAHA}}_{(r)}
  :=\mathbf e_+h_r(Y_1,\ldots,Y_N)\mathbf e_+ .
  \label{eq:def-one-row-thooft-operator}
\end{equation}
Its identification with the boundary 't Hooft operator is Conjecture~\ref{conj:boundary-localization-section4}.

The elementary and complete symmetric functions satisfy
\begin{equation}
  \sum_{r\geq0}u^rh_r(Y)
  =\prod_{i=1}^{N}\frac{1}{1-uY_i}
  =\left(\sum_{k=0}^{N}(-u)^ke_k(Y)\right)^{-1}.
\end{equation}
Using \eqref{eq:spherical-ekY-Dk} gives
\begin{equation}
  \boxed{
  \sum_{r\geq0}u^r\mathsf H^{\mathrm{DAHA}}_{(r)}
  =\left(\sum_{k=0}^{N}(-u)^k\mathcal D_k\right)^{-1}
  }
  \qquad\text{on }\mathscr P_N^{S_N}.
  \label{eq:generating-function-thooft}
\end{equation}
The inverse is a formal series in $u$; its coefficient at every fixed order is a finite $q$-difference operator.  In particular,
\begin{align}
  \mathsf H^{\mathrm{DAHA}}_{(0)}&=1,
  \nonumber\\
  \mathsf H^{\mathrm{DAHA}}_{(1)}&=\mathcal D_1,
  \nonumber\\
  \mathsf H^{\mathrm{DAHA}}_{(2)}&=\mathcal D_1^2-\mathcal D_2,
  \label{eq:T0T1T2}\\
  \mathsf H^{\mathrm{DAHA}}_{(3)}&=\mathcal D_1^3-2\mathcal D_1\mathcal D_2+\mathcal D_3.
  \nonumber
\end{align}
Since the $\mathcal D_k$ commute and are formally self-adjoint,
\begin{equation}
  \langle\mathsf H^{\mathrm{DAHA}}_{(r)}f,g\rangle_{N;q,\mathfrak t}
  =\langle f,\mathsf H^{\mathrm{DAHA}}_{(r)}g\rangle_{N;q,\mathfrak t}.
  \label{eq:Tr-formal-self-adjointness}
\end{equation}

For a composition $\gamma=(\gamma_1,\ldots,\gamma_N)\in\mathbb Z_{\geq0}^N$, write
\begin{equation}
  T_{q^\gamma}
  =T_{q,x_1}^{\gamma_1}\cdots T_{q,x_N}^{\gamma_N},
  \qquad |\gamma|=\sum_i\gamma_i,
\end{equation}
and let $\gamma^+$ be the partition obtained by arranging its entries in non-increasing order.

\begin{proposition}
\label{prop:finite-shift-expansion}
For every $r\geq0$,
\begin{equation}
  \mathsf H^{\mathrm{DAHA}}_{(r)}
  =\sum_{\substack{\gamma\in\mathbb Z_{\geq0}^{N}\\|\gamma|=r}}
  \mathsf C^{\mathrm{DAHA}}_{r,\gamma}(\mathbf x;q,\mathfrak t)T_{q^\gamma}.
  \label{eq:finite-shift-expansion}
\end{equation}
The coefficients are rational functions whose possible poles lie on a finite collection of affine root hyperplanes
\begin{equation}
  q^m x_i=x_j,
  \qquad i\neq j,
  \label{eq:affine-root-poles}
\end{equation}
with $m\in\mathbb Z$ in a finite range determined by $r$.  They obey
\begin{equation}
  \mathsf C^{\mathrm{DAHA}}_{r,w\gamma}(\mathbf x;q,\mathfrak t)
  =\mathsf C^{\mathrm{DAHA}}_{r,\gamma}(w^{-1}\mathbf x;q,\mathfrak t),
  \qquad w\in S_N.
  \label{eq:shift-coeff-Weyl-covariance}
\end{equation}
Grouping the monomial shifts into Weyl orbits gives
\begin{equation}
  \mathsf H^{\mathrm{DAHA}}_{(r)}
  =\sum_{\substack{v\vdash r\\\ell(v)\leq N}}
  \mathsf H^{\mathrm{DAHA}}_{(r)\to v},
  \qquad
  \mathsf H^{\mathrm{DAHA}}_{(r)\to v}
  =\sum_{\gamma^+=v}
  \mathsf C^{\mathrm{DAHA}}_{r,\gamma}(\mathbf x;q,\mathfrak t)T_{q^\gamma}.
  \label{eq:sector-decomposition}
\end{equation}
\end{proposition}

\begin{proof}
Each $\mathcal D_k$ is a finite sum of shifts $T_{q^\epsilon}$, with $\epsilon_i\in\{0,1\}$ and $|\epsilon|=k$.  The expression of $h_r$ in the elementary symmetric functions is homogeneous of degree $r$ when $e_k$ has degree $k$.  Consequently every term contributing to the coefficient of $u^r$ in \eqref{eq:generating-function-thooft} has a shift $T_{q^\gamma}$ with $|\gamma|=r$, and there are only finitely many such compositions.  When the difference operators are composed, their rational coefficients are evaluated at finitely many $q$-translates of $\mathbf x$; this gives the possible divisors \eqref{eq:affine-root-poles}.  Weyl covariance follows from the Weyl invariance of the $\mathcal D_k$.  The Weyl orbits of the compositions are indexed by $v=\gamma^+$, which gives \eqref{eq:sector-decomposition}.
\end{proof}

If the spherical operator is identified with a boundary 't Hooft operator, the top orbit $v=(r)$ in \eqref{eq:sector-decomposition} has the bare magnetic shift and the lower orbits are the candidate screened shifts.

For later use, define
\begin{equation}
  A_i(\mathbf x)
  =\prod_{j\neq i}\frac{\mathfrak t x_i-x_j}{x_i-x_j}.
  \label{eq:Ai-coefficient}
\end{equation}
The coefficient of $T_{q,x_i}^r$ is
\begin{equation}
  \mathsf C^{\mathrm{DAHA}}_{r,r\varepsilon_i}(\mathbf x;q,\mathfrak t)
  =\prod_{m=0}^{r-1}A_i(x_1,\ldots,q^m x_i,\ldots,x_N)
  =\prod_{m=0}^{r-1}\prod_{j\neq i}
  \frac{\mathfrak t q^m x_i-x_j}{q^m x_i-x_j}.
  \label{eq:top-sector-coeff}
\end{equation}
No term involving $\mathcal D_k$ with $k>1$ can contribute to this shift, since such a term shifts at least two distinct variables.

For $r=2$, put
\begin{equation}
  B_{ij}(\mathbf x)
  =\prod_{\substack{a\in\{i,j\}\\b\notin\{i,j\}}}
  \frac{\mathfrak t x_a-x_b}{x_a-x_b},
  \qquad i<j.
  \label{eq:Bij-coefficient}
\end{equation}
Expanding $\mathcal D_1^2-\mathcal D_2$ gives
\begin{align}
  \mathsf H^{\mathrm{DAHA}}_{(2)}
  ={}&\sum_i
  A_i(\mathbf x)A_i(x_1,\ldots,qx_i,\ldots,x_N)T_{q,x_i}^2
  \nonumber\\
  &+\sum_{i<j}
  \bigl[
  A_i(\mathbf x)A_j(x_1,\ldots,qx_i,\ldots,x_N)
  \nonumber\\
  &\hspace{2.2cm}
  +A_j(\mathbf x)A_i(x_1,\ldots,qx_j,\ldots,x_N)
  -\mathfrak t B_{ij}(\mathbf x)
  \bigr]T_{q,x_i}T_{q,x_j}.
  \label{eq:r2-shift-expansion}
\end{align}
In the second coefficient the first two terms arise from $\mathcal D_1^2$ and the term $-\mathfrak t B_{ij}$ from $-\mathcal D_2$.

\subsection{The discrete Macdonald spectral theorem}
\label{subsec:Macdonald-kernel-identity}

Only the discrete $N$-variable Macdonald spectrum is required below, and the spectral result is stated directly in the Macdonald basis.  The Cauchy kernel is the reproducing kernel of the stable Macdonald scalar product and an intertwining kernel.  In equal rank it obeys the corresponding bispectral relation between the $x$- and $y$-difference operators, and it does not satisfy $\mathcal D_k^{(\mathbf x)}\Pi=e_k(\mathbf y)\Pi$ for generic $\mathbf y$.

\begin{theorem}[Discrete Macdonald spectrum]
\label{thm:one-row-kernel}
For every $r\geq0$ and every $\mu\in\Lambda_N^+$,
\begin{equation}
  \boxed{
  \mathsf H^{\mathrm{DAHA}}_{(r)}P_\mu(\mathbf x;q,\mathfrak t)
  =h_r(\boldsymbol\xi_\mu)P_\mu(\mathbf x;q,\mathfrak t).
  }
  \label{eq:one-row-P-eigenvalue}
\end{equation}
For symmetric polynomials it is enough to take $\mu$ to be a partition; the full set $\Lambda_N^+$ is required for symmetric Laurent polynomials.
\end{theorem}

\begin{proof}
Applying \eqref{eq:generating-function-thooft} to $P_\mu$ and using \eqref{eq:Dk-eigenvalue} gives
\begin{align}
  \left(\sum_{r\geq0}u^r\mathsf H^{\mathrm{DAHA}}_{(r)}\right)P_\mu
  &=\left(\sum_{k=0}^{N}(-u)^ke_k(\boldsymbol\xi_\mu)\right)^{-1}P_\mu
  \nonumber\\
  &=\prod_{i=1}^{N}\frac{1}{1-u\xi_{\mu,i}}P_\mu
  \nonumber\\
  &=\left(\sum_{r\geq0}u^rh_r(\boldsymbol\xi_\mu)\right)P_\mu.
  \label{eq:kernel-generating-proof}
\end{align}
Comparison of the coefficients of $u^r$ proves \eqref{eq:one-row-P-eigenvalue}.
\end{proof}

Define the finite-$N$ Macdonald transform by
\begin{equation}
  \mathsf F_N[f](\mu)
  =\frac{\langle f,P_\mu\rangle_{N;q,\mathfrak t}}
  {\mathcal N_\mu^{(N)}(q,\mathfrak t)},
  \qquad \mu\in\Lambda_N^+.
  \label{eq:Macdonald-transform}
\end{equation}
For a symmetric Laurent polynomial,
\begin{equation}
  f(\mathbf x)
  =\sum_{\mu\in\Lambda_N^+}\mathsf F_N[f](\mu)P_\mu(\mathbf x;q,\mathfrak t),
  \label{eq:Macdonald-transform-inversion}
\end{equation}
where the sum is finite; the same formula holds in the chosen completion for formal series.  Formal self-adjointness and Theorem~\ref{thm:one-row-kernel} imply
\begin{align}
  \mathsf F_N\!\left[\mathsf H^{\mathrm{DAHA}}_{(r)}f\right](\mu)
  &=\frac{\langle\mathsf H^{\mathrm{DAHA}}_{(r)}f,P_\mu\rangle_{N;q,\mathfrak t}}
  {\mathcal N_\mu^{(N)}}
  \nonumber\\
  &=\frac{\langle f,\mathsf H^{\mathrm{DAHA}}_{(r)}P_\mu\rangle_{N;q,\mathfrak t}}
  {\mathcal N_\mu^{(N)}}
  \nonumber\\
  &=h_r(\boldsymbol\xi_\mu)\mathsf F_N[f](\mu).
  \label{eq:transform-diagonalizes-thooft}
\end{align}

The shift and spectral expansions are different resolutions of the same operator.  The composition $\gamma$ in \eqref{eq:finite-shift-expansion} labels a monomial magnetic shift, whereas $\mu$ in \eqref{eq:Macdonald-transform-inversion} labels a Macdonald eigenchannel.  At the level of operators there is no termwise identification between these labels, even when both are represented by partitions of $r$.  Section~\ref{sec:eval-module} shows that they do coincide in the matrix element at the point $\boldsymbol\xi_0$.

\paragraph{Fourier orientations.}
At the level of the two commuting $X$- and $Y$-tori, the two Fourier
orientations are most transparent for the raw operators
\eqref{eq:cherednik-Y-raw}.  At this torus level, write
\begin{equation}
  \sigma_+:
  \quad X_i\longmapsto (Y_i^{(0)})^{-1},
  \qquad Y_i^{(0)}\longmapsto X_i,
  \label{eq:DAHA-Fourier-plus}
\end{equation}
and
\begin{equation}
  \sigma_-:=\sigma_+^{-1}:
  \quad X_i\longmapsto Y_i^{(0)},
  \qquad Y_i^{(0)}\longmapsto X_i^{-1}.
  \label{eq:DAHA-Fourier-minus}
\end{equation}
Thus $\sigma_+$ and $\sigma_-$ are inverse to each other, and their squares act by inversion on the raw torus generators.

Let $\alpha=\tau^{N-1}$, so that $Y_i=\alpha Y_i^{(0)}$.  In the normalized convention,
\begin{align}
  \sigma_+(X_i)&=\alpha Y_i^{-1},
  &\sigma_+(Y_i)&=\alpha X_i,
  \nonumber\\
  \sigma_-(X_i)&=\alpha^{-1}Y_i,
  &\sigma_-(Y_i)&=\alpha X_i^{-1}.
  \label{eq:normalized-Fourier-actions}
\end{align}
The symbols $\sigma_\pm$ in \eqref{eq:DAHA-Fourier-plus}--
\eqref{eq:normalized-Fourier-actions} specify only these torus-level
substitutions.  A full DAHA Fourier automorphism also specifies the action on
the finite-Hecke generators and on the parameters; that additional structure
is not required here.  With these torus substitutions, homogeneity gives
\begin{equation}
  \mathbf e_+h_r(Y)\mathbf e_+
  =\alpha^r\mathbf e_+\sigma_-\!\left(h_r(X)\right)\mathbf e_+,
  \label{eq:X-Y-duality-principle}
\end{equation}
whereas the opposite orientation satisfies
\begin{equation}
  \mathsf H^{\mathrm{DAHA}}_{(r)^\vee}
  :=\mathbf e_+h_r(Y^{-1})\mathbf e_+
  =\alpha^{-r}\mathbf e_+\sigma_+\!\left(h_r(X)\right)\mathbf e_+.
  \label{eq:reverse-Fourier-normalization}
\end{equation}
Here $h_r(Y^{-1})$ is the character of the dual one-row representation.  The homogeneous factors $\alpha^{\pm r}$ cancel in the product of the two oriented operators.  In the convention $\mathfrak t=q^{1/2}u_R^{-2}$, the replacement $u_R\mapsto u_R^{-1}$, equivalently $\mathfrak t\mapsto q/\mathfrak t$, is part of the physical boundary S-duality input and is separate from either DAHA Fourier map.  Both oriented operators are written in the electric parameter $\mathfrak t$.

\subsection{The evaluation module and the Wilson two-point function}
\label{sec:eval-module}

Write $H_r=\mathsf H^{\mathrm{DAHA}}_{(r)}$ and $H_r^\vee=\mathsf H^{\mathrm{DAHA}}_{(r)^\vee}$, and let $\rho=(N-1,N-2,\ldots,0)$, so that $\boldsymbol\xi_0=\mathfrak t^\rho$ in \eqref{eq:spectral-point}.  Since $e_k(Y^{-1})=e_{N-k}(Y)e_N(Y)^{-1}$, Theorem~\ref{thm:one-row-kernel} extends to the opposite orientation, $H_r^\vee P_\mu=h_r(\boldsymbol\xi_\mu^{-1})P_\mu$.  For the constant function $1=P_0$ this gives
\begin{equation}
  H_r\,1=h_r(\mathfrak t^\rho)\,1,
  \qquad
  H_r^\vee\,1=h_r(\mathfrak t^{-\rho})\,1 .
  \label{eq:constant-function-eigenvalues}
\end{equation}
A normalized matrix element of $H_r^\vee H_r$ between constant functions therefore equals $h_r(\mathfrak t^\rho)h_r(\mathfrak t^{-\rho})$ for every choice of pairing.  For $U(2)$ this is $(1+\mathfrak t)^2/\mathfrak t$ at $r=1$ and $(1+\mathfrak t+\mathfrak t^2)^2/\mathfrak t^2$ at $r=2$, whereas the Wilson two-point functions are $(1-q)(1+\mathfrak t)/(1-q\mathfrak t)$ and \eqref{eq:U2-r2-rational-section6}.  The constant function of the polynomial representation thus cannot represent the regular Nahm-pole boundary without line insertion.  The obstruction does not depend on the choice of vector.  No invertible linear map $U$ of $\mathscr P_N^{S_N}$ satisfies $U\,e_1(X)\,U^{-1}=\mathcal D_1$, because $w=U^{-1}1$ would obey $\bigl(e_1(X)-e_1(\mathfrak t^\rho)\bigr)w=0$ and hence vanish, $\mathscr P_N$ being an integral domain.  The exchange of $p(X)$ with $\mathbf e_+p(Y)\mathbf e_+$ therefore has to be realized on a different module.

\paragraph{The module.}
For $\nu\in\mathbb Z^N$ put $\boldsymbol\xi_\nu=(q^{\nu_1}\mathfrak t^{N-1},\ldots,q^{\nu_N})$ as in \eqref{eq:spectral-point}, so that $q^\gamma\boldsymbol\xi_\nu=\boldsymbol\xi_{\nu+\gamma}$.  For $\lambda\in\Lambda_N^+$ let $\mathrm{ev}_\lambda$ be the evaluation functional on $\mathscr P_N^{S_N}$ at $\boldsymbol\xi_\lambda$, and let $\mathcal M_{\mathrm{ev}}$ be the span of these functionals,
\begin{equation}
  \mathrm{ev}_\lambda(g)=g(\boldsymbol\xi_\lambda),
  \qquad
  \mathcal M_{\mathrm{ev}}
  =\sum_{\lambda\in\Lambda_N^+}\mathbb C(q,\tau)\,\mathrm{ev}_\lambda .
  \label{eq:eval-module}
\end{equation}
The $\mathrm{ev}_\lambda$ are linearly independent (Appendix~\ref{app:evaluation-proof}).  A difference operator $D$ on $\mathscr P_N^{S_N}$ acts on functionals by composition, $(\mathrm{ev}_\lambda D)(g)=(Dg)(\boldsymbol\xi_\lambda)$.  Let $\mathcal A_N$ be the commutative algebra generated by $\mathcal D_1,\ldots,\mathcal D_N$ and $\mathcal D_N^{-1}$.  For $D=\sum_\gamma C_\gamma(\mathbf x)T_{q^\gamma}$ in $\mathcal A_N$ one has
\begin{equation}
  \mathrm{ev}_\lambda D
  =\sum_{\gamma:\ \lambda+\gamma\in\Lambda_N^+}
  C_\gamma(\boldsymbol\xi_\lambda)\,\mathrm{ev}_{\lambda+\gamma}.
  \label{eq:eval-right-action}
\end{equation}
The coefficients are regular at the points $\boldsymbol\xi_\lambda$, and those with $\lambda+\gamma\notin\Lambda_N^+$ vanish there.  For $\mathcal D_k$ this is seen from \eqref{eq:Macdonald-Ruijsenaars-Dk}.  Write $\varepsilon_I=\sum_{i\in I}\varepsilon_i$.  If $\lambda+\varepsilon_I$ is not dominant, there is an $i$ with $\lambda_i=\lambda_{i+1}$, $i\notin I$ and $i+1\in I$, and the factor $\mathfrak t x_{i+1}-x_i$ of the coefficient of $T_{q^{\varepsilon_I}}$ vanishes at $\boldsymbol\xi_\lambda$.  The general case follows by induction on products (Appendix~\ref{app:evaluation-proof}).

Thus $\mathcal M_{\mathrm{ev}}$ is a right module for $\mathcal A_N$.  Multiplication by a symmetric Laurent polynomial $p$ acts on the same functionals diagonally, $\mathrm{ev}_\lambda\,p(X)=p(\boldsymbol\xi_\lambda)\,\mathrm{ev}_\lambda$.  The action \eqref{eq:eval-right-action} uses the same coefficient functions $C_\gamma$ as the action of $D$ on $\mathscr P_N^{S_N}$, evaluated at the points $\boldsymbol\xi_\lambda$.  No boundary wavefunction is absorbed into them.  If all operators are conjugated by the same multiplication operator as in \eqref{eq:conjugation-by-multiplication-section4}, the individual coefficients change, but the coefficient of the identity shift in a product does not, because the factors $\Psi(q^\gamma\mathbf x)/\Psi(\mathbf x)$ cancel when the total shift vanishes.

\paragraph{The opposite orientation.}
By \eqref{eq:spherical-ekY-Dk} and \eqref{eq:determinant-translation}, $e_k(Y^{-1})$ restricts to $\mathcal D_N^{-1}\mathcal D_{N-k}$ on $\mathscr P_N^{S_N}$.  Explicitly,
\begin{equation}
  \mathcal D_N^{-1}\mathcal D_{N-k}
  =\mathfrak t^{-k(k-1)/2}
  \sum_{\substack{I\subset\{1,\ldots,N\}\\ |I|=k}}
  \prod_{\substack{i\in I\\ j\notin I}}
  \frac{\mathfrak t^{-1}x_i-x_j}{x_i-x_j}
  \prod_{i\in I}T_{q,x_i}^{-1},
  \label{eq:dual-Macdonald-operators}
\end{equation}
which is $\mathcal D_k$ with $(q,\mathfrak t)$ replaced by $(q^{-1},\mathfrak t^{-1})$.  Hence $H_r^\vee$ belongs to $\mathcal A_N$ and
\begin{equation}
  H_r^\vee
  =\sum_{\substack{\gamma\in\mathbb Z_{\geq0}^N\\ |\gamma|=r}}
  \mathsf C^{\vee}_{r,\gamma}(\mathbf x;q,\mathfrak t)\,T_{q^{-\gamma}},
  \qquad
  \mathsf C^{\vee}_{r,\gamma}(\mathbf x;q,\mathfrak t)
  =\mathsf C^{\mathrm{DAHA}}_{r,\gamma}(\mathbf x;q^{-1},\mathfrak t^{-1}).
  \label{eq:dual-shift-coefficients}
\end{equation}
For $N=2$ one also has $H_r^\vee=\mathcal D_2^{-r}H_r$.

\paragraph{The intertwiner.}
The principal specialization $P_\lambda(\boldsymbol\xi_0)$ is non-zero.  Put $\widetilde P_\lambda=P_\lambda/P_\lambda(\boldsymbol\xi_0)$ for $\lambda\in\Lambda_N^+$.  Recall from \eqref{eq:one-row-transition} that $h_r=\sum_vQ_v[\mathbb A_{q,\mathfrak t}]P_v$.

\begin{proposition}[Evaluation module]
\label{prop:evaluation-module}
The following statements hold.
\begin{enumerate}
\item[(i)] The linear map $U_{\mathrm{ev}}:\mathscr P_N^{S_N}\to\mathcal M_{\mathrm{ev}}$ defined by $U_{\mathrm{ev}}\widetilde P_\lambda=\mathrm{ev}_\lambda$ is an isomorphism with $U_{\mathrm{ev}}(1)=\mathrm{ev}_0$.  For every symmetric Laurent polynomial $p$ and every $f\in\mathscr P_N^{S_N}$,
\begin{equation}
  U_{\mathrm{ev}}(p\,f)
  =U_{\mathrm{ev}}(f)\;\mathbf e_+p(Y)\mathbf e_+ .
  \label{eq:eval-intertwining}
\end{equation}
\item[(ii)] For every partition $v\vdash r$ with $\ell(v)\leq N$,
\begin{equation}
  \mathsf C^{\mathrm{DAHA}}_{r,v}(\boldsymbol\xi_0)
  =Q_v[\mathbb A_{q,\mathfrak t}]\,P_v(\boldsymbol\xi_0),
  \qquad
  \mathsf C^{\vee}_{r,v}(\boldsymbol\xi_v)
  =Q_v[\mathbb A_{q,\mathfrak t}]\,\frac{\mathcal N_v^{(N)}}{P_v(\boldsymbol\xi_0)},
  \label{eq:eval-coefficients}
\end{equation}
and $\mathsf C^{\mathrm{DAHA}}_{r,\gamma}(\boldsymbol\xi_0)=0$ if the composition $\gamma$ is not a partition.
\item[(iii)] For symmetric Laurent polynomials $p$ and $p'$, the coefficient of $\mathrm{ev}_0$ in
\begin{equation}
  \mathrm{ev}_0\,\bigl(\mathbf e_+p(Y)\mathbf e_+\bigr)
  \bigl(\mathbf e_+p'(Y^{-1})\mathbf e_+\bigr)
  \label{eq:eval-two-insertions}
\end{equation}
is the Neumann pairing $\langle p,p'\rangle_{N;q,\mathfrak t}$.  For $p=p'=h_r$ this coefficient is
\begin{equation}
\begin{split}
  \bigl[H_rH_r^\vee\bigr]_{T^0}(\boldsymbol\xi_0)
  &=\sum_{\substack{v\vdash r\\ \ell(v)\leq N}}
  \mathsf C^{\mathrm{DAHA}}_{r,v}(\boldsymbol\xi_0)\,
  \mathsf C^{\vee}_{r,v}(\boldsymbol\xi_v)
  \\
  &=\sum_{\substack{v\vdash r\\ \ell(v)\leq N}}
  Q_v[\mathbb A_{q,\mathfrak t}]^2\,\mathcal N_v^{(N)}
  =\widehat{\mathbb I}^{\mathrm{Neu}}_{(r)}(q,\mathfrak t),
\end{split}
  \label{eq:eval-two-point}
\end{equation}
where $[D]_{T^0}$ denotes the coefficient of the identity shift in $D$, called its zero-shift coefficient below.
\end{enumerate}
\end{proposition}

The proof is given in Appendix~\ref{app:evaluation-proof}.  It uses the evaluation symmetry and the orthogonality of Macdonald polynomials \cite{Macdonald:1995}, and holds for every $r$.  Part (i) expresses the known duality of Macdonald polynomials between the argument and the spectral label \cite{Macdonald:1995,Cherednik:2005} in the form needed for the shift coefficients.  The identities \eqref{eq:eval-coefficients} and \eqref{eq:eval-two-point} are also checked symbolically for $U(2)$ and with exact rational parameters for $U(3)$ in the ancillary file \texttt{verify\_evaluation\_module.py}.

Part (i) realizes on $\mathcal M_{\mathrm{ev}}$ the exchange of $p(X)$ with $\mathbf e_+p(Y)\mathbf e_+$ which is not possible on $\mathscr P_N^{S_N}$.  For instance $U_{\mathrm{ev}}(e_1)=\mathrm{ev}_0\mathcal D_1=e_1(\mathfrak t^\rho)\,\mathrm{ev}_{\varepsilon_1}$, so that $e_1-e_1(\mathfrak t^\rho)$ is mapped to $e_1(\mathfrak t^\rho)\,(\mathrm{ev}_{\varepsilon_1}-\mathrm{ev}_0)$, which is not zero.  Conversely, since $\mathbf e_+p(Y)\mathbf e_+$ has the eigenvalue $p(\boldsymbol\xi_\lambda)$ on $\widetilde P_\lambda$, the map $U_{\mathrm{ev}}$ takes it to multiplication by $p$ on $\mathcal M_{\mathrm{ev}}$.  Part (iii) accordingly holds for arbitrary $p$ and $p'$.

By \eqref{eq:eval-right-action} the coefficient of $\mathrm{ev}_0$ in $\mathrm{ev}_0D$ is $[D]_{T^0}(\boldsymbol\xi_0)$.  We call it the matrix element of $D$ at $\mathrm{ev}_0$.  The functional $\mathrm{ev}_0$ is not an eigenvector of the right action of $\mathcal A_N$.  The value of $\mathrm{ev}_0H_rH_r^\vee$ on the constant function is
\begin{equation}
  \bigl(\mathrm{ev}_0H_rH_r^\vee\bigr)(1)
  =h_r(\mathfrak t^\rho)\,h_r(\mathfrak t^{-\rho}),
  \label{eq:eval-constant-function}
\end{equation}
which is the quantity obtained from \eqref{eq:constant-function-eigenvalues}.  The matrix element \eqref{eq:eval-two-point} is a different quantity, in which the total shift of the two insertions vanishes.  In the sum \eqref{eq:eval-two-point} the operator $H_r$ takes $\mathrm{ev}_0$ to $\mathrm{ev}_v$ through its shift orbit $v$, and $H_r^\vee$ returns it.  The shift orbit $v$ therefore contributes exactly the summand $\mu=v$ of the spectral sum \eqref{eq:one-row-wilson-closed}.  This termwise correspondence is a property of the matrix element at $\boldsymbol\xi_0$.  On $\mathrm{ev}_\lambda$ with $\lambda\neq0$ one shift orbit may reach several points, and there is no such correspondence for the operators themselves.

\paragraph{The first two charges in $U(2)$.}
For $N=2$ we use the coefficients $A_1,A_2$ of \eqref{eq:U2-A1-A2-section4}, which depend only on $x_1/x_2$.  Then $\mathcal D_1=A_1T_{q,x_1}+A_2T_{q,x_2}$, $\mathcal D_2=\mathfrak t\,T_{q,x_1}T_{q,x_2}$ and $\boldsymbol\xi_0=(\mathfrak t,1)$.  The values needed below are
\begin{equation}
  A_1(\mathfrak t,1)=1+\mathfrak t,\quad
  A_2(\mathfrak t,1)=0,\quad
  A_1(q\mathfrak t,1)=\frac{1-q\mathfrak t^2}{1-q\mathfrak t},\quad
  A_2(q^m\mathfrak t,1)=\frac{\mathfrak t(1-q^m)}{1-q^m\mathfrak t}.
  \label{eq:U2-A-values}
\end{equation}
For the fundamental charge, $H_1=\mathcal D_1$ and $H_1^\vee=\mathcal D_2^{-1}\mathcal D_1=\mathfrak t^{-1}\bigl(A_2T_{q,x_1}^{-1}+A_1T_{q,x_2}^{-1}\bigr)$.  Since $A_2(\mathfrak t,1)=0$,
\begin{equation}
  \mathrm{ev}_0H_1=(1+\mathfrak t)\,\mathrm{ev}_{(1,0)},
  \qquad
  \mathrm{ev}_{(1,0)}H_1^\vee
  =\frac{1-q}{1-q\mathfrak t}\,\mathrm{ev}_0
  +\frac{1-q\mathfrak t^2}{\mathfrak t(1-q\mathfrak t)}\,\mathrm{ev}_{(1,-1)},
  \label{eq:eval-r1-steps}
\end{equation}
and therefore
\begin{equation}
  \bigl[H_1H_1^\vee\bigr]_{T^0}(\boldsymbol\xi_0)
  =\frac{(1-q)(1+\mathfrak t)}{1-q\mathfrak t}
  =\widehat{\mathbb I}^{\mathrm{Neu},U(2)}_{(1,0)}(q,\mathfrak t).
  \label{eq:eval-r1}
\end{equation}

For the second symmetric power, $H_2=\mathcal D_1^2-\mathcal D_2$ and $H_2^\vee=\mathcal D_2^{-2}H_2$.  The coefficient of $T_{q,x_2}^2$ in $H_2$ vanishes at $\boldsymbol\xi_0$, and the two remaining coefficients are
\begin{align}
  \mathsf C^{\mathrm{DAHA}}_{2,(2,0)}(\boldsymbol\xi_0)
  &=A_1(\mathfrak t,1)A_1(q\mathfrak t,1)
  =\frac{(1+\mathfrak t)(1-q\mathfrak t^2)}{1-q\mathfrak t},
  \nonumber\\
  \mathsf C^{\mathrm{DAHA}}_{2,(1,1)}(\boldsymbol\xi_0)
  &=R_{12}(\boldsymbol\xi_0)+R_{21}(\boldsymbol\xi_0)-\mathfrak t
  =\frac{\mathfrak t(1+\mathfrak t)(1-q)}{1-q\mathfrak t}+0-\mathfrak t
  =\mathfrak t\,\kappa_2,
  \label{eq:eval-r2-forward}
\end{align}
with $\kappa_2=(\mathfrak t-q)/(1-q\mathfrak t)$ as in \eqref{eq:kappa2-section6}.  Here $R_{12}(\boldsymbol\xi_0)=A_1(\mathfrak t,1)A_2(q\mathfrak t,1)$ and $R_{21}(\boldsymbol\xi_0)=A_2(\mathfrak t,1)A_1(\mathfrak t,q)=0$ are the two convolution residues of Section~\ref{sec:boundary-bubbling} at the point $\boldsymbol\xi_0$, and $-\mathfrak t$ is the subtraction of the determinant summand.  For the opposite orientation,
\begin{align}
  \mathsf C^{\vee}_{2,(2,0)}(\boldsymbol\xi_{(2,0)})
  &=\mathfrak t^{-2}A_2(q^2\mathfrak t,1)A_2(q\mathfrak t,1)
  =\frac{(1-q)(1-q^2)}{(1-q\mathfrak t)(1-q^2\mathfrak t)},
  \nonumber\\
  \mathsf C^{\vee}_{2,(1,1)}(\boldsymbol\xi_{(1,1)})
  &=\mathfrak t^{-2}\,\mathsf C^{\mathrm{DAHA}}_{2,(1,1)}(\boldsymbol\xi_0)
  =\frac{\kappa_2}{\mathfrak t},
  \label{eq:eval-r2-reverse}
\end{align}
where the second line uses $\boldsymbol\xi_{(1,1)}=q\,\boldsymbol\xi_0$.  Hence
\begin{equation}
\begin{split}
  \bigl[H_2H_2^\vee\bigr]_{T^0}(\boldsymbol\xi_0)
  &=\frac{(1+\mathfrak t)(1-q\mathfrak t^2)(1-q)(1-q^2)}
        {(1-q\mathfrak t)^2(1-q^2\mathfrak t)}
  +\kappa_2^2
  \\
  &=\mathcal M_2^{(2)}+\kappa_2^2
  =\widehat{\mathbb I}^{\mathrm{Neu},U(2)}_{(2,0)}(q,\mathfrak t),
\end{split}
  \label{eq:eval-r2}
\end{equation}
in agreement with \eqref{eq:U2-r2-wilson} and \eqref{eq:U2-r2-rational-section6}.  If the determinant summand is not subtracted, that is, if $H_2$ is replaced by $\mathcal D_1^2$ in both orientations, the second term becomes $(1+\kappa_2)^2$ and the result is the two-point function of the reducible character $e_1^2$.  The term $-\mathfrak t$ is thus required for the agreement with the symmetric Wilson correlator.

\paragraph{Relation to the regular Nahm-pole boundary.}
For minuscule charge $B=(1^k,0^{N-k})$ the points $\boldsymbol\xi_B$ are, up to a common factor and the order of the entries, the fugacities at which the Dirichlet half-index is Higgsed in Ref.~\cite{Hatsuda:2025bmn}.  By part (iii) with $p=p'=e_k$, the zero-shift coefficient of $\mathcal D_k\,\mathcal D_N^{-1}\mathcal D_{N-k}$ at $\boldsymbol\xi_0$ is $\langle e_k,e_k\rangle_{N;q,\mathfrak t}=\mathcal N^{(N)}_{(1^k,0^{N-k})}$, which by \eqref{eq:section2-higgsing-fundamental} agrees with the normalized regular Nahm-pole two-point function obtained there.  For $k=1$ and $N=2$ this is \eqref{eq:eval-r1}.  For the non-minuscule duals of the symmetric representations that reference notes that the Higgsing procedure does not apply directly.  Equation \eqref{eq:eval-two-point} shows that the matrix element of $H_rH_r^\vee$ at $\mathrm{ev}_0$, with all lower shifts of $H_r$ included, equals the Wilson two-point function of the S-dual representation for every $r$.  That this matrix element is the regular Nahm-pole two-point function of the non-minuscule line remains to be established.  This question is separate from Conjecture~\ref{conj:boundary-localization-section4}, which concerns the operator of a single line.

\subsection{Determinant shifts and general weights}
\label{subsec:U2-corollary-daha}

For a dominant integral weight $B=(B_1,\ldots,B_N)$, write
\begin{equation}
  \mathsf H^{\mathrm{DAHA}}_B
  :=\mathbf e_+s_B(Y)\mathbf e_+
\end{equation}
for the associated spherical operator, and set
\begin{equation}
  m=B_N,
  \qquad
  \bar B=B-m\mathbf 1,
  \qquad
  \bar B_N=0.
  \label{eq:determinant-reduced-charge}
\end{equation}
Since
\begin{equation}
  s_B(Y)=(Y_1\cdots Y_N)^m s_{\bar B}(Y),
\end{equation}
the spherical operator factorizes as
\begin{equation}
  \mathsf H^{\mathrm{DAHA}}_B
  =\mathsf C_N^m\mathsf H^{\mathrm{DAHA}}_{\bar B},
  \qquad
  \mathsf H^{\mathrm{DAHA}}_{B^\vee}
  =\mathsf C_N^{-m}\mathsf H^{\mathrm{DAHA}}_{\bar B^\vee}.
  \label{eq:general-determinant-factorization}
\end{equation}
All these operators are symmetric functions of the commuting $Y_i$.  The inverse determinant factors therefore cancel in the product of the two oriented operators:
\begin{equation}
  \mathsf H^{\mathrm{DAHA}}_{B^\vee}
  \mathsf H^{\mathrm{DAHA}}_{B}
  =\mathsf H^{\mathrm{DAHA}}_{\bar B^\vee}
  \mathsf H^{\mathrm{DAHA}}_{\bar B}.
  \label{eq:determinant-factor-cancellation}
\end{equation}

For $U(2)$, let $B=(a,b)$ with $a\geq b$.  Then
\begin{equation}
  B=b(1,1)+(a-b,0),
  \label{eq:U2-det-shift-decomp-section5}
\end{equation}
and
\begin{equation}
  \mathsf H^{\mathrm{DAHA}}_{(a,b)}
  =(Y_1Y_2)^b
  \mathsf H^{\mathrm{DAHA}}_{(a-b,0)}.
  \label{eq:U2-magnetic-det-shift}
\end{equation}

\begin{corollary}
\label{cor:U2-all-dominant}
For every dominant $U(2)$ charge $(a,b)$, the determinant factors cancel in the oriented two-line operator, leaving the one-row charge $a-b$, and
\begin{equation}
  \bigl[\mathsf H^{\mathrm{DAHA}}_{(a,b)}\mathsf H^{\mathrm{DAHA}}_{(a,b)^\vee}\bigr]_{T^0}(\boldsymbol\xi_0)
  =\widehat{\mathbb I}^{\mathrm{Neu},U(2)}_{(a,b)}(q,\mathfrak t)
  =\widehat{\mathbb I}^{\mathrm{Neu},U(2)}_{(a-b,0)}(q,\mathfrak t).
  \label{eq:U2-all-S-duality}
\end{equation}
\end{corollary}

\begin{proof}
Equation \eqref{eq:determinant-factor-cancellation} reduces the product of the two oriented operators to that of the determinant-reduced charge $(a-b,0)$.  The corresponding electric determinant factors cancel in the Macdonald pairing by \eqref{eq:determinant-invariance}.  The equality is then \eqref{eq:eval-two-point} for $r=a-b$.
\end{proof}

\subsubsection{General dominant weights}
\label{subsec:general-weights-role}

For a general dominant integral weight $\lambda$, define the spherical operator
\begin{equation}
  \mathsf H^{\mathrm{DAHA}}_{\lambda}
  =\mathbf e_+s_\lambda(Y_1,\ldots,Y_N)\mathbf e_+.
  \label{eq:general-lambda-DAHA-prescription}
\end{equation}
The same simultaneous diagonalization gives
\begin{equation}
  \mathsf H^{\mathrm{DAHA}}_{\lambda}
  P_\mu(\mathbf x;q,\mathfrak t)
  =s_\lambda(\boldsymbol\xi_\mu)P_\mu(\mathbf x;q,\mathfrak t).
  \label{eq:general-lambda-kernel}
\end{equation}
The identification of $\mathsf H^{\mathrm{DAHA}}_\lambda$ with the boundary localization operator for arbitrary $\lambda$ is Conjecture~\ref{conj:general-boundary-Hecke}.  The one-row family is distinguished because its full shift support is controlled by \eqref{eq:generating-function-thooft}, its top coefficient has the closed form \eqref{eq:top-sector-coeff}, and the first non-minuscule lower coefficient admits the local convolution/IC realization in ambient ranks two and three described in Section~\ref{sec:boundary-bubbling}.  By Proposition~\ref{prop:evaluation-module}(iii) with $p=p'=s_\lambda$, the zero-shift coefficient of $\mathsf H^{\mathrm{DAHA}}_\lambda\mathsf H^{\mathrm{DAHA}}_{\lambda^\vee}$ at $\boldsymbol\xi_0$ equals $\widehat{\mathbb I}^{\mathrm{Neu}}_\lambda(q,\mathfrak t)$ for every dominant $\lambda$.

\section{Examples, limits and checks}
\label{sec:examples-limits-checks}

Two finite decompositions are displayed in this section.  The DAHA operator has a shift-orbit expansion
\begin{equation}
  \mathsf H^{\mathrm{DAHA}}_{(r)}
  =
  \sum_{\substack{v\vdash r\\ \ell(v)\leq N}}
  \ \sum_{\gamma^+=v}
  \mathsf C^{\mathrm{DAHA}}_{r,\gamma}
  (\mathbf x;q,\mathfrak t)T_{q^\gamma},
  \label{eq:section6-shift-orbit-decomposition}
\end{equation}
where $v$ labels a Weyl orbit of shifts and has the quantum numbers of a candidate effective magnetic charge.  The Wilson two-point function instead has the Macdonald spectral expansion
\begin{equation}
  \widehat{\mathbb I}^{\mathrm{Neu}}_{(r)}
  =
  \sum_{\substack{\mu\vdash r\\ \ell(\mu)\leq N}}
  Q_\mu[\mathbb A_{q,\mathfrak t}]^2
  \mathcal N_\mu^{(N)}.
  \label{eq:section6-spectral-decomposition}
\end{equation}
Here $\mu$ labels a Macdonald eigenfunction.  An individual operator in \eqref{eq:section6-shift-orbit-decomposition} is in general not diagonal in the Macdonald basis, and the two decompositions are related through the matrix element \eqref{eq:eval-two-point}, in which the shift orbit $v$ contributes the summand $\mu=v$.  The examples below display both sets of data.

Throughout this section, the magnetic quantity is the zero-shift coefficient
$[H_rH_r^\vee]_{T^0}(\boldsymbol\xi_0)$ of
Proposition~\ref{prop:evaluation-module}.

\begin{table}[t]
\centering
\small
\begin{tabular}{p{0.16\textwidth} p{0.20\textwidth} p{0.18\textwidth} p{0.34\textwidth}}
\hline
Gauge/charge & DAHA shift orbits $v$ & Spectral labels $\mu$ & Wilson spectral pairing \\
\hline
$U(2),\ (2,0)$ & $(2,0),(1,1)$ & $(2),(1,1)$ & $\mathcal R_1(2)+\kappa_2^2$ \\
$U(2),\ (3,0)$ & $(3,0),(2,1)$ & $(3),(2,1)$ & $\mathcal R_1(3)+\kappa_3^2\mathcal R_1(1)$ \\
$U(3),\ (2,0,0)$ & $(2,0,0),(1,1,0)$ & $(2),(1,1)$ & $\mathcal R_1(2)\mathcal R_2(2)+\kappa_2^2\mathcal R_2(1)\mathcal R_1(1)$ \\
\hline
\end{tabular}
\caption{Low-rank DAHA shift-orbit data and Macdonald spectral data.  The factors $\mathcal R_s$ and the coefficients $\kappa_2,\kappa_3$ are defined in \eqref{eq:Rs-factor-section6}, \eqref{eq:kappa2-section6}, and \eqref{eq:kappa3-section6}.  By Proposition~\ref{prop:evaluation-module}, the summands in the last column are, in the order shown, the contributions of the shift orbits in the second column to $[H_rH_r^\vee]_{T^0}(\boldsymbol\xi_0)$.}
\label{tab:low-rank-sector-data}
\end{table}

For later use we isolate a rank-one factor which appears in the finite-$N$ Macdonald norms.  For $s\geq 1$ and $d\geq 0$ set
\begin{equation}
  \mathcal R_s(d;q,\mathfrak t)
  =\frac{(q^d\mathfrak t^s;q)_\infty(q^{d+1}\mathfrak t^s;q)_\infty}
        {(q^d\mathfrak t^{s+1};q)_\infty(q^{d+1}\mathfrak t^{s-1};q)_\infty}
   \frac{(\mathfrak t^{s+1};q)_\infty(q\mathfrak t^{s-1};q)_\infty}
        {(\mathfrak t^s;q)_\infty(q\mathfrak t^s;q)_\infty}.
  \label{eq:Rs-factor-section6}
\end{equation}
Then the norm formula \eqref{eq:expanded-normalized-norm} may be written as
\begin{equation}
  \mathcal N^{(N)}_\mu(q,\mathfrak t)
  =\prod_{1\leq i<j\leq N}
  \mathcal R_{j-i}(\mu_i-\mu_j;q,\mathfrak t).
  \label{eq:norm-R-factorization-section6}
\end{equation}
For $s=1$ this reduces to a finite product,
\begin{equation}
  \mathcal R_1(d;q,\mathfrak t)
  =\frac{(q;q)_d(\mathfrak t^2;q)_d}
        {(\mathfrak t;q)_d(q\mathfrak t;q)_d},
  \qquad d\geq0,
  \label{eq:R1-finite-product-section6}
\end{equation}
where $(a;q)_d=\prod_{m=0}^{d-1}(1-aq^m)$.  Equation \eqref{eq:R1-finite-product-section6} is obtained by cancelling the tails of the four infinite products in \eqref{eq:Rs-factor-section6}.  In particular
\begin{equation}
  \mathcal R_1(0;q,\mathfrak t)=1,
  \qquad
  \mathcal R_1(1;q,\mathfrak t)=\frac{(1-q)(1+\mathfrak t)}{1-q\mathfrak t}.
  \label{eq:R1-first-values-section6}
\end{equation}
The $U(2)$ norm used in Section~\ref{sec:wilson-side} is
\begin{equation}
  \mathcal M^{(2)}_d(q,\mathfrak t)=\mathcal R_1(d;q,\mathfrak t).
  \label{eq:M2-as-R1-section6}
\end{equation}

\subsection{The first non-minuscule charge in rank two}
\label{subsec:U2-r2-check-section6}

The charge $B=(2,0)$ in $U(2)$ is the smallest case in which bubbling can occur.  The dominance order gives two possible effective charges,
\begin{equation}
  (2,0)\longrightarrow (2,0),
  \qquad
  (2,0)\longrightarrow (1,1).
  \label{eq:U2-r2-sector-map-section6}
\end{equation}
The corresponding dominant shift orbits give a two-sector, rather than a two-shift, decomposition of
\begin{equation}
  \mathsf H^{\mathrm{DAHA}}_{(2)}=\mathcal D_1^2-\mathcal D_2.
  \label{eq:U2-r2-D-expression-section6}
\end{equation}
Since $\mathcal D_2=\mathfrak t T_{q,x_1}T_{q,x_2}$ in two variables, \eqref{eq:U2-r2-D-expression-section6} contains the top shifts $T_{q,x_1}^2,T_{q,x_2}^2$ and the lower shift $T_{q,x_1}T_{q,x_2}$.  With $y=x_1/x_2$ and $A_i$ as in \eqref{eq:U2-A1-A2-section4}, the lower coefficient is
\begin{equation}
  \mathsf C^{\mathrm{DAHA}}_{2,(1,1)}(y;q,\mathfrak t)
  =A_1(\mathbf x)A_2(qx_1,x_2)
   +A_2(\mathbf x)A_1(x_1,qx_2)-\mathfrak t.
  \label{eq:U2-r2-lower-coeff-short-section6}
\end{equation}
The first two terms arise from the two ordered products in $\mathcal D_1^2$, while the last term comes from $-\mathcal D_2$.  The resolved convolution $T^*\mathbb P^1$ has two torus fixed points, so the subtraction removes an IC summand rather than adding another fixed-point residue.  The calculation of Section~\ref{sec:boundary-bubbling} reproduces the first two terms from the two minuscule-convolution residues and the last term from subtracting the determinant intersection-cohomology/Satake summand.

The final Wilson two-point function follows from the spectral decomposition.  Put
\begin{equation}
  \kappa_2(q,\mathfrak t)=\frac{\mathfrak t-q}{1-q\mathfrak t}.
  \label{eq:kappa2-section6}
\end{equation}
Then \eqref{eq:U2-r2-wilson} gives
\begin{equation}
  \widehat{\mathbb I}^{\mathrm{Neu},U(2)}_{(2,0)}(q,\mathfrak t)
  =\mathcal R_1(2;q,\mathfrak t)+\kappa_2(q,\mathfrak t)^2
  .
  \label{eq:U2-r2-exact-check-section6}
\end{equation}
The two summands in \eqref{eq:U2-r2-exact-check-section6} are the spectral contributions of $\mu=(2)$ and $\mu=(1,1)$.  By \eqref{eq:eval-r2} they are also the contributions of the shift orbits $(2,0)$ and $(1,1)$ to the zero-shift coefficient $[H_2H_2^\vee]_{T^0}(\boldsymbol\xi_0)$.  Using \eqref{eq:R1-finite-product-section6}, the answer can be written as a single rational function,
\begin{equation}
  \widehat{\mathbb I}^{\mathrm{Neu},U(2)}_{(2,0)}
  =\frac{1+\mathfrak t+\mathfrak t^2-q(1+2\mathfrak t)
     -q^2(\mathfrak t^2+2\mathfrak t)+q^3(1+\mathfrak t+\mathfrak t^2)}
    {(1-q\mathfrak t)(1-q^2\mathfrak t)}.
  \label{eq:U2-r2-rational-section6}
\end{equation}
This is the correlator of Eq.~(3.12) of \cite{Hatsuda:2025bmn} after division by the Neumann vacuum half-index and the substitution $\mathfrak t=q^{1/2}u_R^{-2}$.  At $\mathfrak t=q$ the rational expression \eqref{eq:U2-r2-rational-section6} gives one.  At $q=0$ it gives
\begin{equation}
  \widehat{\mathbb I}^{\mathrm{Neu},U(2)}_{(2,0)}(0,\mathfrak t)=1+\mathfrak t+\mathfrak t^2,
  \label{eq:U2-r2-HL-value-section6}
\end{equation}
which is the Hall--Littlewood specialization of the spectral formula.

The role of the lower shift is visible on eigenstates.  For a determinant-reduced Macdonald label $\mu=(d,0)$, the full operator has eigenvalue
\begin{equation}
  h_2(q^d\mathfrak t,1)=q^{2d}\mathfrak t^2+q^d\mathfrak t+1.
  \label{eq:U2-r2-eigenvalue-section6}
\end{equation}
\subsection{The next rank-two charge}
\label{subsec:U2-r3-check-section6}

For $B=(3,0)$ the possible dominant effective charges are
\begin{equation}
  (3,0)\longrightarrow (3,0),
  \qquad
  (3,0)\longrightarrow (2,1).
  \label{eq:U2-r3-sector-map-section6}
\end{equation}
In two variables $\mathcal D_3$ vanishes, and the one-row DAHA operator is
\begin{equation}
  \mathsf H^{\mathrm{DAHA}}_{(3)}
  =\mathcal D_1^3-2\mathcal D_1\mathcal D_2.
  \label{eq:U2-r3-D-expression-section6}
\end{equation}
The coefficient of the top shift $T_{q,x_i}^3$ is
\begin{equation}
  \prod_{m=0}^{2}A_i(x_1,\ldots,q^m x_i,\ldots,x_N),
  \label{eq:U2-r3-top-coeff-section6}
\end{equation}
which is the specialization of \eqref{eq:top-sector-coeff}.  The lower coefficient of $T_{q,x_1}^2T_{q,x_2}$ is
\begin{align}
  \mathsf C^{\mathrm{DAHA}}_{3,(2,1)}(x_1,x_2;q,\mathfrak t)
  ={}&A_1(\mathbf x)A_1(qx_1,x_2)A_2(q^2x_1,x_2)
  \nonumber\\
  &+A_1(\mathbf x)A_2(qx_1,x_2)A_1(qx_1,qx_2)
  \nonumber\\
  &+A_2(\mathbf x)A_1(x_1,qx_2)A_1(qx_1,qx_2)
  -2\mathfrak t A_1(\mathbf x),
  \label{eq:U2-r3-lower-coeff-section6}
\end{align}
with the coefficient of $T_{q,x_1}T_{q,x_2}^2$ obtained by interchanging $x_1$ and $x_2$.  The three positive terms arise from ordered products in $\mathcal D_1^3$, whereas the subtraction comes from $-2\mathcal D_1\mathcal D_2$.

Let
\begin{equation}
  \kappa_3(q,\mathfrak t)=\frac{(1+q)(\mathfrak t-q)}{1-q^2\mathfrak t}.
  \label{eq:kappa3-section6}
\end{equation}
The Wilson pairing is
\begin{equation}
  \widehat{\mathbb I}^{\mathrm{Neu},U(2)}_{(3,0)}(q,\mathfrak t)
  =\mathcal R_1(3;q,\mathfrak t)+\kappa_3(q,\mathfrak t)^2\mathcal R_1(1;q,\mathfrak t)
  .
  \label{eq:U2-r3-exact-check-section6}
\end{equation}
The Macdonald spectral channels are $\mu=(3)$ and $\mu=(2,1)$, and \eqref{eq:U2-r3-exact-check-section6} is the correlator of Eq.~(3.13) of \cite{Hatsuda:2025bmn} in the same normalization.  The finite product form is
\begin{equation}
  \mathcal R_1(3;q,\mathfrak t)
  =\frac{(q;q)_3(\mathfrak t^2;q)_3}
        {(\mathfrak t;q)_3(q\mathfrak t;q)_3}.
  \label{eq:R13-section6}
\end{equation}
At $\mathfrak t=q$ the lower spectral coefficient $\kappa_3$ vanishes and $\mathcal R_1(3;q,q)=1$, so the normalized Wilson pairing is one.  At $q=0$ one obtains
\begin{equation}
  \widehat{\mathbb I}^{\mathrm{Neu},U(2)}_{(3,0)}(0,\mathfrak t)
  =(1+\mathfrak t)(1+\mathfrak t^2),
  \label{eq:U2-r3-HL-value-section6}
\end{equation}
which is the Hall--Littlewood limit of the spectral sum over the two partitions of three with at most two parts.

The determinant-translation identity of Section~\ref{subsec:U2-corollary-daha} gives
\begin{equation}
  \widehat{\mathbb I}^{\mathrm{Neu},U(2)}_{(a,b)}(q,\mathfrak t)
  =\widehat{\mathbb I}^{\mathrm{Neu},U(2)}_{(a-b,0)}(q,\mathfrak t),
  \qquad a\geq b.
  \label{eq:U2-determinant-check-section6}
\end{equation}
The determinant factorization gives the same identity for the zero-shift
coefficient of Corollary~\ref{cor:U2-all-dominant}.

\subsection{Rank-three spectral check}
\label{subsec:U3-r2-check-section6}

For $G=U(3)$ and $B=(2,0,0)$ the possible dominant effective charges are
\begin{equation}
  (2,0,0)\longrightarrow (2,0,0),
  \qquad
  (2,0,0)\longrightarrow (1,1,0).
  \label{eq:U3-r2-sector-map-section6}
\end{equation}
The lower DAHA shift orbit is non-central.  The spectral label $\mu=(1,1,0)$ has the non-trivial finite-$N$ norm displayed below.

The norm factorization \eqref{eq:norm-R-factorization-section6} gives
\begin{equation}
  \mathcal N^{(3)}_{(a,b,0)}(q,\mathfrak t)
  =\mathcal R_1(a-b;q,\mathfrak t)
   \mathcal R_2(a;q,\mathfrak t)
   \mathcal R_1(b;q,\mathfrak t).
  \label{eq:U3-norm-factorization-section6}
\end{equation}
Therefore
\begin{align}
  \mathcal N^{(3)}_{(2,0,0)}
  &=\mathcal R_1(2)\mathcal R_2(2),
  \nonumber\\
  \mathcal N^{(3)}_{(1,1,0)}
  &=\mathcal R_2(1)\mathcal R_1(1),
  \label{eq:U3-r2-two-norms-section6}
\end{align}
where the arguments $(q,\mathfrak t)$ have been suppressed.  Since the Schur--Macdonald transition coefficient of the spectral label $(1,1,0)$ is again $\kappa_2$, the Wilson two-point function is
\begin{equation}
  \widehat{\mathbb I}^{\mathrm{Neu},U(3)}_{(2,0,0)}(q,\mathfrak t)
  =\mathcal R_1(2)\mathcal R_2(2)
   +\kappa_2^2\mathcal R_2(1)\mathcal R_1(1)
  .
  \label{eq:U3-r2-exact-check-section6}
\end{equation}
The transition coefficient is unchanged from rank two because the electric representation is still $h_2$.  The accompanying norm differs because it is the $N=3$ norm of the spectral eigenfunction $P_{(1,1,0)}$.  Equation \eqref{eq:U3-r2-exact-check-section6} is the case $N=3$ of Eq.~(3.64) of \cite{Hatsuda:2025bmn}.

For reference we also record the degree-three rank-three formula.  Put
\begin{equation}
  \lambda_3(q,\mathfrak t)
  =\frac{(q-\mathfrak t)(q-\mathfrak t^2)}
        {(1-q\mathfrak t)(1-q\mathfrak t^2)}.
  \label{eq:lambda3-section6}
\end{equation}
Using \eqref{eq:h3-expansion} and \eqref{eq:U3-norm-factorization-section6}, one gets
\begin{align}
  \widehat{\mathbb I}^{\mathrm{Neu},U(3)}_{(3,0,0)}(q,\mathfrak t)
  ={}&\mathcal R_1(3)\mathcal R_2(3)
  +\kappa_3^2\mathcal R_1(1)^2\mathcal R_2(2)
  +\lambda_3^2.
  \label{eq:U3-r3-exact-check-section6}
\end{align}
The three terms are the spectral channels $\mu=(3,0,0),(2,1,0),(1,1,1)$.  The last Macdonald polynomial is a determinant translation and has norm one.

\subsection{Hall--Littlewood and Schur degenerations}
\label{subsec:limits-section6}

The expressions above have two simple degenerations.  First take the Schur specialization
\begin{equation}
  \mathfrak t=q.
  \label{eq:schur-specialization-section6}
\end{equation}
Then $P_\mu(\mathbf x;q,q)=s_\mu(\mathbf x)$, and on the unit torus the Macdonald density is the squared modulus of the Weyl denominator.  Schur orthogonality gives $\langle s_\lambda,s_\mu\rangle_{N;q,q}=\delta_{\lambda\mu}$.  Since $h_r=s_{(r)}$, the expansion \eqref{eq:one-row-transition} contains only the top spectral label,
\begin{equation}
  Q_\mu[\mathbb A_{q,q}]=\delta_{\mu,(r)}.
  \label{eq:schur-transition-collapse-section6}
\end{equation}
Therefore
\begin{equation}
  \widehat{\mathbb I}^{\mathrm{Neu}}_{(r)}(q,q)=1.
  \label{eq:schur-index-one-section6}
\end{equation}
The electric Schur point $\mathfrak t=q$ maps to the physical Nahm coordinate
$q/\mathfrak t=1$.  Individual shift coefficients can remain non-zero at
this specialization as functions of $\mathbf x$.  In the matrix element
\eqref{eq:eval-two-point} only the top orbit $v=(r)$ survives in the limit
$\mathfrak t\to q$, because the coefficients \eqref{eq:eval-coefficients} are proportional to $Q_v[\mathbb A_{q,\mathfrak t}]$.  Here the zero-shift coefficient is evaluated at $\boldsymbol\xi_0$ for generic $\mathfrak t$, and the limit $\mathfrak t\to q$ is taken afterwards.  At $\mathfrak t=q$ itself the points $\boldsymbol\xi_\nu$ are no longer distinct up to permutation, and the zero-shift coefficient of the product of the specialized operators differs from the coefficient of $\mathrm{ev}_0$.  For $U(2)$ and $r=1$ the former equals $2$, while the shift $(-1,1)$ takes $\boldsymbol\xi_0=(q,1)$ to $(1,q)$ and contributes $-1$, so that the coefficient of $\mathrm{ev}_0$ is again one.

The Hall--Littlewood limit of the rational Wilson correlator is obtained by taking
\begin{equation}
  q\to0
  \label{eq:HL-limit-section6}
\end{equation}
with $\mathfrak t$ fixed.  Macdonald polynomials become Hall--Littlewood polynomials.  From the product formula \eqref{eq:one-row-product-coeff},
\begin{equation}
  Q_\mu[\mathbb A_{0,\mathfrak t}]
  =\prod_{s\in\mu}\mathfrak t^{\ell'(s)}
  =\mathfrak t^{n(\mu)},
  \qquad
  n(\mu)=\sum_{i\geq1}(i-1)\mu_i.
  \label{eq:HL-one-row-transition-section6}
\end{equation}
The coefficient is the one-row Kostka--Foulkes coefficient in the convention of \eqref{eq:schur-macdonald-expansion}.  The Wilson answer becomes
\begin{equation}
  \widehat{\mathbb I}^{\mathrm{Neu}}_{(r)}(0,\mathfrak t)
  =\sum_{\substack{\mu\vdash r\\
  \ell(\mu)\leq N}}
  \mathfrak t^{2n(\mu)}\mathcal N^{(N)}_\mu(0,\mathfrak t).
  \label{eq:HL-one-row-index-section6}
\end{equation}
At $q=0$, $n(\mu)$ is the Kostka--Foulkes statistic of the Macdonald spectral label.  For $U(2)$ and $r=2,3$, and for $U(3)$ and $r=2,3$, the allowed partitions reproduce the finite-$N$ spectral sums displayed above.

\section{Algebraic DAHA operators and conjectural geometry}
\label{sec:hecke-outlook}

The spherical-DAHA construction gives a difference operator and its Macdonald spectrum.  Its interpretation as the regular Nahm-pole bubbling operator is the localization conjecture stated in this section, for which the affine-Grassmannian and Satake language provides motivation.

\subsection{The spherical-DAHA operator}
\label{subsec:boundary-hecke-operator}

Let $\mathsf{SH}_{N}(q,\mathfrak t)$ denote the spherical DAHA acting in the symmetric Laurent-polynomial representation.  For a dominant $U(N)$ weight $\lambda$ define
\begin{equation}
  \mathsf H^{\mathrm{DAHA}}_\lambda
  =\mathbf e_+s_\lambda(Y)\mathbf e_+.
  \label{eq:section7-boundary-Hecke-operator}
\end{equation}
With $Y_i=\tau^{N-1}Y_i^{(0)}$, the two torus orientations use $Y$ and $Y^{-1}$ with inverse homogeneous normalizations, as described in Section~\ref{subsec:Macdonald-kernel-identity}.  The physical boundary S-duality map includes $u_R\mapsto u_R^{-1}$, equivalently $\mathfrak t\mapsto q/\mathfrak t$, while all operators here are written in the electric parameter $\mathfrak t$.

For $\lambda=(r)$ the operator has the expansion
\begin{equation}
  \mathsf H^{\mathrm{DAHA}}_{(r)}
  =\sum_{\substack{v\vdash r\\ \ell(v)\leq N}}
   \ \sum_{\gamma^+=v}
  \mathsf C^{\mathrm{DAHA}}_{r,\gamma}
  (\mathbf x;q,\mathfrak t)T_{q^\gamma}.
  \label{eq:section7-one-row-Hecke-expansion}
\end{equation}
This form retains every composition in a Weyl orbit.  The orbit $v=(r)$ is the top shift orbit.  The remaining orbits carry precisely the coweights allowed by the dominance condition for monopole screening and are therefore the candidate lower magnetic sectors.

The Macdonald spectral statement is
\begin{equation}
  \mathsf H^{\mathrm{DAHA}}_\lambda
  P_\mu(\mathbf x;q,\mathfrak t)
  =s_\lambda(\boldsymbol\xi_\mu)
  P_\mu(\mathbf x;q,\mathfrak t),
  \qquad
  \boldsymbol\xi_\mu
  =(q^{\mu_1}\mathfrak t^{N-1},\ldots,q^{\mu_N}).
  \label{eq:section7-Hecke-eigenkernel-one-row}
\end{equation}
Equation \eqref{eq:section7-Hecke-eigenkernel-one-row} is the discrete spectral theorem used in Section~\ref{sec:daha-kernel}.

The determinant translation is
\begin{equation}
  \mathsf C_N
  =Y_1\cdots Y_N
  =\mathcal D_N
  =\mathfrak t^{N(N-1)/2}T_{q^{\mathbf 1}}.
  \label{eq:section7-determinant-translation}
\end{equation}
For every integral $m$,
\begin{equation}
  \mathsf H^{\mathrm{DAHA}}_{\lambda+m\mathbf 1}
  =\mathsf C_N^m\mathsf H^{\mathrm{DAHA}}_\lambda,
  \qquad
  \mathsf H^{\mathrm{DAHA}}_{(\lambda+m\mathbf 1)^\vee}
  =\mathsf C_N^{-m}\mathsf H^{\mathrm{DAHA}}_{\lambda^\vee}.
  \label{eq:section7-central-twist}
\end{equation}
The inverse factors cancel in an oriented two-line insertion:
\begin{equation}
  \mathsf H^{\mathrm{DAHA}}_{(\lambda+m\mathbf 1)^\vee}
  \mathsf H^{\mathrm{DAHA}}_{\lambda+m\mathbf 1}
  =
  \mathsf H^{\mathrm{DAHA}}_{\lambda^\vee}
  \mathsf H^{\mathrm{DAHA}}_\lambda.
  \label{eq:section7-determinant-pair-cancellation}
\end{equation}

\subsection{Affine-Grassmannian motivation}
\label{subsec:strata-and-screening}

For a dominant cocharacter $\lambda$, the affine-Grassmannian Schubert variety has the standard stratification
\begin{equation}
  \overline{\mathrm{Gr}}^\lambda
  =\bigsqcup_{v\leq\lambda}\mathrm{Gr}^v,
  \qquad
  \lambda-v\in Q^\vee_+.
  \label{eq:section7-affine-Gr-stratification}
\end{equation}
The same dominance condition governs the possible screened magnetic charges.  This agreement explains why the shift orbits of the DAHA operator are naturally organized by the affine-Grassmannian closure order.

Whether a boundary-bubbling moduli space can be identified with a transverse slice to $\mathrm{Gr}^v\subset\overline{\mathrm{Gr}}^\lambda$ is an open question.  The one-sided handsaw quotient examined in Section~\ref{sec:boundary-bubbling} is insufficient already in the first non-minuscule example.  Coulomb-branch constructions of affine-Grassmannian slices provide a possible framework for a derivation \cite{Braverman:2016wma}.

Geometric Satake identifies representations of the Langlands-dual group with an appropriate convolution category on the affine Grassmannian \cite{Mirkovic:2007}.  At the level of its classical Grothendieck group, it supplies an unrefined representation-ring counterpart of the lambda-ring relation used in the one-row recursion.  This suggests comparing $s_\lambda(X)$ with a magnetic operator resolved by the order \eqref{eq:section7-affine-Gr-stratification}; related $K$-theoretic structures provide further motivation \cite{Cautis:2015aia}.

\subsection{The localization conjecture}
\label{subsec:general-partitions-prescription}

For a dominant integral weight $\lambda=(\lambda_1,\ldots,\lambda_N)$, set
\begin{equation}
  m=\lambda_N,
  \qquad
  \bar\lambda=\lambda-m\mathbf1,
  \qquad
  \bar\lambda_N=0.
  \label{eq:section7-lambda-determinant-reduction}
\end{equation}
The determinant-reduced weight $\bar\lambda$ is a partition, and its DAHA operator has a unique shift expansion
\begin{equation}
  \mathsf H^{\mathrm{DAHA}}_{\bar\lambda}
  =\sum_{\substack{\gamma\in\mathbb Z_{\geq0}^N\\
                    |\gamma|=|\bar\lambda|}}
  \mathsf C^{\mathrm{DAHA}}_{\bar\lambda,\gamma}
  (\mathbf x;q,\mathfrak t)T_{q^\gamma}.
  \label{eq:section7-general-shift-coeff-definition}
\end{equation}
Equation \eqref{eq:section7-general-shift-coeff-definition} defines the coefficients that enter the conjecture.  The operator for the original integral weight is restored by
\begin{equation}
  \mathsf H^{\mathrm{DAHA}}_\lambda
  =\mathsf C_N^m\mathsf H^{\mathrm{DAHA}}_{\bar\lambda};
  \label{eq:section7-general-determinant-restoration}
\end{equation}
equivalently, the shifts in \eqref{eq:section7-general-shift-coeff-definition} are translated by $m\mathbf1$, together with the scalar prefactor contained in $\mathsf C_N^m$.

For a physical boundary-localization operator, let $\widehat{\mathcal T}^{\partial,\mathrm{loc}}_{\bar\lambda\to v}$ denote the complete contribution of the effective charge $v$, including the chosen unbubbled operator.  The natural comparison with the DAHA operator is sector-wise only after every composition in a Weyl orbit has been retained.

\begin{conjecture}[Boundary localization]
\label{conj:general-boundary-Hecke}
For every determinant-reduced dominant magnetic charge $\bar\lambda$ and every allowed dominant effective charge $v\leq\bar\lambda$ with $|v|=|\bar\lambda|$,
\begin{equation}
  \widehat{\mathcal T}^{\partial,\mathrm{loc}}_{\bar\lambda\to v}
  =\sum_{\gamma^+=v}
  \mathsf C^{\mathrm{DAHA}}_{\bar\lambda,\gamma}
  (\mathbf x;q,\mathfrak t)T_{q^\gamma},
  \label{eq:section7-general-coeff-conjecture}
\end{equation}
after the same normalization is used on both sides, namely that of Conjecture~\ref{conj:boundary-localization-section4}, in which the minuscule operators are the $\mathcal D_k$.  Equivalently,
\begin{equation}
  \widehat{\mathcal T}^{\partial,\mathrm{loc}}_{\bar\lambda}
  =\mathbf e_+s_{\bar\lambda}(Y)\mathbf e_+,
  \label{eq:section7-general-formal-operator}
\end{equation}
and an arbitrary dominant integral charge is restored by
\begin{equation}
  \widehat{\mathcal T}^{\partial,\mathrm{loc}}_\lambda
  =\mathsf C_N^m
   \widehat{\mathcal T}^{\partial,\mathrm{loc}}_{\bar\lambda}.
  \label{eq:section7-general-physical-determinant-restoration}
\end{equation}
\end{conjecture}

The normalization $\mathsf Z^\partial_{\bar\lambda,\bar\lambda}=1$ refers to the residual bubbling factor after the unbubbled top-shift coefficient has been extracted, and the full composition coefficient $\mathsf C^{\mathrm{DAHA}}_{\bar\lambda,\gamma}$ contains that coefficient.  The corresponding statement for $\lambda$ is obtained through the determinant restoration \eqref{eq:section7-general-physical-determinant-restoration}.

For $U(2)$ and $B=(2,0)\to v=(1,1)$, the local two-step convolution gives the two JK residues $R_{12}$ and $R_{21}$, while subtracting the determinant IC/Satake summand contributes $-\mathfrak t$.  This is the complete DAHA coefficient $R_{12}+R_{21}-\mathfrak t$.  The two residues are the cross terms of $\mathcal D_1^2$ and the subtraction realizes $h_2=e_1^2-e_2$.  For the $U(3)$ representative $(2,0,0)\to(1,1,0)$, both residues and the antisymmetric subtraction acquire the common spectator factor $\mathsf S_{12|3}$, giving
\begin{equation}
 \mathsf C^{\mathrm{DAHA},U(3)}_{2,(1,1,0)}
 =\mathsf S_{12|3}\mathsf C^{\mathrm{DAHA},U(2)}_{2,(1,1)}.
\end{equation}
These are the first two ambient-rank realizations of the same transverse
$A_1$ convolution, which realizes the algebraic coefficient geometrically.

Proposition~\ref{prop:evaluation-module} shows that, on the evaluation module, the zero-shift coefficient at $\boldsymbol\xi_0$ of the product of the two oriented operators equals the Neumann pairing,
\begin{equation}
  \bigl[\mathsf H^{\mathrm{DAHA}}_\lambda\mathsf H^{\mathrm{DAHA}}_{\lambda^\vee}\bigr]_{T^0}(\boldsymbol\xi_0)
  =\widehat{\mathbb I}^{\Neu}_{\lambda}(q,\mathfrak t).
  \label{eq:section7-general-S-duality-conjecture}
\end{equation}
A physical S-duality statement for the bubbling-corrected line requires two further inputs, which remain open.  They are the localization conjecture for the magnetic operator and the identification of the regular Nahm-pole two-point function with the left-hand side of \eqref{eq:section7-general-S-duality-conjecture}.  For minuscule charge the left-hand side agrees with the result of the Higgsing computation of \cite{Hatsuda:2025bmn}, as recalled in Sections~\ref{subsec:higgsing-points} and~\ref{sec:eval-module}.

\subsection{Further directions}
\label{subsec:further-directions}

The first open problem is a first-principles derivation of the boundary SQM.  Such a derivation must specify the microscopic multiplets, their gauge and flavor characters, the equivariant charges, the stability condition, the JK chamber, and possible contributions from non-compact directions.  It must reproduce the first non-minuscule lower coefficient without using the DAHA identity as an input.

A second problem is to determine whether the resulting localization space admits a direct description as an affine-Grassmannian transverse slice or an associated derived intersection.  Only after that identification, together with the relevant equivariant refinement, is established would a precise $K$-theoretic Satake interpretation of the coefficients be available.  Extensions to other gauge groups and to interfaces can then be formulated using the corresponding root-system DAHA and interface kernels.

\appendix
\makeatletter
\renewcommand{\theHequation}{appendix.\thesection.\arabic{equation}}
\makeatother

\section{Special-function and Macdonald conventions}
\label{app:conventions-special-functions}

This appendix records the conventions needed to compare the finite-$N$ Wilson pairing, the spherical-DAHA operator, and the low-rank formulae.  We take $|q|<1$ for convergent products; all identities also admit a formal interpretation in a completed Laurent ring.

\subsection{Shifted products and plethystic notation}
\label{subsec:appA-q-products}

For $n\geq0$,
\begin{equation}
  (z;q)_n=\prod_{m=0}^{n-1}(1-zq^m),
  \qquad
  (z;q)_\infty=\prod_{m=0}^{\infty}(1-zq^m),
  \qquad
  (z;q)_0=1.
  \label{eq:appA-q-pochhammer}
\end{equation}
The identities
\begin{equation}
  (z;q)_n=\frac{(z;q)_\infty}{(zq^n;q)_\infty},
  \qquad
  \frac{(az;q)_\infty}{(bz;q)_\infty}
  =
  \exp\!\left[
  \sum_{m\geq1}\frac{b^m-a^m}{1-q^m}\frac{z^m}{m}
  \right]
  \label{eq:appA-q-product-ratio}
\end{equation}
will be used below.  For a virtual alphabet $\mathbb A$, the plethystic exponential is
\begin{equation}
  \operatorname{Exp}[\mathbb A]
  =
  \exp\!\left(\sum_{m\geq1}\frac{p_m[\mathbb A]}{m}\right),
  \qquad
  p_m\!\left[\frac{1-q}{1-\mathfrak t}\right]
  =\frac{1-q^m}{1-\mathfrak t^m}.
  \label{eq:appA-plethystic-exp}
\end{equation}
In particular,
\begin{equation}
  \Delta_N(\mathbf x;q,\mathfrak t)
  =
  \exp\!\left[
  -\sum_{m\geq1}\frac{1-\mathfrak t^m}{1-q^m}
  \frac{p_m(\mathbf x)p_m(\mathbf x^{-1})-N}{m}
  \right].
  \label{eq:appA-vector-PE}
\end{equation}

\subsection{The finite-\texorpdfstring{$N$}{N} pairing and determinant direction}
\label{subsec:appA-constant-term-pairing}

The normalized bilinear pairing is
\begin{equation}
  \langle f,g\rangle_{N;q,\mathfrak t}
  =
  \frac{
  \operatorname{CT}_{\mathbf x}
  \left[\Delta_N(\mathbf x;q,\mathfrak t)
  f(\mathbf x)g(\mathbf x^{-1})\right]}
  {\operatorname{CT}_{\mathbf x}
  \Delta_N(\mathbf x;q,\mathfrak t)}.
  \label{eq:appA-normalized-pairing}
\end{equation}
The common factor $1/N!$ has cancelled between numerator and denominator.  If
\begin{equation}
  d_m(\mathbf x)=(x_1\cdots x_N)^m,
  \label{eq:appA-determinant-monomial}
\end{equation}
then
\begin{equation}
  \langle d_m f,d_m g\rangle_{N;q,\mathfrak t}
  =\langle f,g\rangle_{N;q,\mathfrak t}.
  \label{eq:appA-center-pairing-invariance}
\end{equation}
For a dominant Laurent weight $\lambda$,
\begin{equation}
  s_\lambda(\mathbf x)
  =d_{\lambda_N}(\mathbf x)
  s_{\lambda-\lambda_N\mathbf1}(\mathbf x),
  \qquad
  s_{\lambda^\vee}(\mathbf x)=s_\lambda(\mathbf x^{-1}).
  \label{eq:appA-laurent-schur}
\end{equation}
It follows directly that
\begin{equation}
  \widehat{\mathbb I}^{\Neu}_{\lambda+m\mathbf1}
  =\widehat{\mathbb I}^{\Neu}_\lambda.
  \label{eq:appA-electric-determinant-invariance}
\end{equation}
The magnetic determinant reduction is instead implemented by the inverse pair of translations $\mathsf C_N^{m}$ and $\mathsf C_N^{-m}$ in the two oriented spherical-DAHA operators, see \eqref{eq:general-determinant-factorization}.

\subsection{Macdonald polynomials, norms, and transition coefficients}
\label{subsec:appA-macdonald-polynomials}

The monic Macdonald polynomial is characterized by
\begin{equation}
  P_\mu=m_\mu+\sum_{\nu<\mu}u_{\mu\nu}(q,\mathfrak t)m_\nu
  \label{eq:appA-monic-triangularity}
\end{equation}
and orthogonality for \eqref{eq:appA-normalized-pairing}.  Its determinant extension is
\begin{equation}
  P_{\mu+m\mathbf1}(\mathbf x;q,\mathfrak t)
  =d_m(\mathbf x)P_\mu(\mathbf x;q,\mathfrak t).
  \label{eq:appA-determinant-shift-P}
\end{equation}
The finite-$N$ norm is
\begin{equation}
  \langle P_\mu,P_\nu\rangle_{N;q,\mathfrak t}
  =\delta_{\mu\nu}\mathcal N_\mu^{(N)},
  \qquad
  \mathcal N_\mu^{(N)}
  =\frac{D_N(\mu;q,\mathfrak t)}{D_N(0;q,\mathfrak t)},
  \label{eq:appA-macdonald-orthogonality}
\end{equation}
where
\begin{equation}
  D_N(\mu;q,\mathfrak t)
  =
  \prod_{1\leq i<j\leq N}
  \frac{(q^{\mu_i-\mu_j}\mathfrak t^{j-i};q)_\infty
        (q^{\mu_i-\mu_j+1}\mathfrak t^{j-i};q)_\infty}
       {(q^{\mu_i-\mu_j}\mathfrak t^{j-i+1};q)_\infty
        (q^{\mu_i-\mu_j+1}\mathfrak t^{j-i-1};q)_\infty}.
  \label{eq:appA-DN}
\end{equation}
Thus $\mathcal N_{\mu+m\mathbf1}^{(N)}=\mathcal N_\mu^{(N)}$.

For a box $s=(i,j)$ in $\mu$,
\begin{equation}
  a(s)=\mu_i-j,\quad
  \ell(s)=\mu'_j-i,\quad
  a'(s)=j-1,\quad
  \ell'(s)=i-1.
  \label{eq:appA-arm-leg}
\end{equation}
The stable dual polynomial is
\begin{equation}
  Q_\mu=b_\mu P_\mu,
  \qquad
  b_\mu(q,\mathfrak t)
  =\prod_{s\in\mu}
  \frac{1-q^{a(s)}\mathfrak t^{\ell(s)+1}}
       {1-q^{a(s)+1}\mathfrak t^{\ell(s)}}.
  \label{eq:appA-Q-polynomial}
\end{equation}
This stable dual normalization is distinct from the finite-$N$ norm $\mathcal N_\mu^{(N)}$.

The Cauchy identity is
\begin{equation}
  \Pi_{q,\mathfrak t}(\mathbf x,\mathbf y)
  =
  \sum_\mu P_\mu(\mathbf x)Q_\mu(\mathbf y)
  =
  \exp\!\left[
  \sum_{m\geq1}
  \frac{1-\mathfrak t^m}{1-q^m}
  \frac{p_m(\mathbf x)p_m(\mathbf y)}{m}
  \right].
  \label{eq:appA-cauchy-kernel}
\end{equation}
It expresses the duality of the bases $P_\mu$ and $Q_\mu$ for the stable Macdonald scalar product \cite[Chapter~VI]{Macdonald:1995} and is used here to obtain the transition coefficients.

Define
\begin{equation}
  s_\lambda(\mathbf x)
  =\sum_{\mu\preceq\lambda}
  \mathsf K_{\lambda\mu}^{(N)}
  P_\mu(\mathbf x;q,\mathfrak t).
  \label{eq:appA-transition-matrix}
\end{equation}
Then
\begin{equation}
  \widehat{\mathbb I}^{\Neu}_\lambda
  =
  \sum_{\mu\preceq\lambda}
  \left(\mathsf K_{\lambda\mu}^{(N)}\right)^2
  \mathcal N_\mu^{(N)}.
  \label{eq:appA-wilson-spectral}
\end{equation}
For $\lambda=(r)$, specialize \eqref{eq:appA-cauchy-kernel} by
$p_m(\mathbf y)=u^m(1-q^m)/(1-\mathfrak t^m)$ to obtain
\begin{equation}
  h_r(\mathbf x)
  =
  \sum_{\substack{\mu\vdash r\\\ell(\mu)\leq N}}
  Q_\mu\!\left[\frac{1-q}{1-\mathfrak t}\right]
  P_\mu(\mathbf x;q,\mathfrak t),
  \label{eq:appA-one-row-transition}
\end{equation}
where
\begin{equation}
  Q_\mu\!\left[\frac{1-q}{1-\mathfrak t}\right]
  =
  \prod_{s\in\mu}
  \frac{\mathfrak t^{\ell'(s)}-q^{a'(s)+1}}
       {1-q^{a(s)+1}\mathfrak t^{\ell(s)}}.
  \label{eq:appA-one-row-product}
\end{equation}

\subsection{Macdonald--Ruijsenaars operators}
\label{subsec:appA-difference-center}

For $\gamma\in\mathbb Z^N$ let
\begin{equation}
  T_{q^\gamma}=\prod_{i=1}^N T_{q,x_i}^{\gamma_i},
  \qquad
  T_{q,x_i}f(\ldots,x_i,\ldots)
  =f(\ldots,qx_i,\ldots).
  \label{eq:appA-shift-operators}
\end{equation}
The elementary commuting operators are
\begin{equation}
  \mathcal D_k
  =
  \mathfrak t^{k(k-1)/2}
  \sum_{\substack{I\subset\{1,\ldots,N\}\\|I|=k}}
  \prod_{\substack{i\in I\\j\notin I}}
  \frac{\mathfrak t x_i-x_j}{x_i-x_j}
  \prod_{i\in I}T_{q,x_i},
  \qquad
  0\leq k\leq N,
  \label{eq:appA-Dk}
\end{equation}
and $\mathcal D_k=0$ for $k>N$.  For
\begin{equation}
  \boldsymbol\xi_\mu
  =(q^{\mu_1}\mathfrak t^{N-1},\ldots,q^{\mu_N}),
  \label{eq:appA-spectral-point}
\end{equation}
one has
\begin{equation}
  \mathcal D_kP_\mu=e_k(\boldsymbol\xi_\mu)P_\mu.
  \label{eq:appA-Dk-eigenvalue}
\end{equation}
The one-row operator is equivalently characterized by
\begin{equation}
  \sum_{r\geq0}u^r\mathsf H^{\mathrm{DAHA}}_{(r)}
  =
  \left(\sum_{k=0}^{N}(-u)^k\mathcal D_k\right)^{-1}.
  \label{eq:appA-one-row-D-generating}
\end{equation}
The determinant translation is
\begin{equation}
  \mathsf C_N
  =Y_1\cdots Y_N
  =\mathcal D_N
  =\mathfrak t^{N(N-1)/2}T_{q^{\mathbf1}}.
  \label{eq:appA-central-shift}
\end{equation}
It is a coweight translation.

\section{Algebraic DAHA coefficients}
\label{app:boundary-bubbling-fixed-point}

We collect low-degree coefficients of the operator \eqref{eq:appA-one-row-D-generating} for comparison with a microscopic localization calculation.  The two contour residues of the local $U(2)$ convolution calculation of Section~\ref{subsec:rank-two-independent-residue} and their spectator-dressed $U(3)$ lift appear among them.

\subsection{Composition recursion}
\label{subsec:appB-agreement-daha}

For $I\subset\{1,\ldots,N\}$ put
\begin{equation}
  \varepsilon_I=\sum_{i\in I}\varepsilon_i,
  \qquad
  \mathsf A_I(\mathbf x)
  =
  \mathfrak t^{|I|(|I|-1)/2}
  \prod_{\substack{i\in I\\j\notin I}}
  \frac{\mathfrak t x_i-x_j}{x_i-x_j}.
  \label{eq:appB-AI}
\end{equation}
Write
\begin{equation}
  \mathsf H^{\mathrm{DAHA}}_{(r)}
  =
  \sum_{|\gamma|=r}
  \mathsf C^{\mathrm{DAHA}}_{r,\gamma}(\mathbf x)
  T_{q^\gamma},
  \qquad
  \mathsf C^{\mathrm{DAHA}}_{0,0}=1.
  \label{eq:appB-fp-generating-function}
\end{equation}
The multiplication rule
\begin{equation}
  F(\mathbf x)T_{q^\alpha}\,
  G(\mathbf x)T_{q^\beta}
  =
  F(\mathbf x)G(q^\alpha\mathbf x)
  T_{q^{\alpha+\beta}}
  \label{eq:appB-difference-multiplication}
\end{equation}
and \eqref{eq:appA-one-row-D-generating} give, for $r\geq1$,
\begin{equation}
  \mathsf C^{\mathrm{DAHA}}_{r,\gamma}(\mathbf x)
  =
  \sum_{\substack{\varnothing\neq I\subset\{1,\ldots,N\}\\
  \varepsilon_I\leq\gamma}}
  (-1)^{|I|+1}
  \mathsf A_I(\mathbf x)
  \mathsf C^{\mathrm{DAHA}}_{r-|I|,\gamma-\varepsilon_I}
  (q^{\varepsilon_I}\mathbf x).
  \label{eq:appB-projector-recursion}
\end{equation}
At each step \eqref{eq:appB-projector-recursion} removes the first exterior-power layer.  Iterating it gives ordered composition paths $(I_1,\ldots,I_\ell)$ satisfying
$\varepsilon_{I_1}+\cdots+\varepsilon_{I_\ell}=\gamma$.  The path weight is
\begin{equation}
  \prod_{a=1}^{\ell}
  (-1)^{|I_a|+1}
  \mathsf A_{I_a}\!\left(q^{\sum_{b<a}\varepsilon_{I_b}}\mathbf x\right).
  \label{eq:appB-path-weight}
\end{equation}

For the top composition $r\varepsilon_i$ only singleton subsets occur, and
\begin{equation}
  \mathsf C^{\mathrm{DAHA}}_{r,r\varepsilon_i}
  =
  \prod_{m=0}^{r-1}\prod_{j\neq i}
  \frac{\mathfrak t q^m x_i-x_j}{q^m x_i-x_j}.
  \label{eq:appB-top-product}
\end{equation}
This is the full unbubbled shift coefficient.  The pure bubbling factor of the top sector is one.

\subsection{Rank-two coefficients}
\label{subsec:appB-low-rank-coefficients}

For $U(2)$ define
\begin{equation}
  A_1(\mathbf x)=\frac{\mathfrak t x_1-x_2}{x_1-x_2},
  \qquad
  A_2(\mathbf x)=\frac{\mathfrak t x_2-x_1}{x_2-x_1}.
  \label{eq:appB-U2-A1A2}
\end{equation}
At degree two,
\begin{align}
  \mathsf C^{\mathrm{DAHA}}_{2,(2,0)}
  &=
  A_1(\mathbf x)A_1(qx_1,x_2),
  \label{eq:appB-U2-r2-top}
  \\
  \mathsf C^{\mathrm{DAHA}}_{2,(1,1)}
  &=
  A_1(\mathbf x)A_2(qx_1,x_2)
  +A_2(\mathbf x)A_1(x_1,qx_2)
  -\mathfrak t.
  \label{eq:appB-U2-r2-det-layer}
\end{align}
The first two summands in \eqref{eq:appB-U2-r2-det-layer} are the contour residues $R_{12}$ and $R_{21}$ of the local resolved-convolution model; the last is the localized Tate-twisted zero-section subtraction.

At degree three the coefficient of the representative lower shift $(2,1)$ is
\begin{equation}
\begin{split}
  \mathsf C^{\mathrm{DAHA}}_{3,(2,1)}
  ={}&
  A_1(\mathbf x)A_1(qx_1,x_2)A_2(q^2x_1,x_2)
  \\
  &+
  A_1(\mathbf x)A_2(qx_1,x_2)A_1(qx_1,qx_2)
  \\
  &+
  A_2(\mathbf x)A_1(x_1,qx_2)A_1(qx_1,qx_2)
  \\
  &-\mathfrak t A_1(\mathbf x)
  -\mathfrak t A_1(qx_1,qx_2).
\end{split}
  \label{eq:appB-U2-r3-lower}
\end{equation}
Together with the Weyl-reflected coefficients, these formulae reproduce
\begin{equation}
  \mathsf H^{\mathrm{DAHA}}_{(2)}
  =\mathcal D_1^2-\mathcal D_2,
  \qquad
  \mathsf H^{\mathrm{DAHA}}_{(3)}
  =\mathcal D_1^3-2\mathcal D_1\mathcal D_2.
  \label{eq:appB-U2-low-degree-D}
\end{equation}

\subsection{The first rank-three lower coefficient}
\label{subsec:appB-rank-three-coefficients}

For $U(3)$ and the representative shift $\gamma=(1,1,0)$,
\begin{equation}
\begin{split}
  \mathsf C^{\mathrm{DAHA}}_{2,(1,1,0)}
  ={}&
  \mathsf A_{\{1\}}^{(0)}(\mathbf x)
  \mathsf A_{\{2\}}^{(0)}(qx_1,x_2,x_3)
  \\
  &+
  \mathsf A_{\{2\}}^{(0)}(\mathbf x)
  \mathsf A_{\{1\}}^{(0)}(x_1,qx_2,x_3)
  -
  \mathfrak t\,\mathsf A_{\{1,2\}}^{(0)}(\mathbf x),
\end{split}
  \label{eq:appB-U3-r2-lower}
\end{equation}
where
\begin{equation}
	  \mathsf A_I^{(0)}(\mathbf x)
	  =
	  \prod_{\substack{i\in I\\j\notin I}}
	  \frac{\mathfrak t x_i-x_j}{x_i-x_j},
	  \qquad
	  \mathsf A_I
	  =\mathfrak t^{|I|(|I|-1)/2}\mathsf A_I^{(0)}.
  \label{eq:appB-AI-zero}
\end{equation}
Equation \eqref{eq:appB-U3-r2-lower} is an algebraic coefficient of
$\mathcal D_1^2-\mathcal D_2$.  With
\begin{equation}
 \mathsf S_{12|3}
 =\mathsf A_{\{1,2\}}^{(0)}
 =\prod_{a=1}^{2}\frac{\mathfrak t x_a-x_3}{x_a-x_3},
\end{equation}
its three terms factorize as
\begin{equation}
 \mathsf C^{\mathrm{DAHA},U(3)}_{2,(1,1,0)}
 =\mathsf S_{12|3}
  \mathsf C^{\mathrm{DAHA},U(2)}_{2,(1,1)}.
 \label{eq:appB-U3-U2-factorization}
\end{equation}
This is the spectator-root lift of the transverse $A_1$ calculation in
Section~\ref{subsec:U3-spectator-lift-section4}; the other shift
representatives follow by Weyl permutation.

\section{Finite-\texorpdfstring{$N$}{N} expansion data}
\label{app:expansion-data}

The following rational expressions and short $q$-expansions check the Wilson pairing and the discrete Macdonald spectral formula.  The expressions \eqref{eq:appC-U2-r2-exact} and \eqref{eq:appC-U2-r3-exact} are the vacuum-normalized correlators of Eqs.~(3.12) and (3.13) of \cite{Hatsuda:2025bmn}, and \eqref{eq:appC-U3-r2-exact} is the case $N=3$ of Eq.~(3.64) there.  Higher orders are recorded in the ancillary file \texttt{expansion\_data\_full.txt}.

Define
\begin{equation}
  \mathcal R_s(d;q,\mathfrak t)
  =
  \frac{(q^d\mathfrak t^s;q)_\infty
        (q^{d+1}\mathfrak t^s;q)_\infty}
       {(q^d\mathfrak t^{s+1};q)_\infty
        (q^{d+1}\mathfrak t^{s-1};q)_\infty}
  \frac{(\mathfrak t^{s+1};q)_\infty
        (q\mathfrak t^{s-1};q)_\infty}
       {(\mathfrak t^s;q)_\infty
        (q\mathfrak t^s;q)_\infty}.
  \label{eq:appC-Rs}
\end{equation}
Then
\begin{equation}
  \mathcal N_\mu^{(N)}
  =
  \prod_{1\leq i<j\leq N}
  \mathcal R_{j-i}(\mu_i-\mu_j;q,\mathfrak t).
  \label{eq:appC-norm-factorization}
\end{equation}
For $s=1$ this simplifies to
\begin{equation}
  \mathcal R_1(d;q,\mathfrak t)
  =
  \frac{(q;q)_d(\mathfrak t^2;q)_d}
       {(\mathfrak t;q)_d(q\mathfrak t;q)_d}.
  \label{eq:appC-R1}
\end{equation}

\subsection{Rank two}
\label{subsec:appC-rank-two-data}

Put
\begin{equation}
  \kappa_2=\frac{\mathfrak t-q}{1-q\mathfrak t},
  \qquad
  \kappa_3=\frac{(1+q)(\mathfrak t-q)}
  {1-q^2\mathfrak t},
  \label{eq:appC-kappas}
\end{equation}
and
\begin{equation}
  \rho_{31}
  =\frac{(1+q+q^2)(\mathfrak t-q)}
  {1-q^3\mathfrak t},
  \qquad
  \rho_{22}
  =\frac{(q-\mathfrak t)(q^2-\mathfrak t)}
  {(1-q\mathfrak t)(1-q^2\mathfrak t)}.
  \label{eq:appC-U2-r4-rhos}
\end{equation}
For $N=2$ the rational expressions are
\begin{align}
  \widehat{\mathbb I}^{U(2)}_{(2,0)}
  &=\mathcal R_1(2)+\kappa_2^2,
  \label{eq:appC-U2-r2-exact}
  \\
  \widehat{\mathbb I}^{U(2)}_{(3,0)}
  &=\mathcal R_1(3)+\kappa_3^2\mathcal R_1(1),
  \label{eq:appC-U2-r3-exact}
  \\
  \widehat{\mathbb I}^{U(2)}_{(4,0)}
  &=\mathcal R_1(4)+\rho_{31}^{\,2}\mathcal R_1(2)+\rho_{22}^{\,2}.
  \label{eq:appC-U2-r4-exact}
\end{align}
Their first two angular-fugacity coefficients are
\begin{align}
  \widehat{\mathbb I}^{U(2)}_{(2,0)}
  &=
  (1+\mathfrak t+\mathfrak t^2)
  +q(\mathfrak t-1)(\mathfrak t+1)^2+O(q^2),
  \label{eq:appC-U2-r2-short-series}
  \\
  \widehat{\mathbb I}^{U(2)}_{(3,0)}
  &=
  (1+\mathfrak t)(1+\mathfrak t^2)
  +q(\mathfrak t-1)(\mathfrak t+1)
  (\mathfrak t^2+\mathfrak t+1)+O(q^2),
  \label{eq:appC-U2-r3-short-series}
  \\
  \widehat{\mathbb I}^{U(2)}_{(4,0)}
  &=
  (1+\mathfrak t+\mathfrak t^2+\mathfrak t^3+\mathfrak t^4)
  +q(\mathfrak t-1)(\mathfrak t+1)^2
  (\mathfrak t^2+1)+O(q^2).
  \label{eq:appC-U2-r4-short-series}
\end{align}

\subsection{Rank three}
\label{subsec:appC-rank-three-data}

For $N=3$,
\begin{equation}
  \widehat{\mathbb I}^{U(3)}_{(2,0,0)}
  =
  \mathcal R_1(2)\mathcal R_2(2)
  +\kappa_2^2\mathcal R_2(1)\mathcal R_1(1),
  \label{eq:appC-U3-r2-exact}
\end{equation}
whereas
\begin{equation}
\begin{split}
  \widehat{\mathbb I}^{U(3)}_{(3,0,0)}
  ={}&
  \mathcal R_1(3)\mathcal R_2(3)
  +\kappa_3^2\mathcal R_1(1)^2\mathcal R_2(2)
  +\lambda_3^2,
\end{split}
  \label{eq:appC-U3-r3-exact}
\end{equation}
with
\begin{equation}
  \lambda_3
  =
  \frac{(q-\mathfrak t)(q-\mathfrak t^2)}
  {(1-q\mathfrak t)(1-q\mathfrak t^2)}.
\end{equation}
The short expansions are
\begin{align}
  \widehat{\mathbb I}^{U(3)}_{(2,0,0)}
  &=
  (\mathfrak t^2+1)(\mathfrak t^2+\mathfrak t+1)
  \nonumber\\
  &\quad
  +q(\mathfrak t-1)(\mathfrak t+1)
  (\mathfrak t^2+\mathfrak t+1)^2+O(q^2),
  \label{eq:appC-U3-r2-short-series}
  \\
  \widehat{\mathbb I}^{U(3)}_{(3,0,0)}
  &=
  (\mathfrak t^2+1)
  (\mathfrak t^4+\mathfrak t^3+\mathfrak t^2+\mathfrak t+1)
  \nonumber\\
  &\quad
  +q(\mathfrak t-1)(\mathfrak t+1)(\mathfrak t^2+1)
  (\mathfrak t^2+\mathfrak t+1)^2+O(q^2).
  \label{eq:appC-U3-r3-short-series}
\end{align}
These expressions make the finite-$N$ cutoff explicit: the spectral channel $(1,1,1)$ occurs for $N=3$ but is absent for $N=2$.

\subsection{Special limits}
\label{subsec:appC-limits}

At $q=0$, the constant terms in
\eqref{eq:appC-U2-r2-short-series}--\eqref{eq:appC-U3-r3-short-series}
are the Hall--Littlewood specializations of the finite-$N$ Wilson pairings.  Their powers of $\mathfrak t$ encode the Kostka--Foulkes data of the Macdonald spectral labels.

At the Schur specialization one sets $\mathfrak t=q$ in the rational expressions.  The transition coefficients vanish unless $\mu=(r)$, and
\begin{equation}
  \widehat{\mathbb I}^{U(N)}_{(r)}(q,q)=1.
  \label{eq:appC-schur-limit}
\end{equation}

\section{Proof of Proposition~\ref{prop:evaluation-module}}
\label{app:evaluation-proof}

Throughout this appendix $q$ and $\tau$ are indeterminates and $\mathfrak t=\tau^2$, as in Section~\ref{subsec:polynomial-representation}.  We use the notation of Section~\ref{sec:eval-module}, namely $\boldsymbol\xi_\nu=(q^{\nu_1}\mathfrak t^{N-1},\ldots,q^{\nu_N})$ for $\nu\in\mathbb Z^N$ and $\widetilde P_\lambda=P_\lambda/P_\lambda(\boldsymbol\xi_0)$ for $\lambda\in\Lambda_N^+$.  For a partition $\lambda$ the principal specialization formula of \cite[Chapter~VI, Section~6]{Macdonald:1995} shows that $P_\lambda(\boldsymbol\xi_0)\neq0$, and \eqref{eq:Laurent-Macdonald-convention} gives $P_{\lambda+m\mathbf 1}(\boldsymbol\xi_0)=\mathfrak t^{mN(N-1)/2}P_\lambda(\boldsymbol\xi_0)$.

\subsection{Evaluation symmetry}
\label{subsec:appD-evaluation-symmetry}

For partitions $\lambda$ and $\mu$ of length at most $N$, the evaluation symmetry of Macdonald polynomials \cite[Chapter~VI, Section~6]{Macdonald:1995} is
\begin{equation}
  \widetilde P_\lambda(\boldsymbol\xi_\mu)
  =\widetilde P_\mu(\boldsymbol\xi_\lambda).
  \label{eq:appD-evaluation-symmetry}
\end{equation}
It extends to all $\lambda,\mu\in\Lambda_N^+$.  Replacing $\lambda$ by $\lambda+m\mathbf 1$ multiplies $P_\lambda$ by $(x_1\cdots x_N)^m$ and hence the left-hand side by $q^{m|\mu|}$, while the right-hand side becomes $\widetilde P_\mu(q^m\boldsymbol\xi_\lambda)=q^{m|\mu|}\widetilde P_\mu(\boldsymbol\xi_\lambda)$ by homogeneity.  The same argument applies with the roles of $\lambda$ and $\mu$ exchanged.

\begin{lemma}
\label{lem:appD-density}
A Laurent polynomial in $x_1,\ldots,x_N$ over $\mathbb C(q,\tau)$ which vanishes at $\boldsymbol\xi_\mu$ for every $\mu\in\Lambda_N^+$ is zero.  The functionals $\mathrm{ev}_\lambda$, $\lambda\in\Lambda_N^+$, on $\mathscr P_N^{S_N}$ are linearly independent.
\end{lemma}

\begin{proof}
In the coordinates $z_i=x_i/x_{i+1}$ for $i<N$ and $z_N=x_N$, the point $\boldsymbol\xi_\mu$ has $z_i=\mathfrak t\,q^{\mu_i-\mu_{i+1}}$ and $z_N=q^{\mu_N}$.  The integers $\mu_i-\mu_{i+1}\geq0$ and $\mu_N$ can be chosen independently.  Each coordinate therefore ranges over an infinite set independently of the others, and a Laurent polynomial in $z_1,\ldots,z_N$ which vanishes on such a product of infinite sets is zero.  For the second statement, suppose that a finite combination $\sum_\lambda a_\lambda\mathrm{ev}_\lambda$ vanishes on $\mathscr P_N^{S_N}$.  Applying it to $\widetilde P_\mu$ and using \eqref{eq:appD-evaluation-symmetry} shows that $\sum_\lambda a_\lambda\widetilde P_\lambda$ vanishes at every $\boldsymbol\xi_\mu$.  It is therefore zero, and $a_\lambda=0$ because the $P_\lambda$ are linearly independent.
\end{proof}

\subsection{The right action}
\label{subsec:appD-right-action}

Let $\mathcal A_N$ be the commutative algebra generated by $\mathcal D_1,\ldots,\mathcal D_N$ and $\mathcal D_N^{-1}$.  We claim that every $D=\sum_\gamma C_\gamma(\mathbf x)T_{q^\gamma}$ in $\mathcal A_N$ has the following two properties.
\begin{enumerate}
\item[(a)] Each $C_\gamma$ is regular at $\boldsymbol\xi_\nu$ for all $\nu\in\mathbb Z^N$.
\item[(b)] $C_\gamma(\boldsymbol\xi_\lambda)=0$ whenever $\lambda\in\Lambda_N^+$ and $\lambda+\gamma\notin\Lambda_N^+$.
\end{enumerate}
For $\mathcal D_N^{\pm1}$ both properties are clear.  For $\mathcal D_k$ the coefficient of $T_{q^{\varepsilon_I}}$ is $\mathsf A_I$ of \eqref{eq:A-I-projector-section4}.  Its denominators $x_i-x_j$ do not vanish at $\boldsymbol\xi_\nu$, which gives (a).  If $\lambda\in\Lambda_N^+$ and $\lambda+\varepsilon_I\notin\Lambda_N^+$, there is an index $i$ with $\lambda_i=\lambda_{i+1}$, $i\notin I$ and $i+1\in I$, and the factor $\mathfrak t x_{i+1}-x_i$ of $\mathsf A_I$ vanishes at $\boldsymbol\xi_\lambda$, which gives (b).  Both properties are preserved by sums.  For a product $DD'$ the coefficient of $T_{q^\gamma}$ is $\sum_{\alpha+\beta=\gamma}C_\alpha(\mathbf x)C'_\beta(q^\alpha\mathbf x)$.  It is regular at $\boldsymbol\xi_\nu$ because $q^\alpha\boldsymbol\xi_\nu=\boldsymbol\xi_{\nu+\alpha}$.  If $\lambda\in\Lambda_N^+$ and $\lambda+\gamma\notin\Lambda_N^+$, the terms with $\lambda+\alpha\notin\Lambda_N^+$ vanish at $\boldsymbol\xi_\lambda$ by (b) for $D$, and the remaining terms vanish by (b) for $D'$ at the dominant weight $\lambda+\alpha$.  This proves the claim.  It follows that, for $g\in\mathscr P_N^{S_N}$,
\begin{equation}
  (\mathrm{ev}_\lambda D)(g)
  =(Dg)(\boldsymbol\xi_\lambda)
  =\sum_\gamma C_\gamma(\boldsymbol\xi_\lambda)\,g(\boldsymbol\xi_{\lambda+\gamma})
  =\sum_{\gamma:\ \lambda+\gamma\in\Lambda_N^+}
   C_\gamma(\boldsymbol\xi_\lambda)\,\mathrm{ev}_{\lambda+\gamma}(g),
  \label{eq:appD-right-action}
\end{equation}
which is \eqref{eq:eval-right-action}.  Since $(\mathrm{ev}_\lambda D)D'=\mathrm{ev}_\lambda(DD')$, the space $\mathcal M_{\mathrm{ev}}$ is a right $\mathcal A_N$-module.

\subsection{The recurrence and part (i)}
\label{subsec:appD-recurrence}

Let $\lambda,\mu\in\Lambda_N^+$.  Evaluating $\mathcal D_k\widetilde P_\mu=e_k(\boldsymbol\xi_\mu)\widetilde P_\mu$ at $\boldsymbol\xi_\lambda$ and using property (b) gives
\begin{equation}
  \sum_{\substack{|I|=k\\ \lambda+\varepsilon_I\in\Lambda_N^+}}
  \mathsf A_I(\boldsymbol\xi_\lambda)\,
  \widetilde P_\mu(\boldsymbol\xi_{\lambda+\varepsilon_I})
  =e_k(\boldsymbol\xi_\mu)\,\widetilde P_\mu(\boldsymbol\xi_\lambda).
  \label{eq:appD-evaluated-eigenvalue}
\end{equation}
By \eqref{eq:appD-evaluation-symmetry}, applied to every term, this says that the two sides of
\begin{equation}
  e_k(\mathbf x)\,\widetilde P_\lambda(\mathbf x)
  =\sum_{\substack{|I|=k\\ \lambda+\varepsilon_I\in\Lambda_N^+}}
  \mathsf A_I(\boldsymbol\xi_\lambda)\,
  \widetilde P_{\lambda+\varepsilon_I}(\mathbf x)
  \label{eq:appD-Pieri}
\end{equation}
agree at $\mathbf x=\boldsymbol\xi_\mu$ for every $\mu\in\Lambda_N^+$.  By Lemma~\ref{lem:appD-density} they are equal.  Equation \eqref{eq:appD-Pieri} is the Pieri formula for $e_k$ in the normalization $\widetilde P$, and this derivation from the evaluation symmetry is the standard one \cite[Chapter~VI, Section~6]{Macdonald:1995}.

The $P_\lambda$ with $\lambda\in\Lambda_N^+$ form a basis of $\mathscr P_N^{S_N}$, and by Lemma~\ref{lem:appD-density} the $\mathrm{ev}_\lambda$ form a basis of $\mathcal M_{\mathrm{ev}}$.  Hence $U_{\mathrm{ev}}$ is an isomorphism, and $U_{\mathrm{ev}}(1)=\mathrm{ev}_0$ because $\widetilde P_0=1$.  Comparing \eqref{eq:appD-Pieri} with \eqref{eq:appD-right-action} for $D=\mathcal D_k$ gives $U_{\mathrm{ev}}(e_kf)=U_{\mathrm{ev}}(f)\,\mathcal D_k$ for $f=\widetilde P_\lambda$, and hence for all $f$.  Multiplication by $e_N^{-1}$ and the right action of $\mathcal D_N^{-1}$ are the inverses of multiplication by $e_N$ and of the right action of $\mathcal D_N$, so the same relation holds for them.  The ring of symmetric Laurent polynomials is generated by $e_1,\ldots,e_N$ and $e_N^{-1}$, and each $e_k(Y)$ preserves $\mathscr P_N^{S_N}$ and restricts to $\mathcal D_k$ there by \eqref{eq:spherical-ekY-Dk}.  The operator $\mathbf e_+p(Y)\mathbf e_+$ therefore acts on $\mathscr P_N^{S_N}$ as the corresponding polynomial in $\mathcal D_1,\ldots,\mathcal D_N$ and $\mathcal D_N^{-1}$.  Since these operators commute, \eqref{eq:eval-intertwining} follows.

\subsection{Parts (ii) and (iii)}
\label{subsec:appD-parts-ii-iii}

Let $p$ and $p'$ be symmetric Laurent polynomials and write $p'^{\vee}(\mathbf x)=p'(\mathbf x^{-1})$, so that $\mathbf e_+p'(Y^{-1})\mathbf e_+=\mathbf e_+p'^{\vee}(Y)\mathbf e_+$.  By \eqref{eq:eval-intertwining} with $f=1$,
\begin{equation}
  \mathrm{ev}_0\,\bigl(\mathbf e_+p(Y)\mathbf e_+\bigr)
  \bigl(\mathbf e_+p'(Y^{-1})\mathbf e_+\bigr)
  =U_{\mathrm{ev}}\bigl(p\,p'^{\vee}\bigr).
  \label{eq:appD-two-insertions}
\end{equation}
The coefficient of $\mathrm{ev}_0$ on the right is the coefficient of $\widetilde P_0=1$ in the expansion of $p\,p'^{\vee}$ in Macdonald polynomials.  The orthogonality relation \eqref{eq:macdonald-orthogonality} holds for all dominant weights by \eqref{eq:determinant-shift-P} and \eqref{eq:appA-center-pairing-invariance}.  With $\langle1,1\rangle_{N;q,\mathfrak t}=1$ the coefficient therefore equals
\begin{equation}
  \langle p\,p'^{\vee},1\rangle_{N;q,\mathfrak t}
  =\langle p,p'\rangle_{N;q,\mathfrak t}.
  \label{eq:appD-vacuum-coefficient}
\end{equation}
This is the first statement of (iii).  For $p=p'=h_r$ it gives the last equality in \eqref{eq:eval-two-point}, and the first expression there is the coefficient of $\mathrm{ev}_0$ by \eqref{eq:appD-right-action} applied to $D=H_rH_r^\vee$.

For (ii), apply \eqref{eq:eval-intertwining} to $f=1$ and $p=h_r=\sum_vQ_v[\mathbb A_{q,\mathfrak t}]\,P_v(\boldsymbol\xi_0)\,\widetilde P_v$.  This gives
\begin{equation}
  \mathrm{ev}_0H_r
  =\sum_{\substack{v\vdash r\\ \ell(v)\leq N}}
  Q_v[\mathbb A_{q,\mathfrak t}]\,P_v(\boldsymbol\xi_0)\,\mathrm{ev}_v .
  \label{eq:appD-forward}
\end{equation}
Comparison with \eqref{eq:appD-right-action} and the linear independence of the $\mathrm{ev}_v$ give the first formula in \eqref{eq:eval-coefficients}.  The vanishing of $\mathsf C^{\mathrm{DAHA}}_{r,\gamma}(\boldsymbol\xi_0)$ for $\gamma$ not a partition is property (b).  Next, $\mathrm{ev}_vH_r^\vee=U_{\mathrm{ev}}(h_r^\vee\,\widetilde P_v)$ with $h_r^\vee(\mathbf x)=h_r(\mathbf x^{-1})$, and as in \eqref{eq:appD-vacuum-coefficient} the coefficient of $\mathrm{ev}_0$ is
\begin{equation}
  \langle h_r^\vee\,\widetilde P_v,1\rangle_{N;q,\mathfrak t}
  =\frac{\langle P_v,h_r\rangle_{N;q,\mathfrak t}}{P_v(\boldsymbol\xi_0)}
  =Q_v[\mathbb A_{q,\mathfrak t}]\,\frac{\mathcal N_v^{(N)}}{P_v(\boldsymbol\xi_0)}.
  \label{eq:appD-reverse}
\end{equation}
By \eqref{eq:appD-right-action} and \eqref{eq:dual-shift-coefficients} this coefficient is $\mathsf C^\vee_{r,v}(\boldsymbol\xi_v)$, which gives the second formula in \eqref{eq:eval-coefficients}.  Finally, \eqref{eq:appD-forward} and \eqref{eq:eval-coefficients} give $\mathrm{ev}_0H_rH_r^\vee=\sum_v\mathsf C^{\mathrm{DAHA}}_{r,v}(\boldsymbol\xi_0)\,\mathrm{ev}_vH_r^\vee$, and taking the coefficient of $\mathrm{ev}_0$ gives the two middle expressions in \eqref{eq:eval-two-point}.

\end{document}